\documentclass[11pt]{article}
\usepackage[letterpaper,margin=1in]{geometry}
\usepackage[T1]{fontenc}
\usepackage{amsmath,amsthm,mathtools}
\usepackage{newtxtext,newtxmath}
\usepackage{microtype,graphicx,booktabs,array}
\usepackage[round,authoryear]{natbib}
\usepackage{enumitem}
\usepackage{algorithm,algpseudocode}
\usepackage{caption,placeins,needspace}

\usepackage[hidelinks]{hyperref}
\usepackage[nameinlink,noabbrev]{cleveref}
\hypersetup{pdftitle={Anchored multiple testing: a simple use of e-closure to improve procedures},pdfsubject={Research draft on FDR, PRDS, Gaussian scores, and exact computation},pdfkeywords={FDR, e-closure, PRDS, Gaussian correlation, BH}}
\usepackage{fancyhdr}
\setlist{topsep=3pt,itemsep=2pt}
\newtheorem{theorem}{Theorem}[section]
\newtheorem{lemma}[theorem]{Lemma}
\newtheorem{proposition}[theorem]{Proposition}
\newtheorem{corollary}[theorem]{Corollary}
\theoremstyle{definition}

\theoremstyle{remark}
\newtheorem{remark}[theorem]{Remark}
\newcommand{\E}{\mathbb E}
\newcommand{\Prob}{\mathbb P}
\newcommand{\R}{\mathbb R}
\newcommand{\ind}{\mathbf 1}
\newcommand{\FDP}{\operatorname{FDP}}
\newcommand{\FDR}{\operatorname{FDR}}
\newcommand{\BH}{\operatorname{BH}}
\newcommand{\BY}{\operatorname{BY}}
\newcommand{\cR}{\mathcal R}
\newcommand{\phibar}{\overline\Phi}

\newcommand{\given}{\,\middle|\,}
\DeclareMathOperator{\Unif}{Unif}
\DeclareMathOperator{\Beta}{Beta}
\title{\vspace{-1.3em}Anchored multiple testing:\\a transparent use of e-closure to improve FDR procedures}
\author{Aaditya Ramdas\\ \texttt{aramdas@stanford.edu}}
\date{\today}
\begin{document}
\maketitle
\vspace{-2em}
\begin{abstract}
The recent e-closure method can recover every procedure that controls FDR (and other expectation losses). But the recovered e-collection is ``self-referential'' and gives no insight on how to improve the procedure (if improvable), and some recent improvements have been somewhat opaque.
We introduce an elementary new technique called anchoring that exploits looseness in existing FDR proofs to enlarge a baseline multiple-testing procedure's self-referential local e-value. The resulting e-closure thus transparently retains every baseline discovery (and usually adding more) and controlling the false discovery rate under the same conditions as the baseline. 
To show that this principle is broadly applicable, we use it to improve a large suite of multiple testing procedures:
(i) Anchored-BH dominates the Benjamini-Hochberg (BH) procedure under PRDS while being incomparable to Goeman's recent closed-BH, (ii) Anchored-BY dominates the Benjamini-Yekutieli (BY) procedure under arbitrary dependence while being incomparable to closed-BY, (iii) For two-sided Gaussian p-values (under appropriate covariance conditions), Anchored-2BH dominates running BH twice at half the level on two one-sided p-values,
(iv) Anchored-dBH dominates dependence-adjusted BH, (v) Anchored e-BH dominates e-BH and is incomparable to closed e-BH,
(vi) Anchored SeqStep+ improves the original (including selective and adpative variants) while preserving ordered rejection structures. All of these are accomplished in sorting or quadratic time. 
The appendix shows how to dominate Shifted-BH (for two-sided arbitrarily correlated Gaussians) and NDBH (for negatively dependent p-values), IndBH (for graph-dependent p-values), conformal multiple testing and the support line procedure for boundary FDR control. 
\end{abstract}
\noindent\textbf{Keywords:} false discovery rate; e-values; ordered testing; conditional calibration; closure.

\section{Introduction}
We test $m$ hypotheses $H_1,\ldots,H_m$ using p-values $p_1,\ldots,p_m$. Write $[m]=\{1,\ldots,m\}$, $I_0\subseteq[m]$ for the true nulls, and $m_0=|I_0|$. For a rejection set $R$ and a hypothetical null set $S$, let $\FDP_S(R)=|R\cap S|/(|R|\vee1)$; the false discovery rate is $\FDR(R)=\E\FDP_{I_0}(R)$. Our target is $\FDR\le\alpha$.

An e-value for $H_S=\bigcap_{i\in S}H_i$ is a nonnegative random variable with expectation at most one under $H_S$. The e-closure method begins by constructing an e-collection, which is a collection of local e-values $(E_S)_{S \subseteq [m]}$. E-closure then constructs the collection
\begin{equation}\label{eq:introclosure}
 \cR_\alpha(E)=\left\{R\subseteq[m]:
 \FDP_S(R)\le\alpha E_S\text{ for every nonempty }S\subseteq[m]\right\}.
\end{equation}
One can then pick any  element $R$ of $\cR_\alpha(E)$ and call that the rejected set. 
any such choice controls FDR, since by definition of $ \cR_\alpha(E)$, the FDP of $R$ must be bounded by $\alpha E_{I_0}$ \citep{XuClosure}; taking expectations yields the claim.

E-closure can recover \emph{any} level-$\alpha$ FDR-controlling procedure; denoting the procedure's rejection set by $A$, one can choose local e-values:
\begin{equation}\label{eq:canonical}
 E_S^A=\frac{\FDP_S(A)}\alpha
      =\frac{|A\cap S|}{\alpha(|A|\vee1)}.
\end{equation}
These are valid local e-values because $S\subseteq I_0$ implies $\FDP_S(A)\le\FDP_{I_0}(A)$. Their closure is exactly $\{\varnothing,A\}$: singleton constraints exclude indices outside $A$, and the constraint $S=R$ excludes every nonempty proper subset of $A$.

We call these local e-values as ``self-referential'' because they use the rejected set $A$ in order to define the local e-values such that e-closure eventually reproduces the set $A$.
Recovering a procedure is therefore automatic; improving it is not. 

One direct route to improve procedures is to find \emph{pointwise larger, still-valid local e-values}, with strict increases that unlock new rejection sets. Such local-evidence improvements yield the closed Benjamini-Yekutieli (BY) \citep{XuClosure,SunWang2026} method that outperforms the original.

The Benjamini-Hochberg (BH) case is subtler: canonical closed BH uses a horizon matrix and coupled adaptive thresholds \citep{Goeman2026}. It is a pointwise uniform improvement of BH, but the construction is highly nontrivial and by no means elementary.

Our aim in this paper is to provide a very simple and very broadly applicable technique that can deliver improvements to some of the best known FDR procedures over the past decade or two. We first write it down for BH before we isolate the key ingredients that allow us to apply it in many settings.

\paragraph{Anchored-BH.}
Let $B=\BH_\alpha(p)$ for superuniform p-values that are positively regression dependent on the true-null coordinates (PRDS). BH controls FDR at $\alpha m_0/m$ \citep{BH1995,BY2001}. First, for each $S \subseteq [m]$, choose any local e-values $F_S$ for $H_S$ (this is a composite null hypothesis testing problem; not a multiple testing problem and thus can be done ignoring $H_{[m]\backslash S}$). Then define the local e-value as a particular average of the self-referential e-value and $F_S$:
\begin{equation}\label{eq:introanchor}
 E_S=E_S^B+\left(1-\frac{|S|}{m}\right)F_S
 =\frac{|B\cap S|}{\alpha(|B|\vee1)}+
       \left(1-\frac{|S|}{m}\right)F_S,
\end{equation}
where the second term is omitted for $S=[m]$. It is not hard to verify that this is indeed a valid local e-value for every $S$. When these are fed into the e-closure rule \eqref{eq:introclosure}, we note that $B$ itself lies in $ \cR_\alpha(E)$, because for every $S$,
\[
 \FDP_S(B)=\alpha E_S^B\le\alpha E_S.
\]
Thus $B$ is always in $ \cR_\alpha(E)$. Define $\bar B$ as a measurable largest-cardinality certified \emph{superset} of $B$, with any fixed tie rule. This procedure is \emph{Anchored-BH}; the choice of $(F_S)$ specifies a member of the family.

\begin{theorem}[Anchored-BH]\label[theorem]{thm:anchor}
Suppose the null p-values are superuniform and the entire vector is PRDS on the true-null coordinates. For any family of local e-values $F_S$ for $H_S$, the values in \eqref{eq:introanchor} are local e-values and
\[
 \bar B\supseteq B\quad\text{for every data vector},\qquad
 \FDR(\bar B)\le\alpha.
\]
No independence between $B$ and $(F_S)$ is required.
\end{theorem}
The expectation calculation behind the theorem is
\[
 \E E_{I_0}\le\frac{m_0}{m}+
       \left(1-\frac{m_0}{m}\right)\E F_{I_0}\le1.
\]
Section~\ref{sec:anchor} gives the complete proof and the simultaneous guarantee. In particular, both $B$ and $\bar B$ remain in one certified collection: the researcher may choose to retain BH after seeing the augmentation. This is a distinction from the published closed-BH e-collection, which need not simultaneously certify the original BH set \citep[Section~6]{Goeman2026}. It is also a stronger requirement than pointwise containment alone.

Based on the preceding display, the gain from anchoring comes from the positive term $\left(1-\frac{m_0}{m}\right)\E F_{I_0}$ that helps to exhaust the FDR budget (when the first term leaves some slack). Thus, we may interpret anchoring as a generic form of \emph{adaptivity to the unknown null proportion}, but it is a technique that works even under dependence, as is evident above. It is also clear that if $m_0 \approx m$, the second additive term is negligible, and we should not expect to see any gains (as is the case with all adaptive procedures); we call this the sparse regimes. When there are a larger fraction of non-nulls (dense regime), the gains of anchoring are more evident.

\paragraph{A generic $F_S$ based on a larger pilot level.}
The auxiliary evidence $F_S$ need not come from a new test statistic. Run the baseline method again at a fixed level $\beta>\alpha$, such as $\beta=2\alpha<1$. For BH, let $C=\BH_\beta(p)$. Its coordinatewise contribution bound gives the individual e-values and their local averages
\[
 x_i^{(\beta)}=\frac{m\ind\{i\in C\}}{\beta(|C|\vee1)},
 \qquad
 F_S^{(\beta)}=\frac1{|S|}\sum_{i\in S}x_i^{(\beta)}
             =\frac{m|C\cap S|}{\beta|S|(|C|\vee1)}.
\]
Substituting $F_S^{(\beta)}$ into \eqref{eq:introanchor} uses the liberal run to support additions while retaining the original level-$\alpha$ guarantee. The two runs use the same data; no sample split or independence between the runs is needed. This single-pilot anchor already gives strict improvements in examples and simulations, and its complete e-closure computation takes $O(m\log m)$ time. Fixed convex mixtures of pilot levels provide a more flexible alternative. Section~\ref{sec:pilots} develops these choices.

The principle is generic, but its normalization must follow the available guarantee. If a procedure $A_\beta$ has known coordinatewise bounds
$\E[\ind\{i\in A_\beta\}/(|A_\beta|\vee1)]\le v_i(\beta)$ under $H_i$, then dividing that contribution by a positive $v_i(\beta)$ gives an individual e-value. Section~\ref{sec:coordinate} isolates this reusable property, which appears before the final sum over true nulls in many FDR proofs \citep{BY2001,BR2008,FithianLei2022,SarkarZhang2025}. 

The paper derives several other e-values $F_S$ under various assumptions like PRDS or negatively dependent p-values, Gaussian data with unknown arbitrary covariance, and so on, and these need to be designed and used well in order to actual deliver a nontrivial improvement over the baseline procedure.


\paragraph{Summary of main contributions.}
We continue our treatment of BH in Section~\ref{sec:anchor}; Sections~\ref{sec:prds} and \ref{sec:gaussian} provide further score designs for PRDS p-values and Gaussian statistics. \emph{Anchored-BH} dominates\footnote{
Throughout, ``dominates'' means pointwise containment with the same level-$\alpha$ guarantee, and a strict addition on some inputs. Since the new procedure is \emph{anchored} to the original one, we call the broader technique as \emph{anchored e-closure} or \emph{anchored multiple testing}, and name the new procedure by appending \emph{Anchored-} to the original. 
} BH under PRDS and is incomparable in general with canonical closed BH. \emph{Anchored-2BH} (Section~\ref{sec:twobh}) studies the case of two-sided Gaussian p-values, and it dominates the union of BH at level $\alpha/2$ in each Gaussian tail, under the appropriate covariance restrictions validating those runs. 
\emph{Anchored-BY} (Section~\ref{sec:by}) dominates BY under arbitrary dependence and is incomparable in general with closed BY. \emph{Anchored e-BH} (Section~\ref{sec:anchored-ebh}) dominates e-BH and is incomparable to closed e-BH. \emph{Anchored-dBH} (Section~\ref{sec:dbh}) improves the conditional-calibration procedure of \citet{FithianLei2022}.

\emph{Anchored ordered testing} (Section~\ref{sec:ordered}) exploits a different looseness. SeqStep+, Selective SeqStep+, and Adaptive SeqStep need not have BH's $m_0/m$ bound. We retain each realized FDP contribution and replace its unused conditional budget by a conditional e-value. 

Appendices~\ref{app:negative} and \ref{app:shifted} derive anchored procedures that dominate the BH procedure under negative-dependence~\cite{CRW2024}, and dominate shifted-BH that works on the two-sided Gaussian p-values under arbitrary covariances, with their algorithms and experiments \citep{SarkarZhang2025,GhoshSarkar2025,GhoshSarkar2026}.  Appendix~\ref{app:anchored-indbh} dominates ind-BH~\citep{NguyenFithian2025}, and Appendix~\ref{app:anchored-sl} dominates the support line procedure for boundary FDR control~\citep{XiangSoloffFithian2025} and Appendix~\ref{app:anchored-conformal} dominates existing conformal FDR procedures.

Section~\ref{sec:general-losses} records a coordinatewise extension to general expected losses; the procedures and experiments studied here concern FDR only. The mechanisms behind the above improvements also directly extend to a plethora of procedures not covered in this paper, like the DAGGER procedure (a directed acyclic graph of hypotheses), p-filter (multilayer FDR control), and QuTE (decentralized multiple testing on networks via quantization), but finding computational shortcuts to anchored closure in these structured settings has not been explored, so these have been omitted and we refer to the original papers for more details~\citep{pfilter2017,pfilter2019,qute2017,dagger2019}.


Appendix~\ref{sec:computation} gives quadratic algorithms for ordered BH/BY baselines, the $O(m\log m)$ single-pilot shortcut, and a near-quadratic extension to arbitrary baseline masks. Ordered procedures have separate structural scans in Section~\ref{sec:ordered}. The simulations in Section~\ref{sec:simulations} follow the five main families in order: averaged-score BH anchors help mainly in dense regimes, whereas ordered conditional anchors can help with few signals. Appendix~\ref{app:sparse} explains the former limitation. Related building blocks include e-value calibration, e-BH, conditional calibration and donation \citep{VovkWang2021,WangRamdas2022,FithianLei2022,LeeRen2024,XuFischerRamdas2026}.

\paragraph{Caveat.} The pointwise dominance of our anchored methods, along with instances showing strict dominance, prove that all the original methods are strictly inadmissible. (BH, BY and e-BH were already known to be inadmissible, using e-closure, with the BH improvement being highly nontrivial and arguably opaque.) This does not mean that our anchored procedures are admissible, or good/optimal in any other sense, but it does mean that they can safely replace the earlier methods whenever those were employed. Our aim here is simplicity and generality: we  anticipate that studying particular settings in more detail may help develop even ``better'' improvements via e-closure, and those may be preferable in practice, but we anticipate that our method will remain a useful pedagogical tool and an easy first pass as at improving other procedures.

\paragraph{AI disclosure.} The main anchoring idea was developed by the author without AI. The author also immediately recognized the key ingredients that made the technique generally applicable, and designed the various e-value families $(F_S)$ found in the paper. After that, AI helped flesh out some details in each of the examples, and also ran all the experiments. The author used AI iteratively to write and revise the paper, and then manually edited for clarity. The author takes responsibility for all content.

\section{Anchored-BH under PRDS}\label{sec:anchor}
\subsection{Notation and the BH guarantee}
Write $[m]=\{1,\ldots,m\}$, $I_0\subseteq[m]$ for the true-null indices, and $m_0=|I_0|$. For a rejection set $R$ and a hypothetical null set $S$, define
\[
 \FDP_S(R)=\frac{|R\cap S|}{|R|\vee1},\qquad
 \FDR(R)=\E\FDP_{I_0}(R).
\]
An intersection e-value $F_S$ is nonnegative and has expectation at most one whenever all hypotheses in $S$ are true, within the stated model class. No independence between different $F_S$ is required.

Throughout this section, the p-values are superuniform under their nulls and PRDS on the true-null coordinates; Section~\ref{sec:prds} states the dependence condition precisely. Let $p_{(1)}\le\cdots\le p_{(m)}$ and
\[
 b=\max\left(\{k:p_{(k)}\le\alpha k/m\}\cup\{0\}\right),
 \qquad B=\BH_\alpha(p)=\{i:p_i\le\alpha b/m\},
\]
with $B=\varnothing$ when $b=0$. The usual BH guarantee is
\begin{equation}\label{eq:bhbudget}
 \FDR(B)\le\alpha\frac{m_0}{m}.
\end{equation}
The same BH proof also bounds the contribution of any subset $S$ of true nulls by $\E\FDP_S(B)\le\alpha|S|/m$ \citep{BY2001,BR2008}.

\subsection{Proof and simultaneous retention}
For later use, write the collection in \eqref{eq:introanchor} as
\begin{equation}\label{eq:anchor}
 \begin{aligned}
 E_S&=\frac{|B\cap S|}{\alpha(b\vee1)}+
       \left(1-\frac{|S|}{m}\right)F_S
          &&(\varnothing\ne S\subsetneq[m]),\\
 E_{[m]}&=\frac{\ind\{b>0\}}\alpha.
 \end{aligned}
\end{equation}
The notation
\begin{equation}\label{eq:cert}
 \cR_\alpha=\cR_\alpha(E),\qquad
 \bar B\in\mathop{\arg\max}_{R\in\cR_\alpha:\,R\supseteq B}|R|
\end{equation}
fixes the Anchored-BH output, with ties resolved measurably. It does not define $\bar B$ as the union of all certified sets: such a union need not itself be certified.

\begin{proof}[Proof of \cref{thm:anchor}]
For any true intersection $S$, the BH contribution bound following \eqref{eq:bhbudget} and linearity give
\[
 \E E_S\le\frac{|S|}{m}+
       \left(1-\frac{|S|}{m}\right)\E F_S\le1.
\]
For $S=[m]$ this is BH's global-null guarantee. For the actual null set $I_0\ne\varnothing$, every certified $R$ satisfies $\FDP_{I_0}(R)\le\alpha E_{I_0}$, so in fact
\begin{equation}\label{eq:simulfdr}
 \E\left[\max_{R\in\cR_\alpha}\FDP_{I_0}(R)\right]\le\alpha.
\end{equation}
When $I_0=\varnothing$ the FDP is zero. Pointwise,
$\FDP_S(B)=\alpha E_S^B\le\alpha E_S$ for every $S$, so $B\in\cR_\alpha$ and the superset search in \eqref{eq:cert} is nonempty. Thus $\bar B\supseteq B$ and $\FDR(\bar B)\le\alpha$.
\end{proof}

No independence between $B$ and $F_S$ is used. Both may depend on the full data, provided the auxiliary e-value bound is justified. The theorem permits post hoc choice between $B$ and $\bar B$, because both belong to the same collection in \eqref{eq:simulfdr}. It does not say that every arbitrary set lying between them is certified.

The full-intersection constraint reveals an important limitation. If $B=\varnothing$, then $E_{[m]}=0$ and no nonempty rejection set is certified. Thus Anchored-BH has exactly BH's event of making any rejection, and exactly BH's global-null FDR. It is an augmentation, not a new detector of configurations missed entirely by BH. Appendix~\ref{sec:sparse} gives a stronger limitation for a broad class of averaged scores.

\subsection{A larger-level pilot and adaptive staircases}\label{sec:pilots}
Fix $\beta\in(\alpha,1)$, let $C=\BH_\beta(p)$, and write $c=|C|$. A default is $\beta=2\alpha$ when $\alpha<1/2$. The coordinatewise BH bound in Section~\ref{sec:coordinate}, applied at level $\beta$, implies that
\begin{equation}\label{eq:pilotscore}
 x_i^{(\beta)}=\frac{m\ind\{i\in C\}}{\beta(c\vee1)}
\end{equation}
obeys $\E x_i^{(\beta)}\le1$ under $H_i$. Hence, for $n=|S|$,
\begin{equation}\label{eq:pilotlocal}
 F_S^{(\beta)}=\frac1n\sum_{i\in S}x_i^{(\beta)}
              =\frac{m|C\cap S|}{\beta(c\vee1)n}
\end{equation}
is a local e-value under $H_S$, simply by linearity. The resulting anchor is
\begin{equation}\label{eq:pilotanchor}
 E_S=\frac{|B\cap S|}{\alpha(b\vee1)}+
 \left(1-\frac n m\right)
 \frac{m|C\cap S|}{\beta(c\vee1)n},\qquad B=\BH_\alpha(p).
\end{equation}
The second term is omitted for $S=[m]$. The pilot supplies evidence, not a final level-$\alpha$ rejection set: the output is the largest superset of $B$ certified by \eqref{eq:pilotanchor}, not automatically $C$. Every added hypothesis must belong to $C$, since a singleton outside $C$ has zero evidence. Appendix~\ref{sec:pilotfast} reduces the full closure to two scalar inequalities and gives an $O(m\log m)$ implementation, including both BH runs.

A fixed mixture of pilot levels gives an adaptive staircase:
\begin{equation}\label{eq:adaptivestair}
 F_S=\frac1n\sum_{i\in S}\sum_{\ell=1}^L
                w_{n,\ell}x_i^{(\beta_\ell)},
 \qquad w_{n,\ell}\ge0,\quad \sum_{\ell=1}^Lw_{n,\ell}=1.
\end{equation}
The levels and weights may depend on $n$, but must be prespecified or independently trained. Each mixture remains a local e-value even though all pilot runs use the same testing data. Selecting the most favorable observed pilot or mixture without further calibration is not covered. A larger pilot level offers more candidate discoveries but lowers their normalized scores, so increasing $\beta$ need not increase the anchor's rejection set.

These constructions already produce useful gains. In the independent, one-sided Gaussian experiment with $m=100$, $\alpha=0.05$, $90\%$ nonnulls and mean $3$, the single $\beta=0.10$ pilot raises power from $89.53\%$ for BH to $94.88\%$; it adds no power in the corresponding $10\%$ and $50\%$ nonnull settings (Section~\ref{sec:sim-bh}). Thus a liberal rerun is a simple reusable score, not a promise of improvement in every configuration. Appendix~\ref{app:incomparability} gives pointwise comparisons with canonical closed BH. Section~\ref{sec:prds} develops fixed-score and adaptive-hinge alternatives.

\subsection{The coordinatewise inequality behind the improvement}\label{sec:coordinate}
The useful property of BH is more precise than its final FDR guarantee. For every fixed $\gamma\in(0,1)$, its proof under PRDS gives
\begin{equation}\label{eq:bhcoordinate}
 \E\left[\frac{\ind\{i\in\BH_\gamma(p)\}}
                   {|\BH_\gamma(p)|\vee1}\right]
       \le\frac\gamma m,\qquad i\in I_0.
\end{equation}
This is the coordinatewise dependency-control bound of \citet{BY2001,BR2008}. Equivalently, use the threshold $T_\gamma=\gamma|\BH_\gamma(p)|/m$ in \eqref{eq:adaptiveindicator}; its ratio equals $m/\gamma$ times the contribution in \eqref{eq:bhcoordinate}, with both zero when the BH set is empty.

The same inequality does two jobs. Summing it over $i\in S$ at level $\alpha$ leaves the unused local budget $1-|S|/m$ in \eqref{eq:anchor}. Normalizing each contribution at level $\beta$ gives the individual pilot e-values in \eqref{eq:pilotscore}. This is why a proof that bounds each null contribution separately can be reused without deriving a new test statistic.

More generally, suppose a family $A_\gamma$ has deterministic positive bounds $v_i(\gamma)$ satisfying
\begin{equation}\label{eq:coordinatefamily}
 \E\left[\frac{\ind\{i\in A_\gamma\}}{|A_\gamma|\vee1}\right]
       \le v_i(\gamma)\quad(i\in I_0),
 \qquad \sum_{i=1}^m v_i(\gamma)\le\gamma.
\end{equation}
For a fixed pilot level $\beta>\alpha$, the scores
\begin{equation}\label{eq:genericpilotscore}
 \widetilde x_i^{(\beta)}=
 \frac{\ind\{i\in A_\beta\}}
      {v_i(\beta)(|A_\beta|\vee1)}
\end{equation}
are individual e-values; their averages over $S$ and prespecified convex mixtures supply auxiliary local e-values. Meanwhile, the baseline $A_\alpha$ admits the weighted anchor
\begin{equation}\label{eq:coordinateanchor}
 E_S=\frac{\FDP_S(A_\alpha)}\alpha+
       \left(1-\frac1\alpha\sum_{i\in S}v_i(\alpha)\right)F_S.
\end{equation}
Its expectation is at most one by \eqref{eq:coordinatefamily}, and its first term retains the baseline certificate pointwise. Neither independence between the runs nor nesting of $A_\alpha$ and $A_\beta$ is needed for this argument. The nesting of BH runs is used only for the especially simple pilot shortcut.

Such bounds are available for BY under arbitrary dependence (Section~\ref{sec:by}), for dependence-adjusted BH after conditional calibration (Section~\ref{sec:dbh}), and with unequal $v_i$ for shifted BH (Appendix~\ref{app:shifted}). The relevant contribution inequalities are extracted from their original proofs, rather than inferred from the final FDR statement. A total guarantee $\FDR(A_\gamma)\le\gamma$ alone gives the weaker local e-value $\FDP_S(A_\gamma)/\gamma$, but does not justify the individual normalization \eqref{eq:genericpilotscore} or the coefficient $1-|S|/m$. Ordered testing therefore receives a separate conditional-budget treatment in Section~\ref{sec:ordered}.

\subsection{Other losses under coordinatewise bounds}\label{sec:general-losses}
The anchoring calculation is not specific to FDP. 
The following specialization identifies when the same coordinatewise budgets can enlarge a baseline for another error criterion.

\begin{proposition}[Coordinatewise anchoring for an expected loss]\label[proposition]{prop:general-loss}
Let $L(R,S)\ge0$, with $L(R,\varnothing)=0$, be a fixed loss. Suppose nonnegative contributions $d_i(A)$ and deterministic weights $w_i\ge0$ satisfy
\begin{equation}\label{eq:loss-coordinate}
 L(A,S)\le\sum_{i\in S}d_i(A)\quad\text{for every }S,
 \qquad \E_{H_i}d_i(A)\le\alpha w_i,
 \qquad\sum_{i=1}^m w_i\le1.
\end{equation}
For any intersection e-values $F_S$, set
\begin{equation}\label{eq:loss-anchor}
 E_S^L=\frac{L(A,S)}\alpha+
 \left(1-\sum_{i\in S}w_i\right)F_S.
\end{equation}
Then these are valid local e-values, $A$ belongs to
\[
 \cR_\alpha^L(E)=\{R:L(R,S)\le\alpha E_S^L\text{ for every nonempty }S\},
\]
and every measurable certified superset $\bar A\supseteq A$ satisfies
$\E L(\bar A,I_0)\le\alpha$. In fact the loss is controlled simultaneously over the certified collection.
\end{proposition}
\begin{proof}
Under $H_S$, the expectation of the first term in \eqref{eq:loss-anchor} is at most $\sum_{i\in S}w_i$ and that of the second is at most its complementary weight. Pointwise, $L(A,S)\le\alpha E_S^L$, so $A$ is certified. Finally,
$\max_{R\in\cR_\alpha^L(E)}L(R,I_0)\le\alpha E_{I_0}^L$.
\end{proof}

We record the general statement for context, but study only FDR-controlling procedures and FDR experiments in this paper and leave other losses for future work.

\section{Scores for PRDS p-values}\label{sec:prds}
A null p-value is superuniform if $\Prob(p_i\le t)\le t$. We use the lower-tail formulation of PRDS: for every increasing measurable set $D$ and true-null $i$, the function $t\mapsto\Prob(p\in D\mid p_i\le t)$ is nondecreasing wherever defined. The usual point-conditioned PRDS assumption implies this formulation. All statements involving an adaptive threshold concern PRDS of the entire vector on its true-null coordinates, not only dependence among the nulls.

\subsection{A capped-sum comparison without convexity}
\begin{lemma}[PRDS capped sums]\label[lemma]{lem:prds}
Let $S$ index true nulls in a superuniform PRDS vector. For fixed nonnegative nonincreasing functions $g_i$ and independent $U_i\sim\Unif(0,1)$,
\begin{equation}\label{eq:prdscap}
 \E\min\left\{\sum_{i\in S}g_i(p_i),C\right\}
 \le
 \E\min\left\{\sum_{i\in S}g_i(U_i),C\right\},\qquad C>0.
\end{equation}
Independence with uniform margins attains the upper bound. Integrability of the uncapped scores is not needed for this bounded comparison.
\end{lemma}
\begin{proof}
Write $X=g_i(p_i)$ and $Y=\sum_{j\in S\setminus\{i\}}g_j(p_j)$. The event $\{Y\le v\}$ is increasing in the remaining p-values. PRDS therefore implies positive quadrant dependence:
\begin{equation}\label{eq:pqd}
 \Prob(X>u,Y>v)\ge\Prob(X>u)\Prob(Y>v).
\end{equation}
To obtain it, apply PRDS to the lower-tail event $\{X>u\}$ for $p_i$, using limits for open endpoints, and compare to the unconditional probability.
For $x,y\ge0$,
\begin{equation}\label{eq:capidentity}
 \min(x+y,C)=\min(x,C)+\min(y,C)
 -\int_0^C\ind\{x>u\}\ind\{y>C-u\}\,du.
\end{equation}
Thus replacing $X,Y$ by independent variables with the same margins can only increase the expected cap. Induction decouples the remaining sum: conditional on an independent $X$, use the residual cap $(C-X)_+$. Finally, independent superuniform p-values may be coupled above independent uniforms, which only increases the decreasing scores on the right. Truncation and monotone convergence cover infinite score values.
\end{proof}
The independence comparison is particular to the capped additive functional; it does not assert that independence is least favorable for every PRDS testing procedure. The calibration in Section~\ref{sec:cap} can therefore use arbitrary decreasing scores, not only Gaussian-derived or convex scores.

\subsection{Fixed hinges and staircases}
A basic normalized indicator score is
\[
 s_t(p)=\frac{\ind\{p\le t\}}t,\qquad 0<t\le1.
\]
For fixed thresholds and simplex weights, the staircase
\begin{equation}\label{eq:staircase}
 g_n(p)=\sum_{k=1}^K w_{n,k}s_{t_{n,k}}(p),\qquad
 w_{n,k}\ge0,\quad \sum_kw_{n,k}=1,
\end{equation}
has uniform mean one. Mixtures of such indicators provide a general basis for decreasing p-to-e calibrators \citep{VovkWang2021}. Their sum distribution under independent uniforms is discrete, making cap calibration accessible by convolution.

A continuous alternative is the p-value hinge
\begin{equation}\label{eq:phinge}
 g_t(p)=\frac{2(t-p)_+}{t^2}
       =\int_0^1 2u\,s_{ut}(p)\,du.
\end{equation}
The integral representation proves unit uniform mean. Fixed mixtures may depend on $n=|S|$. They must be selected independently of the testing data, for example using a prespecified training distribution; unrestricted same-data weight selection is not covered by \cref{lem:prds}.

A useful bridge between the two model classes is the normal-quantile transform. If $f$ is any nonnegative even function with $\E f(G)=1$ for $G\sim N(0,1)$, then
\begin{equation}\label{eq:quantile}
 g(p)=f\{\Phi^{-1}(1-p/2)\}
\end{equation}
has uniform mean one. If $f$ increases in $|z|$, then $g$ decreases in $p$. Using \eqref{eq:quantile} as a calibrator does not assume Gaussian data. The Gaussian hinges of Section~\ref{sec:gaussian} can therefore also be calibrated under PRDS.

\paragraph{Calibrated local evidence.}\label{sec:cap}
For $n=|S|<m$, a scaled score average can be capped at
\begin{equation}\label{eq:cap}
 C_n=\frac1{\alpha(1-n/m)},\qquad
 F_S=\min\left\{\frac{\eta_n}{n}\sum_{i\in S}f_n(X_i),C_n\right\}.
\end{equation}
Here $f_n(X_i)=g_n(p_i)$; Section~\ref{sec:gaussian} uses $f_n(Z_i)$ instead. Choose a finite $\eta_n\ge1$ such that
\begin{equation}\label{eq:calibration}
 \E_\star\min\{\eta_n\overline f_n,C_n\}\le1,
\end{equation}
where $\overline f_n$ is the mean of $n$ scores under the appropriate independent reference: uniform p-values here, and standard normals in Section~\ref{sec:gaussian}. The capped-sum lemma then verifies local validity. The unboosted choice $\eta_n=1$ is valid for unit-mean scores; equality in \eqref{eq:calibration} is neither required nor always attainable. Capping loses no certifications: once the cap is reached, $\alpha(1-n/m)F_S=1$ and the intersection constraint is automatic. Thus a valid factor $\eta_n\ge1$ weakly enlarges the unboosted anchor. Appendix~\ref{app:calibration} records the numerical construction. This calibration concerns fixed scores, not automatically the data-adaptive pilots.

\subsection{Adaptive thresholds and hinges}\label{sec:adaptivehinges}
The PRDS dependency-control inequality of \citet{BR2008} states that, for a coordinatewise nonincreasing threshold $T(p)$,
\begin{equation}\label{eq:adaptiveindicator}
 \E\left[\frac{\ind\{p_i\le T(p)\}}{T(p)}\right]\le1,
 \qquad i\in I_0,
\end{equation}
with ratio zero when $T=0$. Lowering any p-value may increase, but may not decrease, the threshold. This permits genuinely data-adaptive evidence. For $T_\beta=\beta|\BH_\beta(p)|/m$, the ratio is precisely the pilot score in \eqref{eq:pilotscore}; its mixtures were developed in Section~\ref{sec:pilots}.

The same threshold yields a continuous adaptive hinge:
\begin{equation}\label{eq:adaptivehinge}
 h_{i,T}(p)=\frac{2(T(p)-p_i)_+}{T(p)^2}
     =\int_0^1 2u\,\frac{\ind\{p_i\le uT(p)\}}{uT(p)}\,du.
\end{equation}
Each $uT$ is a valid threshold, so Tonelli's theorem and \eqref{eq:adaptiveindicator} give $\E h_{i,T}\le1$. Averages over $S$ therefore provide local e-values, as do prespecified convex mixtures. No Gaussianity, covariance estimate or sample split is needed for these specific adaptive scores. Their validity does not justify arbitrary data-selected mixture weights. Nor does \cref{lem:prds}, which concerns fixed scores, automatically permit independent-null cap recalibration of these data-adaptive functions.

\section{Scores under arbitrary Gaussian dependence}\label{sec:gaussian}
Suppose $Z\sim N_m(\mu,\Sigma)$ with $\Sigma_{ii}=1$, and test $H_i:\mu_i=0$. The matrix $\Sigma$ is unknown and otherwise unrestricted, including singular matrices. We use two-sided p-values $p_i=2\Phi(-|Z_i|)$.

\subsection{A Gaussian capped-sum comparison}
\begin{lemma}[Gaussian capped sums]\label[lemma]{lem:gaussian}
If $Z_S$ is centered Gaussian with unit marginal variances and $f_i:\R\to[0,\infty)$ are convex and even, then, for independent $G_i\sim N(0,1)$,
\begin{equation}\label{eq:gaussiancap}
 \E_\Sigma\min\left\{\sum_{i\in S}f_i(Z_i),C\right\}
 \le
 \E_I\min\left\{\sum_{i\in S}f_i(G_i),C\right\},\qquad C>0.
\end{equation}
The upper bound is attained at independence.
\end{lemma}
\begin{proof}
Set $X=f_i(Z_i)$ and $Y=\sum_{j\ne i}f_j(Z_j)$. Both $\{X\le u\}$ and $\{Y\le v\}$ are centrally symmetric convex sets in Gaussian space. The Gaussian correlation inequality \citep{Royen2014,LatalaMatlak2015} gives positive quadrant dependence \eqref{eq:pqd}, equivalently by complementing the two lower events. Apply \eqref{eq:capidentity}, decouple $X,Y$, and induct exactly as in \cref{lem:prds}. Gaussian marginals already match the independent reference. Singular covariances are covered directly by the Gaussian correlation inequality, or by a limiting argument.
\end{proof}
The nonnull means outside $S$ play no role: under $H_S$, the marginal $Z_S$ is centered. Positive signs of the correlations are not needed. Convexity and evenness are essential to this proof and should not be removed by analogy with the PRDS result.

The lemma constructs auxiliary evidence $F_S$ under arbitrary covariance; it does not validate an uncorrected BH baseline by itself. Section~\ref{sec:gaussiantwosided} identifies covariance restrictions that permit BH-based improvements, while Section~\ref{sec:by} uses a BY baseline when no such restriction is imposed.

\subsection{Normalized hinges and a complete score basis}
For $t\ge0$, put
\begin{equation}\label{eq:gaussianhinge}
 h_t(z)=\frac{(|z|-t)_+}{2\{\phi(t)-t\phibar(t)\}}.
\end{equation}
Since $\E(|G|-t)_+=2\{\phi(t)-t\phibar(t)\}$, each $h_t$ has unit Gaussian mean. It is convex, even and increasing in $|z|$. We use
\begin{equation}\label{eq:hingemixture}
 f_n(z)=\sum_{k=1}^K w_{n,k}h_{t_{n,k}}(z),
 \qquad w_{n,k}\ge0,\quad\sum_kw_{n,k}=1.
\end{equation}
The capped means in \eqref{eq:cap} may then be calibrated by \cref{lem:gaussian}. Without boosting, the local mean is valid using only its marginal expectation, even beyond joint Gaussianity when the null margins remain standard normal.

\begin{proposition}[Hinge representation]\label[proposition]{prop:hingebasis}
Every nonnegative even convex $f$ with $\E f(G)=1$ admits
\begin{equation}\label{eq:hingebasis}
 f(z)=w_0+\int_{[0,\infty)}h_t(z)\,w(dt),\qquad
 w_0\ge0,\quad w\ge0,\quad w_0+\int w(dt)=1.
\end{equation}
\end{proposition}
\begin{proof}
On $[0,\infty)$ the right derivative of $f$ is nonnegative and nondecreasing. Integrating it represents $f(z)=f(0)+\int(|z|-t)_+\,\nu(dt)$; an atom at zero represents a nonzero initial slope. Tonelli gives $1=f(0)+\int 2\{\phi(t)-t\phibar(t)\}\nu(dt)$. Take $w_0=f(0)$ and $w(dt)=2\{\phi(t)-t\phibar(t)\}\nu(dt)$.
\end{proof}
Thus hinge mixtures, together with the constant score one, span the entire normalized convex-even class, rather than only a particular likelihood-ratio subfamily. For comparison, the symmetric Gaussian likelihood ratio $e^{-\lambda^2/2}\cosh(\lambda z)$ is another member of this class.

The geometry also provides a threshold-design heuristic. To make a normalized hinge cross a prescribed score level $T>1$ as early as possible, minimize $t+2T\{\phi(t)-t\phibar(t)\}$. Differentiating gives $1-2T\phibar(t)$, so the minimizer and crossing point are
\begin{equation}\label{eq:hingeoptimal}
 t_T=\phibar^{-1}\{1/(2T)\},\qquad
 z_T=t_T+2T\{\phi(t_T)-t_T\phibar(t_T)\}
     =\frac{\phi(t_T)}{\phibar(t_T)}.
\end{equation}
This is a one-score threshold calculation, not an optimality theorem for the whole closure. Different intersections can require different evidence levels; a high-threshold hinge may preclude moderate individual discoveries. Externally fitted mixtures can trade off those constraints.

\section{Anchored-2BH for Gaussian two-sided testing}\label{sec:twobh}\label{sec:gaussiantwosided}
The Gaussian scores from Section~\ref{sec:gaussian} improve an uncorrected baseline when its covariance structure supplies the required FDR bound. Our main construction here is \emph{Anchored-2BH}: an augmentation of the union of two one-sided BH runs for testing the two-sided nulls $H_i:\mu_i=0$. We distinguish this baseline from ordinary pooled BH on $p_i=2\phibar(|Z_i|)$, which requires a different condition.

Suppose $Z\sim N_m(\mu,\Sigma)$ with $\Sigma_{ii}=1$ and
\begin{equation}\label{eq:tailcov}
 \Sigma_{ij}\ge0\qquad(i\in I_0,\ j\in[m]).
\end{equation}
Entrywise nonnegative $\Sigma$ ensures this for every mean vector. The condition makes each one-sided p-value family PRDS on $I_0$, even when the nonnull means have mixed signs \citep[Section~3.1, Case~1]{BY2001}. Define the \emph{2BH} baseline by
\begin{equation}\label{eq:twotailbaseline}
 \begin{aligned}
 B^+&=\BH_{\alpha/2}\{\phibar(Z_i):i\in[m]\},\qquad
 B^-=\BH_{\alpha/2}\{\Phi(Z_i):i\in[m]\},\\
 B^{\pm}&=B^+\cup B^-.
 \end{aligned}
\end{equation}
Both runs use all $m$ p-values. For $\alpha<1$, the rejection sets are disjoint, since complementary one-sided p-values cannot both be at most $\alpha/2$. 

\subsection{The Anchored-2BH guarantee}
Use any fixed normalized convex-even scores $f_n$, including the hinge mixtures in \eqref{eq:hingemixture}, and choose $\eta_n$ by independent-Gaussian calibration as in \eqref{eq:calibration}. For $n=|S|<m$, set
\begin{equation}\label{eq:twotailanchor}
 E_S^{\pm}=
 \frac{|B^{\pm}\cap S|}{\alpha(|B^{\pm}|\vee1)}
 +\left(1-\frac nm\right)
 \min\left\{\frac{\eta_n}{n}\sum_{i\in S}f_n(Z_i),
                    \frac1{\alpha(1-n/m)}\right\}.
\end{equation}
Set $E_{[m]}^{\pm}=\ind\{|B^{\pm}|>0\}/\alpha$. The largest certified superset $\bar B^{\pm}$, with a fixed measurable tie rule, is \emph{Anchored-2BH}.

\begin{corollary}[Anchored-2BH]\label[corollary]{cor:twotail}
Under \eqref{eq:tailcov}, the values in \eqref{eq:twotailanchor} are local e-values and
\[
 \bar B^{\pm}\supseteq B^{\pm}\quad\text{for every data vector},\qquad
 \FDR(\bar B^{\pm})\le\alpha.
\]
More generally, every certified superset has these properties, and the simultaneous guarantee \eqref{eq:simulfdr} holds. No PRDS assumption on the folded two-sided p-values is required.
\end{corollary}
\begin{proof}
Writing $b_+=|B^+|$ and $b_-=|B^-|$, every true intersection $S$ satisfies
\[
 \FDP_S(B^{\pm})
 \le \frac{|B^+\cap S|}{b_+\vee1}
       +\frac{|B^-\cap S|}{b_-\vee1}.
\]
The two tailwise BH contribution bounds give expectation at most
$\alpha|S|/(2m)+\alpha|S|/(2m)=\alpha|S|/m$.
No independence between the tails is used. The auxiliary capped mean has expectation at most one by \cref{lem:gaussian}, so $\E E_S^{\pm}\le1$. Pointwise, $\alpha E_S^{\pm}\ge\FDP_S(B^{\pm})$; the e-closure proof of \cref{thm:anchor} completes the argument.
\end{proof}
The same argument applies whenever a Gaussian covariance restriction guarantees both tailwise BH contribution bounds. Condition~\eqref{eq:tailcov} is an explicit sufficient restriction, not numerical information required by the algorithm. Anchored-2BH anchors the \emph{union once at level $\alpha$}; it does not split the auxiliary evidence between two level-$\alpha/2$ closures. Its even scores pool evidence from both signs. The guarantee concerns ordinary two-sided nulls, not directional errors.

\subsection{Ordinary two-sided BH under folded PRDS}
For comparison, let $B^{(2)}=\BH_\alpha(p)$ with $p_i=2\phibar(|Z_i|)$. An explicit sufficient restriction for PRDS of these p-values is
\begin{equation}\label{eq:foldedcov}
 \Sigma_{I_0,I_0^c}=0,\qquad
 \Sigma_{I_0,I_0}\text{ positive definite},\qquad
 (D_0\Sigma_{I_0,I_0}^{-1}D_0)_{ij}\le0\quad(i\ne j),
\end{equation}
for a diagonal sign matrix $D_0$. Conditions involving an empty null block are vacuous. The centered null block is independent of the nonnull block and has an M-matrix precision after sign changes. Under \eqref{eq:foldedcov}, $|Z_{I_0}|$ is multivariate totally positive of order two (MTP$_2$), hence PRDS \citep{KarlinRinott1981,Lauritzen2018}. The common decreasing transformation $x\mapsto2\phibar(x)$ preserves MTP$_2$; independence from the nonnull block then gives PRDS of the entire p-value vector on $I_0$ \citep[Section~3.1, Case~3]{BY2001}. An equicorrelated null block with $0\le\rho<1$, independent of an arbitrary Gaussian nonnull block, is one example. The testing rule does not use $I_0$, $D_0$ or the covariance entries.

\begin{corollary}[Anchored-BH for pooled two-sided p-values]\label[corollary]{cor:twosided}
Suppose \eqref{eq:foldedcov} holds, or more generally that the two-sided p-values are PRDS on $I_0$. Replacing $B^{\pm}$ by $B^{(2)}$ in \eqref{eq:twotailanchor}, including the full-intersection convention, gives valid local e-values. Their largest certified superset $\bar B^{(2)}$ satisfies
\[
 \bar B^{(2)}\supseteq B^{(2)}\quad\text{pointwise},\qquad
 \FDR(\bar B^{(2)})\le\alpha.
\]
The simultaneous certification guarantee also holds.
\end{corollary}
\begin{proof}
For each true intersection $S$, folded PRDS gives
$\E\FDP_S(B^{(2)})\le\alpha|S|/m$.
Combine this with \cref{lem:gaussian} and \cref{thm:anchor}.
\end{proof}
This is an application of Anchored-BH to a valid two-sided input vector, not an additional procedure family. Nonnegative correlations alone do not imply PRDS of folded Gaussian p-values \citep{BY2001}. Although $B^{\pm}\subseteq B^{(2)}$ pointwise, retaining 2BH does not imply retaining pooled two-sided BH. Likewise, the centered-null precision condition should not be applied to a noncentral full vector without the independence clause in \eqref{eq:foldedcov}.

For ordinary $B^{(2)}$, the common $|Z|$ ordering permits the quadratic algorithm in Appendix~\ref{sec:quadratic}. The two-tail union need not be a prefix of that ordering; the smallest-sum constraints in Appendix~\ref{sec:computation} still apply, but the prefix-baseline monotone scan is not asserted for Anchored-2BH. Both constructions use the same independent-Gaussian calibration, with no covariance estimation. Since independence is allowed in both covariance classes, these restrictions improve the available \emph{baseline}, not the least-favorable cap factor for a fixed score. Section~\ref{sec:sim-2bh} evaluates Anchored-2BH under independence and positive equicorrelation, with balanced and one-direction signals; the ordinary pooled two-sided extension is not separately anchored in that experiment.

\section{Anchored-BY under arbitrary dependence}\label{sec:by}\label{sec:extensions}

Let $A=\BY_\alpha(p)=\BH_{\alpha/H_m}(p)$, where $H_m=\sum_{k=1}^m1/k$. For arbitrary dependence among superuniform null p-values, the coordinatewise bound
\begin{equation}\label{eq:bybudget}
 \E\left[\frac{\ind\{i\in A\}}{|A|\vee1}\right]\le\frac\alpha m,
 \qquad i\in I_0,
\end{equation}
follows directly from the BY staircase \citep{BY2001}. To see it, put $a=\alpha/(mH_m)$. If $p_i\in((k-1)a,ka]$ is rejected, then $|A|\ge k$, so its FDP contribution is at most $1/k$. The decreasing staircase with these heights and zero beyond $ma$ has uniform expectation $aH_m=\alpha/m$. Superuniformity can only reduce its expectation.

Summing \eqref{eq:bybudget} over a true intersection leaves the same residual coefficient as BH, now without a dependence restriction.

\begin{corollary}[Anchored-BY]\label[corollary]{cor:by}
Suppose the null p-values are superuniform, with arbitrary dependence. For any local e-values $F_S$ for $H_S$, define
\begin{equation}\label{eq:byanchor}
 E_S^{\rm BY}=\frac{|A\cap S|}{\alpha(|A|\vee1)}
                 +\left(1-\frac{|S|}{m}\right)F_S,
 \qquad A=\BY_\alpha(p),
\end{equation}
with the second term omitted for $S=[m]$. The largest certified superset $\bar A$, called \emph{Anchored-BY}, satisfies
\[
 \bar A\supseteq\BY_\alpha(p)\quad\text{pointwise},\qquad
 \FDR(\bar A)\le\alpha.
\]
Every certified superset enjoys the same conclusion, and the simultaneous guarantee holds.
\end{corollary}
\begin{proof}
The contribution bound gives $\E E_S^{\rm BY}\le |S|/m+(1-|S|/m)\E F_S\le1$. Its baseline certificate is retained pointwise, so the proof of \cref{thm:anchor} applies.
\end{proof}
Unboosted averages of normalized decreasing p-to-e scores provide $F_S$ under arbitrary dependence. In the Gaussian model, convex-even hinge mixtures can additionally use the independent-cap calibration of Section~\ref{sec:gaussian}, with no covariance restriction. Calibration based on Gaussianity is not part of the arbitrary-p-value guarantee.

Anchored-BY is incomparable in general with closed BY: Appendix~\ref{app:incomparability} gives pointwise witnesses in both directions for a fixed hinge-based variant, including the improved closed-BY comparator used in the simulations. Anchoring retains a corrected baseline; it does not remove that baseline's harmonic correction. This distinguishes its guarantee from a universally harmonic-free replacement for BY.

\begin{remark}
The broader principle used above is as follows.
Let $A$ be a baseline rejection set. For every nonempty $S$, suppose a deterministic $q_S\in[0,1]$ satisfies
 $\E\{\FDP_S(A)/\alpha\}\le q_S$
whenever all hypotheses in $S$ are true.
For intersection e-values $F_S$, define
 $E_S=\frac{\FDP_S(A)}\alpha+(1-q_S)F_S,$
omitting the auxiliary term when $q_S=1$. Then $E_S$ is an intersection e-value. Every measurable certified superset of $A$ controls FDR at $\alpha$ and contains all of $A$'s discoveries.
\end{remark}

\section{Anchored e-BH under arbitrary dependence}
\label{sec:anchored-ebh}



The anchoring principle also applies directly to e-BH
\cite{WangRamdas2022}. Let $e_1,\ldots,e_m$ be nonnegative
e-values, satisfying $\mathbb E[e_i]\le 1$ whenever $H_i$
is true, with no assumptions on their dependence.
Write $B=\mathrm{eBH}_\alpha(e)$ and $b=|B|$.
The self-consistency property of e-BH gives
\[
\frac{\mathbf 1\{i\in B\}}{\alpha(b\vee1)}
\le \frac{e_i}{m},
\qquad i\in[m].
\]
Consequently, for every $S\subseteq[m]$,
\[
\frac{\mathrm{FDP}_S(B)}{\alpha}
\le \frac1m\sum_{i\in S}e_i.
\]
This suggests the following particularly simple
local e-values:
\begin{equation}
\label{eq:anchored-ebh}
E_S
=
\frac1m\left(
\sum_{i\in S}e_i+m-|S|
\right),
\qquad \varnothing\ne S\subseteq[m].
\end{equation}
Indeed, under $H_S$,
\[
\mathbb E[E_S]
\le \frac{|S|}{m}+\frac{m-|S|}{m}=1.
\]
Moreover, these e-values pointwise dominate
the canonical certificates of $B$.

The anchoring interpretation can be made explicit:
\[
E_S
=
\frac{\mathrm{FDP}_S(B)}{\alpha}
+\left(1-\frac{|S|}{m}\right)
+\frac1m\sum_{i\in S}
\left(
e_i-\frac{m\mathbf 1\{i\in B\}}
{\alpha(b\vee1)}
\right).
\]
Both additional terms are nonnegative.
Thus we retain the baseline certificate and
reallocate the evidence left unused by e-BH.

\begin{proposition}[Anchored e-BH]
\label{prop:anchored-ebh}
Let $(E_S)$ be defined by
\eqref{eq:anchored-ebh}, and let $\bar B$ be
a measurable largest-cardinality certified
superset of $B$, satisfying
\[
\mathrm{FDP}_S(\bar B)\le\alpha E_S
\qquad\text{for every nonempty }S\subseteq[m].
\]
Under arbitrary dependence among the e-values,
\[
\bar B\supseteq B,
\qquad
\mathrm{FDR}(\bar B)\le\alpha.
\]
The containment can be strict.
\end{proposition}

\begin{proof}
The self-consistency inequality above shows that
$\mathrm{FDP}_S(B)\le\alpha E_S$ for every $S$,
so $B$ is certified. Hence a certified superset
$\bar B$ exists. Evaluating its closure constraint
at the true-null set $I_0$ gives
\[
\mathrm{FDR}(\bar B)
\le\alpha\mathbb E[E_{I_0}]
\le\alpha.
\]
The case $I_0=\varnothing$ is immediate.
\end{proof}

\paragraph{Incomparability with closed e-BH.}
Anchored e-BH and mean closed e-BH are pointwise incomparable.
Indeed, let $m=3$ and $\alpha=1/2$, and consider the local
e-values
\[
E_S^{\mathrm A}
=\frac{\sum_{i\in S}e_i+3-|S|}{3},
\qquad
E_S^{\mathrm C}
=\frac{1}{|S|}\sum_{i\in S}e_i,
\]
corresponding to Anchored e-BH and mean closed e-BH,
respectively. For $e=(4,4,1/10)$, direct verification of the
closure constraints yields
\[
\bar B^{\mathrm A}=\{1,2,3\}
\supsetneq \bar B^{\mathrm C}=\{1,2\}.
\]
Conversely, for $e=(4,2,4/5)$, the inequalities reverse:
\[
\bar B^{\mathrm C}=\{1,2,3\}
\supsetneq \bar B^{\mathrm A}=\{1,2\}.
\]
In the latter case, for example, Anchored e-BH cannot reject
all three hypotheses because
$E_{\{2,3\}}^{\mathrm A}=19/15<4/3$.
Thus neither procedure uniformly dominates the other.

\paragraph{Computation and connection to donation.}
The closure has a simple step-up characterization.
Put
\[
h_i=(e_i-1)_+,
\qquad
C=\sum_{i=1}^m(e_i\wedge1).
\]
For a nonempty candidate $R$ of size $r$,
the entire collection of closure constraints is
equivalent to
\begin{equation}
\label{eq:anchored-ebh-stepup}
\sum_{i\in R}
\left(\frac{m}{\alpha r}-h_i\right)_+
\le C.
\end{equation}
For fixed $r$, the left-hand side is minimized
by choosing the $r$ largest e-values.
After sorting, prefix sums and a binary search
for the positive summands evaluate each candidate
size in $O(\log m)$ time. Thus the largest
certified rejection set can be found in
$O(m\log m)$ time and $O(m)$ memory.

Condition \eqref{eq:anchored-ebh-stepup} is also
the feasibility condition obtained by allowing
each e-value to donate up to $e_i\wedge1$
units of evidence to other hypotheses.
Thus Anchored e-BH recovers the offline
donation interpretation of e-BH through a
particularly simple collection of local e-values.

For a strict example, take $m=2$, $\alpha=0.05$,
and $(e_1,e_2)=(39.5,1)$. Ordinary e-BH rejects
nothing. However,
\[
E_{\{1\}}=E_{\{1,2\}}=20.25>1/\alpha,
\]
so $\{1\}$ is certified by the anchored procedure.

\section{Anchored-dBH under conditional calibration}\label{sec:dbh}
Conditional calibration provides another baseline for which the remaining budget is explicit. The dependence-adjusted BH procedure of \citet{FithianLei2022}, denoted dBH, uses a specified conditional null model to calibrate each hypothesis's FDR contribution. This is a different assumption from PRDS: in the Gaussian example the covariance matrix is known, may have either sign, and is used by the baseline. Anchoring adds a second, inexpensive step without recalculating those conditional thresholds.

\subsection{The coordinatewise guarantee behind dBH}
Write $C_i$ for the conditioning statistic used for hypothesis $i$, and let $\widehat r_i\ge1$ be the rejection-count estimate used in calibration. With $\tau_i(c;X)$ a nondecreasing threshold family, the calibrated $\widehat c_i(C_i)$ satisfies
\begin{equation}\label{eq:dbhcalibration}
 \sup_{P\in H_i}\E_P\left[
 \frac{\ind\{p_i\le\tau_i(\widehat c_i;X)\}}{\widehat r_i}
 \given C_i\right]\le\frac\alpha m.
\end{equation}
For dBH$_\gamma$, $\tau_i(c;X)=c|\BH_c(p)|/m$ and
$\widehat r_i=|\BH_{\gamma\alpha}(p)\cup\{i\}|$, with fixed $0<\gamma\le1$.
Let $A^+=\{i:p_i\le\tau_i(\widehat c_i;X)\}$. If $|A^+|\ge\widehat r_i$ for every selected $i$, set $A=A^+$. Otherwise use the randomized pruning step prescribed by \citet{FithianLei2022}; it must not be silently omitted. Their proof gives the coordinatewise inequality
\begin{equation}\label{eq:dbhcoordinate}
 \E\frac{\ind\{i\in A\}}{|A|\vee1}
 \le \E\frac{\ind\{i\in A^+\}}{\widehat r_i}
 \le\frac\alpha m,\qquad i\in I_0.
\end{equation}
The first expectation includes the pruning randomness when present. Thus dBH retains the $\alpha m_0/m$ bound, even though its individual conditional calibration can be tight. The unused budget comes from summing only over true nulls; we do not assume that its fitted conditional thresholds can be increased for free.

\subsection{A no-loss augmentation without further conditional integration}
For any local e-values $F_S$ valid in the conditional-calibration model, define
\begin{equation}\label{eq:dbhanchor}
 E_S^{\rm dBH}=
 \frac{|A\cap S|}{\alpha(|A|\vee1)}
 +\left(1-\frac{|S|}{m}\right)F_S.
\end{equation}
\Needspace{12\baselineskip}
\begin{theorem}[Anchored-dBH]\label{thm:dbhanchor}
Assume the conditional calibration and any required pruning of \citet{FithianLei2022} are performed so that \eqref{eq:dbhcoordinate} holds. Then \eqref{eq:dbhanchor} is a local e-collection. Every measurable e-closure-certified superset $\bar A$ satisfies
\[
 \bar A\supseteq A\quad\text{for every realized data and pruning vector},
 \qquad \FDR(\bar A)\le\alpha.
\]
The auxiliary evidence need not be independent of the baseline. If the baseline is randomized, the FDR statement averages over that same randomness.
\end{theorem}
\begin{proof}
Under $H_S$, summing \eqref{eq:dbhcoordinate} over $i\in S$ gives
$\E E_S^{\rm dBH}\le |S|/m+(1-|S|/m)\E F_S\le1$.
Pointwise, $\alpha E_S^{\rm dBH}\ge\FDP_S(A)$, so $A$ is certified. Apply e-closure at the actual null intersection.
\end{proof}
The result also applies to dBY and valid recursive refinements of dBH when their final coordinate contributions obey \eqref{eq:dbhcoordinate}. It does not assert an ordering between an anchor of first-round dBH and the published second-round procedure dBH$^{(2)}$. Conditional recalibration and anchoring exploit different sources of conservativeness.

In the known-covariance Gaussian model, unboosted Gaussian-hinge means or their frozen mixtures give $F_S$ immediately; \cref{lem:gaussian} permits additional fixed-score cap calibration. For non-Gaussian models, marginal p-to-e calibrators give a simpler alternative. These choices require no new conditional integrations after dBH has returned its rejection mask.

The generic pilot construction of Section~\ref{sec:coordinate} also applies. If $A_\lambda$ is dBH run at a fixed level $\lambda>\alpha$, then
\[
 x_i^{(\lambda)}=
 \frac{m\ind\{i\in A_\lambda\}}{\lambda(|A_\lambda|\vee1)}
\]
is an individual e-value by \eqref{eq:dbhcoordinate}, with $\lambda$ in place of $\alpha$. Means over $S$ and prespecified mixtures of such pilot scores are valid. Unlike PRDS pilot scores, these rely on conditional calibration and require the extra pilot computations. Our experiments use fixed hinges rather than these more expensive pilots.

\subsection{Gaussian implementation and scope}
For $Z\sim N_m(\mu,\Sigma)$ with known positive definite correlation matrix, take
$C_i=Z_{-i}-\Sigma_{-i,i}Z_i$. Under $H_i:\mu_i=0$, the conditional null path is
\[
 Z^{(i,t)}=Z-\Sigma_{\cdot i}Z_i+\Sigma_{\cdot i}t,
 \qquad t\sim N(0,1).
\]
If $q_i$ is the observed BH-adjusted p-value, maximal first-round dBH selects $i$ initially exactly when
\begin{equation}\label{eq:dbhintegral}
 \int_\R
 \frac{\ind\{i\in\BH_{q_i}(p(Z^{(i,t)}))\}}
 {|\BH_{\gamma\alpha}(p(Z^{(i,t)}))\cup\{i\}|}
 \phi(t)\,dt\le\frac\alpha m.
\end{equation}
The finite-knot integration is the homotopy construction of \citet[Section~4]{FithianLei2022}, not a new baseline algorithm. For two-sided Gaussian p-values, all possible knots solve
$|Z_j-\Sigma_{ji}Z_i+\Sigma_{ji}t|=\Phi^{-1}(1-cr/(2m))$
for $c\in\{q_i,\gamma\alpha\}$, $j,r\in[m]$ with $\Sigma_{ji}\ne0$. The integrand is constant between consecutive knots; its integral is a finite sum of normal probabilities. The denominator includes $i$ even when the reference BH run does not reject it.

The final mask need not be a prefix of the marginal p-value order. The three-list algorithm in Appendix~\ref{sec:arbitrarybasecompute} gives the \emph{additional} closure cost
$O((m-|A|)m\log m+m\log m)$ with finitely many common-order score profiles. This does not turn the entire dBH procedure into an $O(m\log m)$ algorithm: the baseline's conditional integrations are still required. The experiments separately time calibration and augmentation, retain any prescribed pruning, and use no numerical integration grid.

\section{Anchored ordered testing via conditional residuals}\label{sec:ordered}
The prescribed order can provide useful information beyond the p-values themselves. We now assume that null p-values are mutually independent exact uniforms and independent of the nonnull vector, and that the hypothesis order is fixed independently of the tested null p-values. These are \emph{offline ordered} procedures, not online tests making irrevocable decisions. SeqStep+ rejects a complete prefix; Selective SeqStep+ and Adaptive SeqStep reject small p-values within a prefix \citep{BarberCandes2015,LeiFithian2016}. We preserve these structures when augmenting them.

Their usual FDR bounds do not supply BH's coefficient $1-|S|/m$. Instead, we recycle a \emph{conditional coordinate budget}. The resulting closure is simpler than the general intersection search and can add discoveries even when the baseline is empty.

\subsection{Conditional residual certificates and structural closure}\label{sec:orderedresidual}
For a baseline $A$, write $d_i(A)=\ind\{i\in A\}/(|A|\vee1)$. Suppose a computable, nonnegative $\mathcal G_i$-measurable budget $\eta_i$ satisfies
\begin{equation}\label{eq:ordbudget}
 \E_{H_i}\{d_i(A)\mid\mathcal G_i\}\le\eta_i,
 \qquad \sum_{i\in I_0}\E\eta_i\le\alpha.
\end{equation}
The conditional inequality must hold throughout the null model. A sufficient source is a data-dependent budget $b_i\ge d_i(A)$ with $\sum_{i\in I_0}\E b_i\le\alpha$, setting $\eta_i=\E_{H_i}(b_i\mid\mathcal G_i)$. Define
\[
 \rho_i=\eta_i-\E_{H_i}\{d_i(A)\mid\mathcal G_i\}\ge0.
\]
\begin{theorem}[Conditional residual anchoring]\label[theorem]{thm:ordresidual}
For $f_i\ge0$ satisfying $\E_{H_i}(f_i\mid\mathcal G_i)\le1$, let
\begin{equation}\label{eq:ordcertificate}
 a_i=d_i(A)+\rho_i f_i,\qquad
 E_S=\frac1\alpha\sum_{i\in S}a_i.
\end{equation}
Then $(E_S)$ is a local e-collection and certifies $A$. Every measurable certified superset $\bar A\supseteq A$ retains the baseline and controls FDR at $\alpha$. For a nonempty candidate $R$, certification is equivalent to
\begin{equation}\label{eq:ordsingleton}
 a_i\ge\frac1{|R|}\quad\text{for every }i\in R.
\end{equation}
\end{theorem}
\begin{proof}
Conditioning gives $\E_{H_i}a_i\le\E\eta_i$, so under $H_S$,
$\E E_S\le\alpha^{-1}\sum_{i\in I_0}\E\eta_i\le1$. Also $a_i\ge d_i(A)$ pointwise, and summing over $S$ certifies $A$. The actual-null constraint bounds the selected FDP by $\alpha E_{I_0}$. Finally, singleton intersections give necessity of \eqref{eq:ordsingleton}; summing these inequalities over $R\cap S$ and using $a_i\ge0$ proves sufficiency.
\end{proof}
No independence between $A$ and the scores is needed beyond the conditional guarantees. The unrestricted closure is e-BH on the compound e-values $ma_i/\alpha$, taking $O(m\log m)$ time after score construction \citep{WangRamdas2022}. The budget-and-score argument is related to conditional calibration \citep{FithianLei2022,LuoFithianLei2024,LeeRen2024}; its role here is to preserve the realized baseline certificate and permit a prescribed ordered structure.

For a complete-prefix baseline, the largest certified prefix has endpoint
\begin{equation}\label{eq:ordprefix}
 k^*=\max\left(\left\{k\ge |A|:\min_{i\le k}a_i\ge1/k\right\}\cup\{|A|\}\right).
\end{equation}
Running minima compute it in $O(m)$ time. For a fixed selective cutoff $s$, the same scan follows only the original-order indices with $p_i\le s$, using their cumulative rejection count. A more flexible family is
\begin{equation}\label{eq:ordrectangle}
 R(k,t)=\{i\le k:p_i\le t\},\qquad A\subseteq R(k,t).
\end{equation}
For each $k$, scan its indices in increasing p-value order and check \eqref{eq:ordsingleton} with running minima. A single global sort followed by these scans gives the largest certified rectangle in $O(m^2)$ time. This permits a \emph{certified} cutoff change; it is not an unsupported union after tuning the original rule.

\subsection{Residual evidence concentrated on near misses}\label{sec:orderedband}
Here $\mathcal G_i=\sigma(p_{-i})$, so the conditional null coordinate is $U_i=p_i\sim\Unif(0,1)$. With all probabilities and expectations conditional on $\mathcal G_i$, define
\[
 M_i=\Prob(i\notin A\mid\mathcal G_i),\qquad
 H_i(u)=\Prob(U_i\le u,\ i\notin A\mid\mathcal G_i),
\]
\[
 \kappa_i=\E\{|A|+1\mid\mathcal G_i,i\notin A\},\qquad
 \ell_i=\min\{M_i,\rho_i\kappa_i\}.
\]
Set the score to zero if $M_i=0$ or $\ell_i=0$; otherwise use the \emph{residual-band score}
\begin{equation}\label{eq:ordband}
 f_i=\frac{\ind\{i\notin A,\ H_i(p_i)\le\ell_i\}}{\ell_i}.
\end{equation}
Under the conditional null and nonselection, $H_i(U_i)$ is uniform on $[0,M_i]$. Therefore the numerator has conditional probability $\ell_i$ and the score has conditional expectation one. All quantities defining its width are functions of the conditioning information, not choices made after inspecting the tested coordinate.

The band spends the first $\ell_i$ units of nonselection probability on the strongest remaining p-values. If its width is not capped, $\rho_i/\ell_i=1/\kappa_i$ matches a conditional prospective rejection scale. This is a design choice, not an optimality claim. It exploits information discarded by categorical baseline decisions: two p-values in $[0,s]$ can have very different residual evidence. Prespecified mixtures of conditional scores also remain valid.

\subsection{Selective SeqStep+ and Adaptive SeqStep}\label{sec:orderedselective}
Fix $0<s\le\lambda<1$ and put
\[
 R_s(k)=\sum_{i\le k}\ind\{p_i\le s\},\qquad
 A_\lambda(k)=\sum_{i\le k}\ind\{p_i>\lambda\}.
\]
Adaptive SeqStep selects
\begin{equation}\label{eq:ordase}
 K=\max\left(\left\{k:
 \frac{s}{1-\lambda}\frac{1+A_\lambda(k)}{R_s(k)\vee1}\le\alpha
 \right\}\cup\{0\}\right),\qquad
 B=\{i\le K:p_i\le s\}.
\end{equation}
An empty eligible prefix produces no rejections. Selective SeqStep+ is the case $\lambda=s$ \citep{LeiFithian2016}. To retain a budget when the baseline is empty, extend the stopping index:
\begin{equation}\label{eq:ordkstar}
 K^*=\max\left\{K,\ \max\left(\left\{k:R_s(k)<
           \frac{s}{\alpha(1-\lambda)}\right\}\cup\{0\}\right)\right\}.
\end{equation}
The following budgets are deductions from the ordered stopping proofs, not procedures attributed to those papers:
\begin{equation}\label{eq:ordasebudget}
 b_i=\frac{\alpha(1-\lambda)}s
       \frac{\ind\{i\le K^*,\ p_i\le s\}}{1+A_\lambda(K^*)}.
\end{equation}
\begin{proposition}[Selective ordered budgets]\label[proposition]{prop:ordbudget}
Under the independent-null assumptions of this section, \eqref{eq:ordasebudget} satisfies $b_i\ge d_i(B)$ and $\sum_{i\in I_0}\E b_i\le\alpha$. Thus \cref{thm:ordresidual} applies with $A=B$ and $\eta_i=\E(b_i\mid p_{-i})$.
\end{proposition}
\begin{proof}
A nonempty accepted prefix has $R_s(K)\ge s/\{\alpha(1-\lambda)\}$, so $K^*=K$ when $B$ is nonempty. The baseline inequality then gives $b_i\ge1/|B|$ on $B$; otherwise domination is automatic. The reverse null-ratio stopping argument for Adaptive SeqStep gives
\[
 \E\frac{V_s(K^*)}{1+V^-_\lambda(K^*)}\le\frac{s}{1-\lambda},
\]
where $V_s$ and $V^-_\lambda$ count null p-values below $s$ and above $\lambda$. The additional early-count crossing in \eqref{eq:ordkstar} is a reverse stopping rule for the same filtration. Conditional on the nulls in the two informative categories, the favorable probability is $s/(s+1-\lambda)$; the middle category contributes neither count. Removing nonnull upper-tail counts only increases the ratio. Summing \eqref{eq:ordasebudget} and applying the displayed bound proves the result. This is the reverse-supermartingale argument of \citet{LeiFithian2016}, with the early stop analogous to the budget stop in \citet{LuoFithianLei2024}.
\end{proof}

For each $i$, replacing $p_i$ by a conditional null uniform changes the baseline and budget only across $[0,s]$, $(s,\lambda]$, and $(\lambda,1]$. Each baseline evaluation is a linear prefix scan. Two or three category evaluations per coordinate therefore compute $\eta_i$, $\rho_i$, $M_i$, $H_i(p_i)$, and $\kappa_i$ exactly in $O(m^2)$ time. The same-cutoff scan is linear and the rectangle scan quadratic. No numerical quadrature or intersection enumeration is required. We call these the \emph{same-cutoff} and \emph{rectangle} ordered anchors, respectively.

\subsection{SeqStep+ and bounded accumulation tests}\label{sec:orderedprefix}
For $0<c<1$, SeqStep+ rejects $B=[K]$, where
\[
 K=\max\left(\left\{k:\frac{1+A_c(k)}{k+1}\le\alpha(1-c)\right\}\cup\{0\}\right).
\]
Extend to $K^*=\max\{K,\max(\{k:k+1<1/[\alpha(1-c)]\}\cup\{0\})\}$. For $K^*>0$, valid budgets are
\begin{equation}\label{eq:ordseqbudget}
 b_i=\alpha(1-c)\left(1+\frac1{K^*}\right)
       \frac{\ind\{i\le K^*\}}{1+A_c(K^*)};
\end{equation}
use zero if $K^*=0$. On a nonempty baseline, $K^*=K$ and the threshold inequality gives domination of $1/K$. Its null sum is bounded by
$\alpha(1-c)\{1+V(K^*)\}/\{1+A_c^0(K^*)\}$, whose expectation is at most $\alpha$ by the Bernoulli reverse-supermartingale bound \citep{BarberCandes2015}. The two-category conditional calculation and \eqref{eq:ordprefix} yield a quadratic, complete-prefix-preserving anchor.

A first conditional fallback calibration may instead be used as an intermediate baseline: choose a cutoff $t_i(p_{-i})$ such that
\[
 \E\left[
 \frac{\ind\{i\in B\}\vee\ind\{U_i\le t_i\}}{|B\cup\{i\}|}
 \ \middle|\ p_{-i}\right]\le\eta_i.
\]
Take its largest prefix retaining $B$, denoted $C$, and recycle the actual conditional contribution $d_i(C)$ in \cref{thm:ordresidual}. Prefix restriction is justified by the fallback certificate: a retained new coordinate has denominator at least $|B\cup\{i\}|$ and a baseline coordinate at least $|B|$. Arbitrary filtering would not have this guarantee. Along a tested coordinate, the other fallback decisions depend only on its category and its own decision changes at its fallback cutoff. At most four intervals suffice for each integral, giving an $O(m^3)$ direct implementation of the recycled prefix, followed by a linear scan.

\paragraph{Bounded accumulation scores and scope.}
More generally, let $h:[0,1]\to[0,C]$ satisfy $\int_0^1h(u)\,du=1$. The corrected accumulation rule rejects $[K]$ with
$K=\max(\{k:[C+\sum_{i\le k}h(p_i)]/(k+1)\le\alpha\}\cup\{0\})$.
Extend it to $K^*=\max\{K,\max(\{k:k+1<C/\alpha\}\cup\{0\})\}$. Then
\begin{equation}\label{eq:ordaccum}
 b_i=\alpha\left(1+\frac1{K^*}\right)
 \frac{\ind\{i\le K^*\}}{C+\sum_{j\le K^*}h(p_j)}
\end{equation}
is a valid budget, with the zero convention. Its null sum is bounded by
$\alpha\{1+V(K^*)\}/\{C+\sum_{j\le K^*,j\in I_0}h(p_j)\}$; Lemma~3 of \citet{LiBarber2017} bounds the expectation of the ratio by one. The forced early segment is a permitted deterministic-threshold last-crossing rule: before its endpoint use a threshold at least $kC$, then $\alpha(k+1)-C\ge0$. For $h$ constant on $J$ intervals, exact categorical integration costs $O(Jm^2)$; no separate power study of general $h$ is claimed here.

The finite-$C$ assumption is essential; unbounded accumulation scores require a separate budget and integration argument. Raw SeqStep and raw Selective SeqStep have modified-FDR guarantees rather than the ordinary-FDR guarantee of their plus variants. They are not included among the ordinary-FDR anchors evaluated below.

\section{Simulation study}\label{sec:simulations}
We present two experiments for each main family in the order of Sections~\ref{sec:anchor}, \ref{sec:twobh}, \ref{sec:by}, \ref{sec:dbh}, and \ref{sec:ordered}. The main study contains 1,122,000 data sets across 227 settings. Appendix~\ref{sec:sim-negative} retains 83,000 negative-dependence data sets (30 settings), and Appendix~\ref{app:shiftsimulations} retains 73,000 shifted-method data sets (37 settings), for an overall FDR-study total of 1,278,000 data sets across 294 settings (\cref{tab:simulationdesign}). The ordered study uses $\alpha=0.10$ and $m=40$; all other experiments use $\alpha=0.05$. BH, 2BH, and BY use $m=100$; dBH uses $m=40$ and $20$. The appendices use $m=100$ and $5000$ for negative dependence and $m=40$ for shifted methods. The ordered evaluations were completed separately using prespecified scores and are integrated here without rerunning or changing any results.

Procedures share the same data within each setting. Supports, signal signs, pilot levels and score weights are fixed before evaluation; there is no fitting to these evaluation data. We estimate FDR by the mean of replicate FDPs and power by the mean fraction of nonnulls rejected. Confidence intervals are pointwise 95\% Monte Carlo intervals, paired when comparing procedures. Comparisons use a common \emph{nominal} FDR level, not matched realized FDR. Full designs, all settings (including no-gain settings), per-replicate counts and implementation checks are supplied.

\begin{table}[htbp]
\centering\small
\caption{FDR simulation coverage in presentation order. Each main subsection contains two experiments, which may share the same setting grid. Appendix counts are separate; no data set is counted twice. Counts include global-null checks.}\label{tab:simulationdesign}
\begin{tabular}{@{}llrr@{}}
\toprule
Family & Experiments & Settings & Data sets\\
\midrule
Anchored-BH & Non-Gaussian PRDS; Gaussian PRDS &74&395,000\\
Anchored-2BH & Balanced signs; one-direction signals &50&260,000\\
Anchored-BY & Antithetic p-values; arbitrary Gaussian &62&335,000\\
Anchored-dBH & Mixture grid; one weak signal &29&67,000\\
Ordered anchors & Selective methods; complete prefixes &12&65,000\\
\midrule
Main-paper total &&227&1,122,000\\
Anchored-NDBH (appendix) & Antithetic; negative Gaussian blocks &30&83,000\\
Shifted anchors (appendix) & Known covariance; Wishart information &37&73,000\\
\midrule
FDR-study total & &294&1,278,000\\
\bottomrule
\end{tabular}
\end{table}

\subsection{Anchored-BH: two PRDS experiments}\label{sec:sim-bh}
\paragraph{Experiment 1: genuinely non-Gaussian p-values.}
Null p-values are uniform; nonnull p-values have distribution $\Beta(\kappa,1)$, with $\kappa\in\{0.05,0.10,0.20,0.35\}$ and nonnull fractions $0.1,0.5,0.9$. Dependence is either independence or ten independent comonotonic blocks of size ten. Within each block, the p-values are increasing transforms of one shared uniform, giving PRDS without Gaussianity. There are 5,000 repetitions in each nonnull setting and 10,000 for each of two global nulls: 26 settings and 140,000 data sets.

We compare BH, canonical closed BH and several Anchored-BH scores. Canonical closed BH is the fixed-$\alpha$ horizon rule of \citet[Section~5]{Goeman2026}, without additional raw-p-value censoring. Anchors use the frozen hinges in \cref{tab:weights}, with and without cap calibration; a BH pilot at $\beta=0.10$; its adaptive hinge; and an equal mixture of pilot indicators at $\beta\in\{0.10,0.20,0.50\}$. A fixed staircase using the same profile weights is also included in the supplement. These are prespecified choices, not optimized power claims.

For independent $\Beta(0.1,1)$ signals with 90\% nonnulls, powers are 70.91\% for BH, 74.22\% for closed BH, 73.25\% for the pilot anchor and 73.36\% for the adaptive staircase. Thus the adaptive staircase gains over BH but loses 0.86 percentage points to closed BH (interval $-0.92$ to $-0.80$). At 10\% nonnulls, the anchors make essentially no additional discoveries. This experiment verifies that the construction is not specific to Gaussian p-values, but does not establish general power superiority over closed BH.

\paragraph{Experiment 2: signal strength and positive dependence.}
Here $p_i=\Phi(-Z_i)$, positive nonnull means are $\mu\in\{2,2.5,3,3.5,4\}$, and $\Sigma$ is either $I$ or equicorrelated with correlation $0.5$ or $0.9$. These one-sided p-values satisfy the Gaussian PRDS sufficient condition. Crossing the same three nonnull fractions, with 5,000 repetitions per setting and three 10,000-repetition global nulls, gives 48 settings and 255,000 data sets.

\Cref{tab:prds} separates sparse, intermediate and dense signals. At 90\% nonnulls and independent mean 3, the adaptive staircase has power 96.64\%, versus 93.96\% for closed BH and 89.53\% for BH. Its paired gains are 2.68 points over closed BH (2.63--2.72) and 7.11 points over BH (7.04--7.17). At dense mean 2, however, it loses 4.25 points to closed BH ($-4.34$ to $-4.16$). \Cref{fig:positive} shows the independent strength grid. At dense mean 3 with equicorrelation $0.5$, powers of BH, closed BH and the adaptive staircase are 86.50\%, 90.03\% and 89.77\%; with correlation $0.9$, they are 84.82\%, 85.98\% and 85.59\%.

\begin{table}[htbp]
\centering\small
\caption{Anchored-BH: power (\%) under independence, with one-sided Gaussian mean 3. All methods have their stated PRDS guarantee. The fixed hinge is cap-calibrated.}\label{tab:prds}
\begin{tabular}{@{}lrrrrrr@{}}
\toprule
Nonnull & BH & Closed BH & Fixed hinge & Pilot BH & Adaptive hinge & Adaptive stair\\
\midrule
10\% &61.53&61.68&61.53&61.53&61.53&61.53\\
50\% &83.40&86.14&83.40&83.40&83.40&83.40\\
90\% &89.53&93.96&89.60&94.88&93.65&96.64\\
\bottomrule
\end{tabular}
\end{table}

At the independent dense mean-3 setting, estimated FDRs are 0.49\% for BH, 0.87\% for closed BH, 0.99\% for the pilot anchor and 1.57\% for the adaptive staircase. Additional power comes with additional false discoveries, though all these estimates are below 5\%. Under the independent global null, BH and its anchors have the same rejection event and estimated FDR 4.88\%. The largest incremental power gain from cap calibration over the identical unboosted hinge anchor is only 0.014 percentage points across the two PRDS experiments; the multi-point gains come primarily from the adaptive score design.

\begin{figure}[htbp]
\centering\includegraphics[width=.6\textwidth]{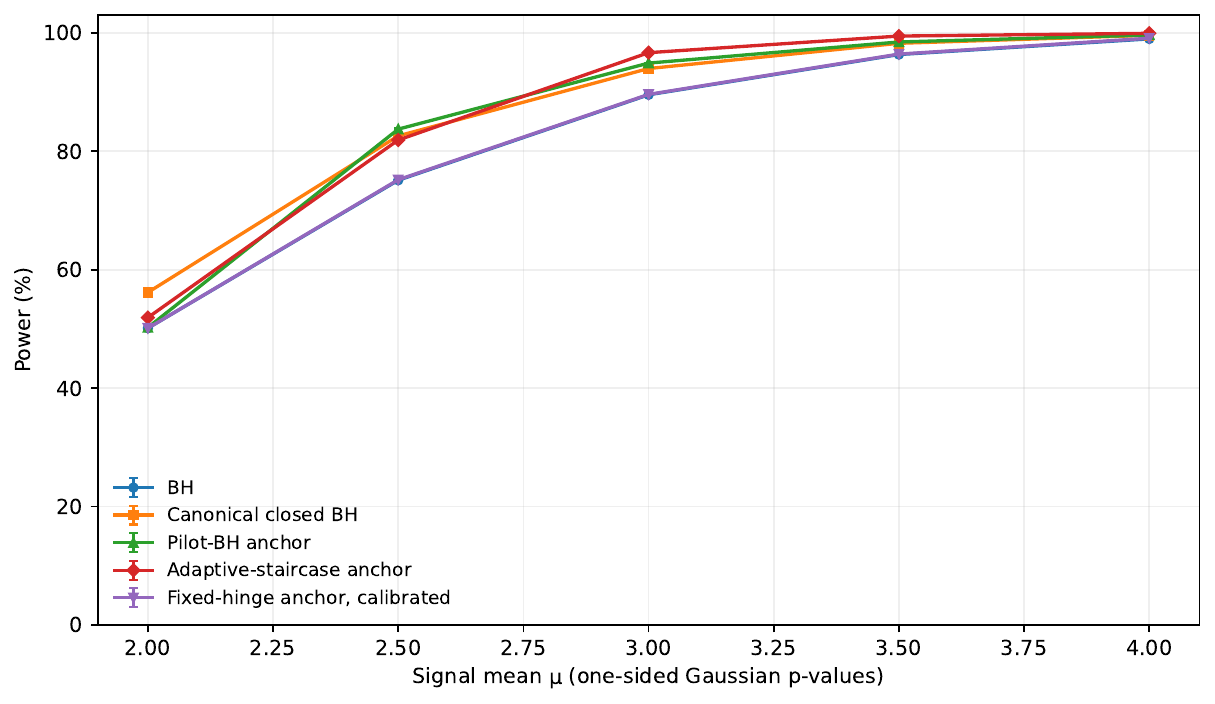}
\caption{Anchored-BH, independent one-sided Gaussian p-values, 90\% nonnulls. Error bars are pointwise 95\% Monte Carlo intervals. The anchor retains BH, but its ordering relative to canonical closed BH changes with signal strength.}\label{fig:positive}
\end{figure}
\FloatBarrier

\subsection{Anchored-2BH: balanced versus one-direction signals}\label{sec:sim-2bh}
\paragraph{Experiment 1: recover part of the cost of splitting the tails.}
We generate $Z\sim N_m(\mu,\Sigma)$ with $\Sigma=I$ or equicorrelation $0.5$, nonnull fractions $0.1,0.5,0.9$, and balanced positive and negative signal means of magnitude $2.5,3,3.5,4$. All these covariance matrices meet \cref{cor:twotail}. The baseline is the union of BH at $\alpha/2$ on each length-$m$ one-sided p-value vector. Anchored-2BH retains that union and adds evidence from the same frozen even hinge profiles as before, with or without independent-Gaussian cap calibration. No covariance matrix is supplied to the method.

An additional comparator is the union of canonical closed BH at $\alpha/2$ in each tail. Each tail has its PRDS guarantee, and adding the two tailwise FDR bounds validates their union at $\alpha$. Ordinary pooled two-sided BH is displayed only as a diagnostic outside independence. The new baseline need not be an absolute-statistic prefix: the implementation therefore checks the full smallest-sum criterion, rather than applying the ordered-baseline quadratic shortcut. There are 26 balanced-sign settings and 140,000 data sets, including two global nulls.

For independent magnitude 3.5 and 90\% nonnulls, powers are 87.98\% for tail-split BH, 89.64\% for tail-split closed BH, and 92.25\% for Anchored-2BH. The paired gains are 4.27 points over its baseline (4.21--4.33) and 2.60 points over the closed comparator (2.55--2.65). Ordinary pooled BH has power 92.89\% in this independent setting. Thus the anchor recovers much, but not all, of the loss from splitting the evidence across signs (\cref{fig:twobh}). The same balanced setting at correlation $0.5$ gives 85.65\%, 86.87\% and 89.88\%, respectively (\cref{tab:twobh}). Sparse and intermediate settings mostly retain the baseline.

\paragraph{Experiment 2: concentrate the signals in one tail.}
We repeat the nonnull grid with all nonnull means positive: 24 settings and 120,000 data sets. Tail-split BH now pools nearly all signal evidence in a single tail. The fixed even-score anchor adds much less, and the two-tail closed-BH comparator can be substantially better. At independent magnitude 3.5 and 90\% nonnulls, tail-split BH has power 92.83\%, Anchored-2BH 92.87\%, and tail-split closed BH 95.79\%. The gain in the balanced experiment is therefore not a generic improvement of the same size for every signal-sign arrangement.

\begin{table}[htbp]
\centering\small
\caption{Anchored-2BH: power (\%) with 90\% nonnulls of magnitude 3.5. All three displayed procedures are valid under both covariance choices. Anchored-2BH uses calibrated frozen hinges.}\label{tab:twobh}
\begin{tabular}{@{}llrrr@{}}
\toprule
Signal signs & Correlation & Tail-split BH & Tail-split closed BH & Anchored-2BH\\
\midrule
Balanced &0&87.98&89.64&92.25\\
Balanced &0.5&85.65&86.87&89.88\\
Positive &0&92.83&95.79&92.87\\
Positive &0.5&91.33&93.77&91.38\\
\bottomrule
\end{tabular}
\end{table}

In the balanced independent magnitude-3.5 setting, estimated FDRs are 0.25\% for the baseline, 0.31\% for the tail-split closed comparator and 0.46\% for Anchored-2BH. Under the independent global null, all three have estimated FDR 4.91\%; at correlation $0.5$, the estimate is 3.54\%. Baseline containment and calibrated-versus-unboosted containment were checked coordinatewise on every simulated data set.

\begin{figure}[htbp]
\centering\includegraphics[width=.6\textwidth]{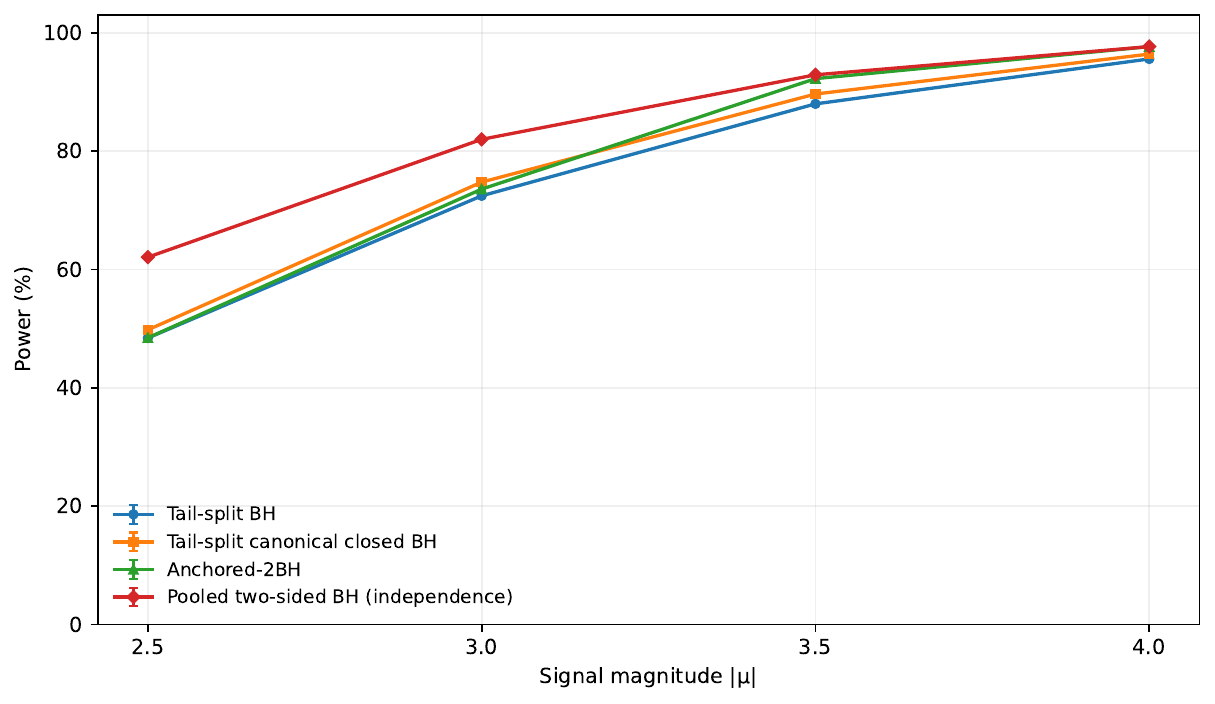}
\caption{Anchored-2BH: independent Gaussian tests with 90\% nonnulls and balanced signs. The ordinary pooled two-sided BH curve has its independence guarantee here; it is not used as a valid benchmark under unrestricted dependent two-sided p-values.}\label{fig:twobh}
\end{figure}
\FloatBarrier

\subsection{Anchored-BY: arbitrary dependence, with and without Gaussianity}\label{sec:sim-by}
\paragraph{Experiment 1: non-Gaussian antithetic p-values.}
For each independent pair draw $U_j\sim\Unif(0,1)$ and set its underlying uniforms to $(U_j,1-U_j)$. Leave null coordinates unchanged and transform nonnull coordinates by $u\mapsto u^{1/\kappa}$, giving $\Beta(\kappa,1)$ alternatives. We use $\kappa\in\{0.05,0.10,0.20\}$ and nonnull fractions $0.1,0.5,0.9$. The antithetic construction is not supplied with a PRDS or Gaussian guarantee. Hence Anchored-BY uses \emph{unboosted} fixed p-to-e hinge profiles: the normal-quantile transform in a marginal calibrator is not a Gaussian assumption on the joint data. There are nine 5,000-repetition nonnull settings and one 10,000-repetition global null, totaling 55,000 data sets.

Comparators are BY and the Section~5.3 improvement of closed BY in \citet{SunWang2026}. At 90\% nonnulls and $\kappa=0.05$, powers are 78.05\% for BY, 80.53\% for improved closed BY and 83.01\% for Anchored-BY. The paired gain over improved closed BY is 2.49 points (2.43--2.55). The ordering reverses for weaker alternatives: at $\kappa=0.10$, the anchor has power 62.87\% versus 63.93\% for improved closed BY; at $\kappa=0.20$, the comparison is 30.70\% versus 34.94\%. The unboosted anchor and BY coincide at all three sparse settings. Under the global null their estimated FDR is 0.84\%; the largest FDR estimate across nonnull settings is below 1\% for every displayed procedure.

\paragraph{Experiment 2: arbitrary-sign Gaussian covariances.}
We retain the two-sided Gaussian suite with balanced signals of magnitude $2.5,3,3.5,4$, nonnull fractions $0.1,0.5,0.9$, and four covariance choices: independence, equicorrelation $0.5$, AR(1) with parameter $-0.6$, and $0.8ss^\top+0.2I$ with alternating signs in $s$. The 48 nonnull settings use 5,000 repetitions each and four global nulls use 10,000, giving 280,000 data sets. BY, improved closed BY and the unboosted hinge anchor remain valid; the Gaussian capped-sum lemma additionally permits the calibrated hinge anchor. Ordinary BH and canonical closed BH on two-sided p-values are diagnostic outside independence.

At independent magnitude 3.5 and 90\% nonnulls, powers are 78.68\% for BY, 84.61\% for improved closed BY and 89.99\% for the calibrated anchor. The paired gain over improved closed BY is 5.38 points (5.30--5.47). At magnitude 2.5, however, improved closed BY has power 42.69\% versus 32.56\% for the anchor (\cref{fig:gaussian}). The dense magnitude-3.5 comparison also favors the anchor under negative AR(1), 90.03\% versus 84.73\%, and under the signed factor, 86.32\% versus 84.31\%. Sparse settings mostly retain the baseline. Anchoring does not remove the cost of BY's harmonic correction.

In the independent Gaussian global-null setting, BY and its anchors have estimated FDR 0.99\%. The largest incremental gain from calibration over the same unboosted anchor is only 0.049 percentage points in this Gaussian suite. Thus the experiments show configuration-dependent gains from score design and anchoring, not large realized gains from the cap factor itself.

\begin{figure}[htbp]
\centering\includegraphics[width=.6\textwidth]{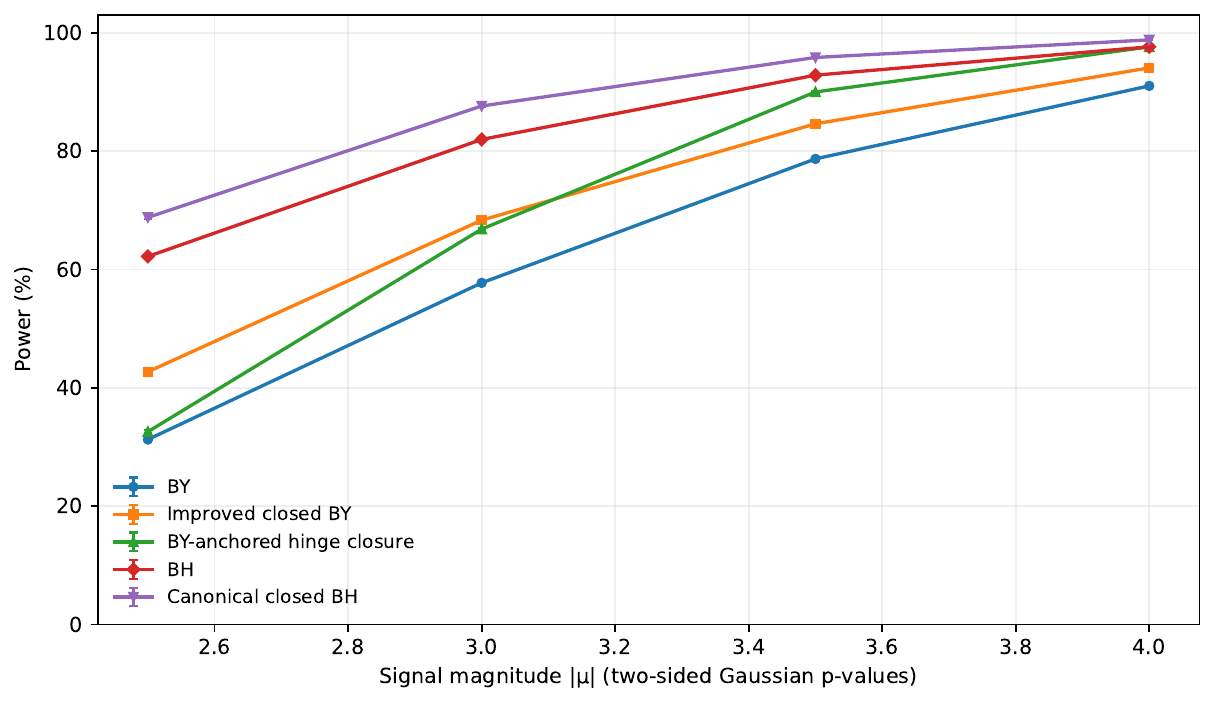}
\caption{Anchored-BY: independent two-sided Gaussian p-values, 90\% nonnulls. All plotted methods are valid here. 
}\label{fig:gaussian}
\end{figure}
\FloatBarrier

\FloatBarrier

\subsection{Anchored-dBH: dense mixtures and a weak signal}\label{sec:sim-dbh}
These experiments use the completed first-round dBH$_{0.9}$ procedure of \citet{FithianLei2022}, including any required randomized pruning. The original BH q-value is calibrated separately along each Gaussian conditional path using \eqref{eq:dbhintegral}. The covariance matrix is supplied to dBH; the additional hinge scores require no covariance estimates or new conditional integrals. Both score choices were fixed before evaluation: the existing four-profile mixture from Appendix~\ref{app:calibration}, with boundaries at fractions $0.1,0.3,0.6$ of $m$, and the single normalized hinge $h_{1.5}$. They are evaluated as separate procedures, never selected by taking the better observed result. Both use unboosted means.

\paragraph{Experiment 1: signal fraction and arbitrary-sign correlation.}
Take $m=40$ and two-sided Gaussian p-values. The covariance models are independent blocks of size ten with correlation $0.5$; their sign-conjugated version with alternating signs; and AR(1) with parameter $-0.5$. Nonnull fractions are $0.2$ and $0.9$, with balanced means of magnitude $2.5,3,3.5$. There are 2,000 repetitions in each nonnull setting and 5,000 for each covariance's global null: 21 settings and 51,000 data sets. Ordinary two-sided BH is only a diagnostic comparator here; dBH and its anchors have the stated known-covariance guarantee. 

\begin{table}[htbp]
\centering\small
\caption{Anchored-dBH: power (percent) with 90\% nonnulls, $m=40$, and known covariance. Each row uses 2,000 repetitions. ``Mixture'' and ``hinge'' are separate fixed auxiliary scores.}\label{tab:dbhpower}
\begin{tabular}{@{}lrrrr@{}}
\toprule
Covariance & $|\mu|$ & dBH$_{0.9}$ & Mixture anchor & $h_{1.5}$ anchor\\
\midrule
Positive blocks & 2.5 & 61.75 & 61.75 & 63.28 \\
Positive blocks & 3 & 82.00 & 82.01 & 86.54 \\
Positive blocks & 3.5 & 92.77 & 92.80 & 96.11 \\
Signed blocks & 2.5 & 61.56 & 61.56 & 63.13 \\
Signed blocks & 3 & 81.84 & 81.85 & 86.41 \\
Signed blocks & 3.5 & 92.60 & 92.62 & 95.84 \\
Negative AR(1) & 2.5 & 61.59 & 61.59 & 63.29 \\
Negative AR(1) & 3 & 81.65 & 81.66 & 86.15 \\
Negative AR(1) & 3.5 & 92.86 & 92.91 & 96.18 \\
\bottomrule
\end{tabular}
\end{table}

The single lower-threshold hinge yields visible gains in the dense settings, whereas the previous mixture typically adds little. At mean magnitude $3$ in positive blocks, its power gain is 4.54 percentage points, with paired 95\% Monte Carlo interval 4.39--4.69. Estimated FDR rises from 0.501\% to 0.789\%, so these are common-nominal-level comparisons, not matched-realized-FDR comparisons. Under the three global nulls the largest estimate for dBH and either anchor is 4.98\%, with ordinary Monte Carlo uncertainty. In all 18,000 sparse-mixture data sets, both anchored outputs are exactly the dBH output. Thus the inexpensive augmentation helps when the baseline discovers many signals, not uniformly across sparsity levels.

\paragraph{Experiment 2: one weaker hypothesis among strong discoveries.}
Take $m=20$ with 18 means of magnitude six and balanced signs, one weaker mean in $\{1.5,2,2.5,3\}$, and one true null. Use positive and signed blocks as above, with 2,000 repetitions per setting: eight settings and 16,000 data sets. Report rejection probability of the weaker hypothesis, not only average power over nearly certain strong discoveries. At weak mean $2$ in positive blocks, dBH rejects it with probability 53.30\%; the frozen-mixture anchor gives 61.35\% and the $h_{1.5}$ anchor 66.25\%. The latter paired gain is 12.95 points (11.48--14.42). The corresponding signed-block probabilities are 50.15\%, 58.40\%, and 63.25\%. Across this heterogeneous grid, the largest estimated full-family FDR for the two anchors is 0.624\%.

\begin{figure}[htbp]
\centering\includegraphics[width=.7\linewidth]{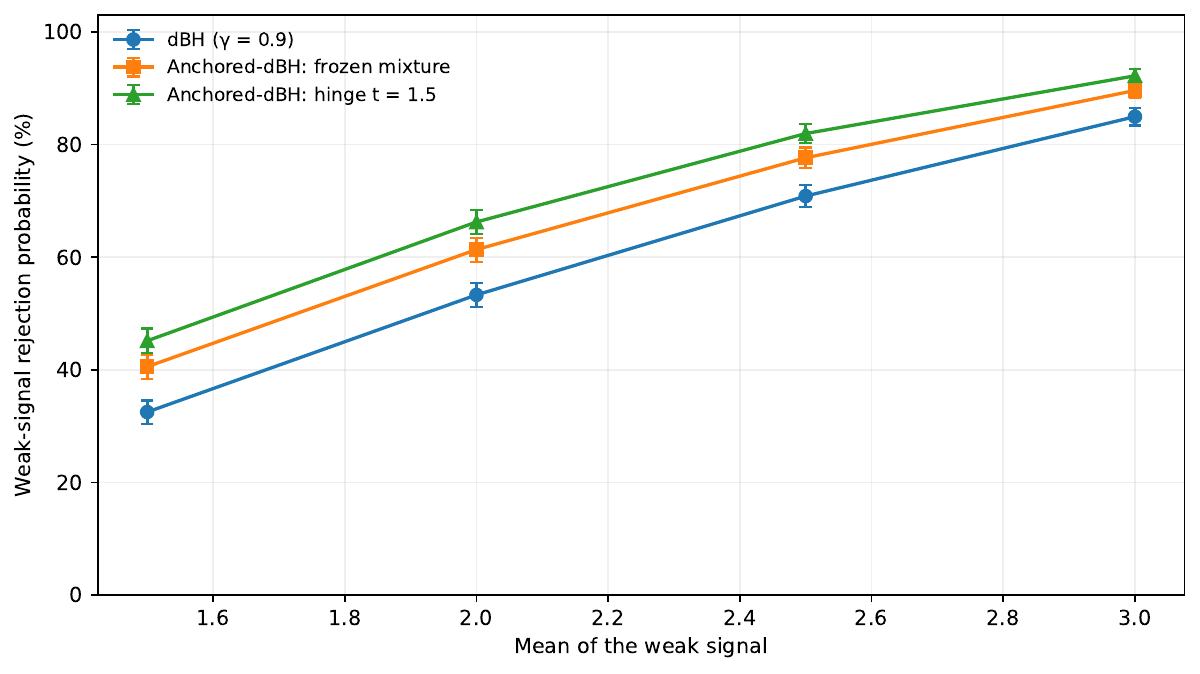}
\caption{Anchored-dBH: rejection probability of one weak signal among 18 strong signals and one true null, with positive correlation blocks ($m=20$). Error bars are pointwise 95\% Monte Carlo intervals; each setting has 2,000 repetitions. The signed-block experiment is reported in full in the supplement.}\label{fig:dbhweak}
\end{figure}

\paragraph{Computation and validation.}
All 67,000 data sets use the exact three-list search over certified supersets, rather than a candidate grid. Randomized pruning was needed on six data sets and was performed before augmentation; every realized baseline set was retained. During the simulation runs, the median dBH baseline cost per data set was 5.44--6.17 milliseconds for $m=40$ block models and 20.36--21.34 milliseconds for negative AR(1). The \emph{two} anchor calls together took median 0.081--0.115 milliseconds per setting. These are measured application times after compilation, not worst-case bounds; baseline cost includes all conditional integrations. The source includes the timing arrays.

The finite-knot conditional integral was compared with independently recomputing BH on every interval in 96 cases; the largest absolute difference was $1.7\times10^{-16}$. Another 1,700 coordinate checks verified the independent, $\gamma=1$ reduction to BH. The arbitrary-baseline closure was checked on 210 exhaustive small problems, involving 36,545 subset checks. These validate the implementation, not the probability assumptions. Normal probabilities, roots and confidence bounds are evaluated in double precision with conservative decision margins; no interval-arithmetic certificate is claimed. We have not compared these anchors with recursive dBH$^{(2)}$, and no dominance over that different refinement is inferred.

\subsection{Ordered testing: selective sets and complete prefixes}\label{sec:sim-ordered}
The two experiments below share one fixed grid of 65,000 data sets across 12 settings; the data are not counted twice. We use $m=40$, $\alpha=0.10$, and independent one-sided Gaussian p-values. In the perfect ordering, the first $1$, $3$, or $12$ coordinates are nonnull with mean $2$, $2.5$, or $3$ (nine settings, 5,000 repetitions each). Two further settings put 12 nonnulls at positions $1,4,\ldots,34$, with means $2.5$ or $3$ (5,000 each). A global null has 10,000 repetitions. Baseline parameters and residual-band scores were fixed before evaluation. No method uses the signal locations or the means; the order is part of the testing problem. All conditional calculations use the exact uniform null law and finite categories from Section~\ref{sec:ordered}.

\paragraph{Experiment 1: Selective SeqStep+ and Adaptive SeqStep.}
We compare Selective SeqStep+ at $s=0.1$ and $0.5$, Adaptive SeqStep at $(s,\lambda)=(0.1,0.5)$, their same-cutoff and rectangle anchors. A same-cutoff anchor preserves the original eligible-prefix structure. The rectangle variant allows a different p-value cutoff through certification, and its extra flexibility is reported separately. Conditional fallback refinements are included in the complete results. These ordered comparisons assess each baseline against its own extensions, not relative to the wider class of ordered-testing procedures.

\begin{table}[htbp]
\centering\small
\caption{Ordered selective procedures: power (\%) at mean $2.5$, with the signals first. There are $m=40$ tests, $\alpha=0.10$, and 5,000 repetitions per column. Same-cutoff and rectangle closures are different search families.}\label{tab:ordered-selective}
\begin{tabular}{@{}lrrr@{}}
\toprule
Procedure & 1 signal & 3 signals & 12 signals\\
\midrule
Selective SeqStep+, $s=0.1$ &10.42&77.49&88.95\\
Same-cutoff residual band &62.46&82.69&88.95\\
Rectangle residual band &62.82&84.10&89.62\\
\addlinespace
Adaptive SeqStep, $(s,\lambda)=(0.1,0.5)$ &15.46&86.99&88.96\\
Same-cutoff residual band &62.04&87.85&88.96\\
Rectangle residual band &62.10&89.77&89.44\\
\bottomrule
\end{tabular}
\end{table}

For one leading signal of mean $2.5$, the same-cutoff gain is 52.04 percentage points over Selective SeqStep+ (95\% interval 50.66--53.42) and 46.58 points over Adaptive SeqStep (45.20--47.96). The respective estimated FDR changes are 5.77\% to 6.47\% and 8.63\% to 9.25\%, at a common nominal 10\%. These are gains in power even when the adaptive baseline's realized FDR is already near the target (\cref{tab:ordered-selective,fig:ordered-single}). With 12 leading signals, both same-cutoff anchors add essentially nothing at $s=0.1$; rectangle gains are modest.

The gain also depends on gaps in the ordering. At mean $2.5$ with 12 signals in positions $1,4,\ldots,34$, Adaptive SeqStep has power 14.12\%, its same-cutoff anchor 17.92\%, and its rectangle anchor 17.97\%. The same-cutoff paired gain is 3.80 points (3.69--3.92). For the more conservative Selective SeqStep+ cutoff $s=0.5$, the corresponding powers are 0.62\%, 5.85\%, and 22.61\%. The last gain partly uses the freedom to certify a smaller selective cutoff; it is not an improvement at the original cutoff alone.

\begin{figure}[htbp]
\centering\includegraphics[width=.6\textwidth]{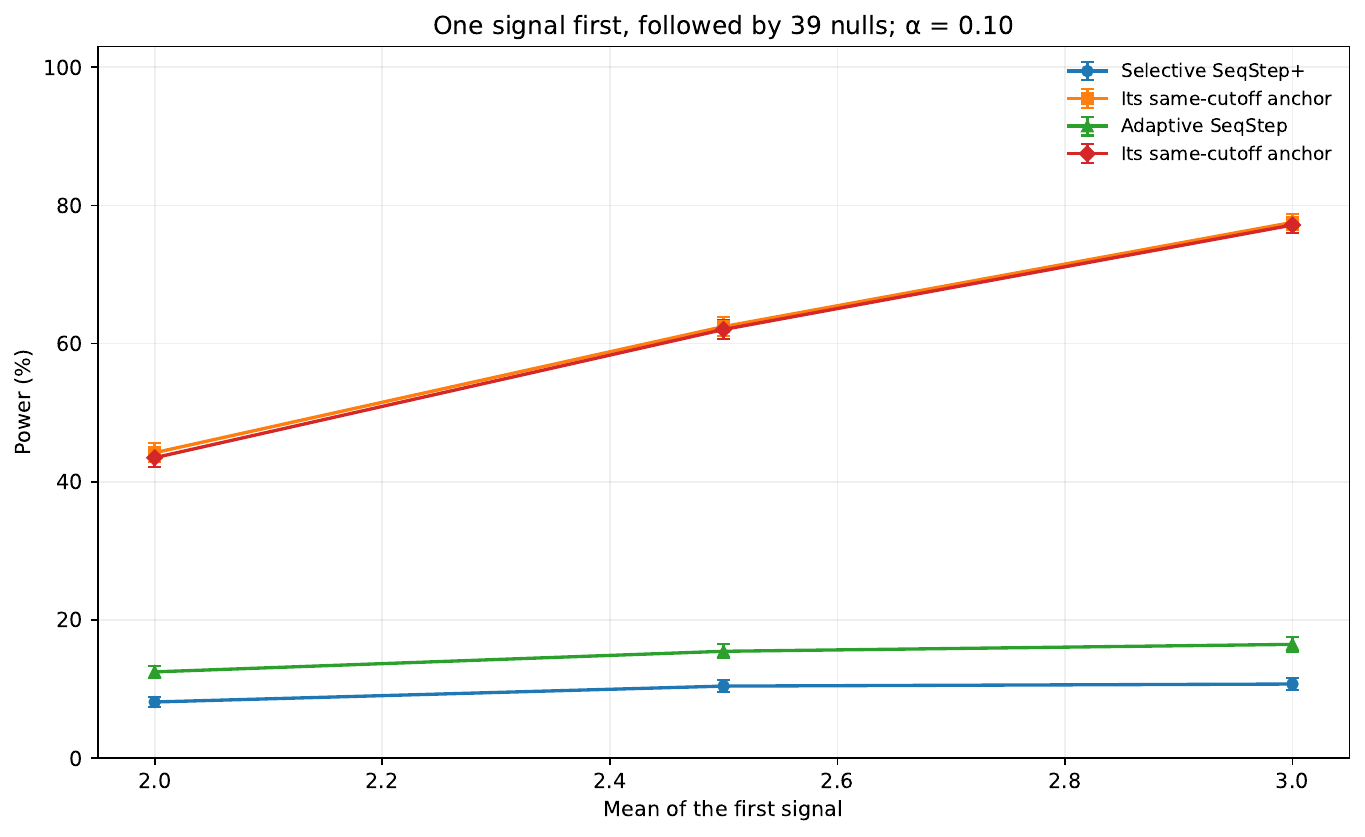}
\caption{One leading signal followed by 39 nulls, $\alpha=0.10$. Both same-cutoff residual-band anchors improve their baseline's minimum-discovery behavior. Error bars are pointwise 95\% Monte Carlo intervals from 5,000 repetitions.}\label{fig:ordered-single}
\end{figure}
\FloatBarrier

\paragraph{Experiment 2: SeqStep+ and complete-prefix refinements.}
At $\alpha=0.1$ and $c=0.5$, SeqStep+ needs at least 19 rejections even with no upper-tail p-values. We compare it with the raw residual prefix closure, a first calibrated prefix, its recycled residual prefix, and a recursive calibrated-prefix comparator, as defined in Section~\ref{sec:orderedprefix}. Every output is a complete prefix containing its own baseline. The first calibrated prefix and raw residual prefix coincide throughout the displayed settings.

\begin{table}[htbp]
\centering\small
\caption{Complete-prefix procedures: power (\%) at mean $2.5$, using the same ordered data as \cref{tab:ordered-selective}. Each refinement is compared with the prefix procedure it retains.}\label{tab:ordered-prefix}
\begin{tabular}{@{}lrrr@{}}
\toprule
Procedure & 1 signal & 3 signals & 12 signals\\
\midrule
SeqStep+, $c=0.5$ &0.00&0.00&0.62\\
Raw residual / first calibrated prefix &49.12&29.05&12.84\\
Recycled residual prefix &49.16&36.33&22.84\\
Recursive calibrated prefix &49.16&36.36&23.11\\
\bottomrule
\end{tabular}
\end{table}

For three leading signals, recycling increases power from 29.05\% to 36.33\%, closely tracking the recursive comparator at 36.36\%. With gaps at mean $2.5$, raw and recycled prefix powers are about 4.32\% and 4.33\%; recycling adds little there. Thus the ordered experiments include genuine improvements over the original and first-calibrated baselines, but also settings with no additions and little benefit from further recycling.

\FloatBarrier

\section{Discussion}
Anchoring uniformly improves many/most existing FDR procedures by exploiting looseness in their proofs and adapting to the unknown set of true null hypotheses, even under dependence. Intersection-level residuals yield the BH-like constructions; conditional coordinate residuals yield the ordered procedures. Both retain the realized baseline FDR certificate and discoveries. The coefficient $1-|S|/m$ is not universal: every extension needs its own budget proof. Conditional anchoring is related to conditional calibration~\citep{FithianLei2022,LeeRen2024}, whereas donation uses a different allocation mechanism \citep{XuFischerRamdas2026}.


There are some limitations to the improvement of ``naive'' anchoring. For example, if an unboosted averaged-score BH anchor adds $q$ discoveries to $b$ baseline discoveries, with $0<b<m$, then
\(
 q<\frac{b^2}{m-b}\) and so \(b(b+1)\le m\ \Longrightarrow\ \bar B=B.
\)
(To see this, examine e-closure for the set $A^c$. Appendix~\ref{app:sparse} proves the bound and gives a jointly calibrated refinement.) In contrast, ordered residuals instead use information discarded by categorical decisions and can help even from an empty baseline, although their gains depend on the ordering and search structure. We were not interested in designing the optimal improvement (if one exists), just presenting a simple, broadly applicable, and computationally efficient method.

In some cases, for example when the nulls are independent uniform p-values, anchoring can be recursively applied (i.e.\ repeated again on the anchored procedure) to exploit the gap between the targeted and realized thresholds, and this can lead to further gains, likely of lower magnitude. We do not pursue this point here.

\FloatBarrier
\phantomsection
\addcontentsline{toc}{section}{References}
\begingroup
\small
\setlength{\bibsep}{2pt}

\endgroup
\clearpage
\appendix

\section{Exact computational shortcuts}\label{sec:computation}
We first analyze Anchored-BH and Anchored-BY through the common closure \eqref{eq:anchor}, not e-BH applied to the same scores. Suppose evidence is ordered so that $x_{n,1}\ge\cdots\ge x_{n,m}$ for every $n$, and the baseline is the prefix $B=[b]$. Both BH and BY have this form when scores decrease in p-values, including ordinary two-sided BH in \cref{cor:twosided}. The Anchored-2BH baseline in \cref{cor:twotail} need not be a prefix and is excluded from the monotone-scan claim. Break ties by original index. The cases $b=0$ and $b=m$ return the baseline immediately.

\subsection{From all subsets to smallest-sum constraints}
By the cap observation in Section~\ref{sec:cap}, use the uncapped scaled means for certification and put
\[
 a_n=\alpha\left(\frac1n-\frac1m\right)\eta_n.
\]
For a candidate $R$ of size $r\ge b$, define adjusted scores
\begin{equation}\label{eq:cost}
 c_i(n,R)=a_nx_{n,i}+\frac{\ind\{i\in B\}}b-
                              \frac{\ind\{i\in R\}}r.
\end{equation}
Its intersection constraints of size $n$ are equivalent to nonnegativity of the sum of the $n$ smallest adjusted scores. There are only $m-1$ such checks, since the $n=m$ sum equals zero.

\begin{proposition}[Prefix reduction]\label[proposition]{prop:prefix}
If all score profiles share the evidence ordering, the largest certified superset of a prefix baseline can be chosen as a prefix $R=[r]$.
\end{proposition}
\begin{proof}
Consider a selected nonbaseline coordinate with smaller score $u$ and an omitted one with larger score $v$. Exchanging their selection changes their two adjusted values from $(a_nu-1/r,a_nv)$ to $(a_nu,a_nv-1/r)$. The sum is fixed, the minimum does not decrease, and the maximum does not increase. Thus no sum of the $n$ smallest adjusted scores decreases. The same exchange works for every $n$. Repeating yields a certified prefix of the same cardinality.
\end{proof}
This optimizes among supersets of $B$, not among all rejection sets that might discard baseline discoveries. It uses ordering, not convexity of the underlying score.

\subsection{A monotone scan and quadratic time}\label{sec:quadratic}
Assume there are $J$ contiguous intersection-size blocks within which $x_{n,i}=x_i$ is fixed, and $a_n$ is nonincreasing in $n$ within each block. Unboosted fixed profiles meet this condition. So do the calibrated profiles constructed in Appendix~\ref{app:calibration}, and the data-adaptive scores of Sections~\ref{sec:pilots} and \ref{sec:adaptivehinges} after they have been evaluated.

Fix $r\ge b$, write $q=r-b$, and let $n\le m-b$. Every baseline adjusted score is at least every nonbaseline adjusted score: the former has a weakly larger raw score and a nonnegative offset. A least-favorable intersection therefore avoids $B$ and contains $j$ of the weakest additions plus $n-j$ of the weakest unselected coordinates. Its score sum is
\begin{equation}\label{eq:Gnj}
 G_n(j)=a_n\left\{\sum_{i=r-j+1}^r x_i+
                    \sum_{i=m-n+j+1}^m x_i\right\}-\frac jr,
\end{equation}
where $\max\{0,n-(m-r)\}\le j\le\min(n,q)$. Its forward difference is
\begin{equation}\label{eq:forward}
 G_n(j+1)-G_n(j)
 =a_n\{x_{r-j}-x_{m-n+j+1}\}-\frac1r.
\end{equation}
This difference is nondecreasing in $j$. For fixed feasible $j$, it is nonincreasing in $n$ within a profile: the nonnegative score gap decreases, as does $a_n$. Consequently a minimizing $j$ never decreases as $n$ increases. Prefix sums and one advancing pointer find all minima in $O(m)$ operations per profile.

For $n>m-b$, a smallest-cost intersection contains all nonbaseline coordinates and the weakest $t=n-(m-b)$ baseline coordinates. Its sum is simply
\begin{equation}\label{eq:largen}
 a_n\sum_{i=b-t+1}^m x_i+\frac tb-\frac{q+t}{r},
\end{equation}
which is available in constant time per $n$ from prefix sums.

\begin{theorem}[Quadratic closure]\label[theorem]{thm:quadratic}
Under the common-order and profile conditions above, a candidate prefix is certified in $O(Jm)$ time. The maximum-cardinality certified prefix superset is found in $O(Jm^2+m\log m)$ time and $O(Jm)$ memory, after evaluating the profiles. For fixed $J$, this is quadratic time and linear memory.
\end{theorem}
\begin{proof}
Use \eqref{eq:forward} to advance the minimizing pointer through at most $m$ positions in each profile block, then use \eqref{eq:largen}. Check candidate sizes downward until the first certified prefix is found. There are at most $m-b$ candidates. \Cref{prop:prefix} proves exactness of the restriction. The sort costs $O(m\log m)$.
\end{proof}

Two qualifications prevent overstating this shortcut. First, feasibility need not be monotone in $r$, so generic binary search or stopping at the first failed addition is not valid. Second, intersection-specific scores depending on the identities of $S$, rather than only its size, need a different analysis. For a non-prefix baseline, the adjusted scores still split into three ordered lists, but the simpler monotone scan above does not apply automatically.

Cheap necessary conditions can reduce the search. In particular, with $n=m-b$ and
\[
 L=\frac{\alpha(1-n/m)\eta_n}{n}\sum_{i\notin B}x_{n,i},
\]
the intersection $S=B^c$ requires $q/(b+q)\le L$. If $L<1$, at most $\lfloor bL/(1-L)\rfloor$ additions are possible. Singleton checks provide another inexpensive screen.

\subsection{A pilot-BH anchor in sorting time}\label{sec:pilotfast}
For \eqref{eq:pilotanchor}, let $B=[b]\subseteq C=[c]$. A certified addition must belong to $C\setminus B$, since every singleton outside $C$ has zero evidence. Write $R=[b+q]$, $r=b+q$.

\begin{theorem}[Two-count characterization]\label[theorem]{thm:pilot}
Let $0<b<m$, $\alpha\le\beta<1$, and $1\le q\le c-b$. The candidate $R=[b+q]$ is certified by \eqref{eq:pilotanchor} if and only if
\begin{align}
 \alpha r(c-q)&\ge\beta c(q+m-c),\label{eq:pilotone}\\
 \frac q r&\le\frac{\alpha b(c-b)}{\beta c(m-b)}.\label{eq:pilottwo}
\end{align}
Both conditions become no easier as $q$ increases. The largest feasible $q$, with $q=0$ always allowed, is found by binary search. The whole procedure has $O(m\log m)$ time and $O(m)$ memory, including sorting and the two BH runs.
\end{theorem}
\begin{proof}
Let $d=m-c$, $n_1=d+q$, $n_2=m-b$, and
$u_n=\alpha(m-n)/(\beta cn)$. Scores outside $C$ are zero, added coordinates have adjusted score $u_n-1/r$, and baseline coordinates have score $u_n+1/b-1/r$.

The intersection comprising all additions and all coordinates outside $C$ gives $u_{n_1}\ge1/r$, exactly \eqref{eq:pilotone}. It makes every adjusted score nonnegative for $n\le n_1$. For $n_1\le n\le n_2$, the minimum sum is
\[
 u_n(n-d)-q/r
 =\frac\alpha{\beta c}\{m+d-n-md/n\}-q/r,
\]
a concave function of $n$. Its endpoint conditions are \eqref{eq:pilotone} and \eqref{eq:pilottwo}, the latter corresponding to $S=B^c$. For $n_2\le n\le m$, add $(n-n_2)(1/b-1/r)$; concavity remains, and the final endpoint is zero at $n=m$. This proves sufficiency and necessity. After expansion, the left side minus right side of \eqref{eq:pilotone} has a negative derivative in $q$ for $\beta\ge\alpha$; the left side of \eqref{eq:pilottwo} increases in $q$. Thus binary search applies in this special case.
\end{proof}

For example, with $m=100$, $\alpha=0.05$, $\beta=0.10$, eighty p-values equal to $0.0001$ and twenty equal to $0.06$, BH rejects 80 and the pilot rejects 100. Equations~\eqref{eq:pilotone}--\eqref{eq:pilottwo} certify all 100 at target $0.05$. This is a pointwise example, not an average-power claim.

\subsection{Arbitrary baselines and weighted size envelopes}\label{sec:arbitrarybasecompute}
The prefix-baseline restriction can be removed without enumerating intersections. Let $A$ be any realized baseline of size $b>0$, and suppose the local collection has the form
\[
 E_S=\frac{\FDP_S(A)}\alpha+d_{|S|}\frac1{|S|}\sum_{i\in S}x_{|S|,i},
 \qquad d_n\ge0,
\]
where the profiles share a common ordering. For Anchored-dBH, $d_n=1-n/m$. The size-envelope shifted anchors in Appendix~\ref{app:shifted} have a different, nonincreasing $d_n$.
For a candidate superset $R\supseteq A$ of size $r$, the adjusted scores are
\begin{equation}\label{eq:arbitrarycost}
 c_i(n,R)=\frac{\alpha d_n}{n}x_{n,i}
       +\frac{\ind\{i\in A\}}b-\frac{\ind\{i\in R\}}r.
\end{equation}
For every $n$, its $n$ smallest entries must sum to a nonnegative number.

The exchange proof of \cref{prop:prefix} only exchanges nonbaseline hypotheses. Therefore, even when $A$ is not a prefix, a maximum-cardinality certified superset can be chosen by adding a prefix of the nonbaseline hypotheses in decreasing score order. Fix such a candidate. Its adjusted scores split into three internally sorted lists: baseline discoveries have offset $1/b-1/r$; additions have offset $-1/r$; unselected coordinates have offset zero. Each list is multiplied by the same nonnegative coefficient $\alpha d_n/n$.

\begin{proposition}[Three-list shortcut]\label[proposition]{prop:threelist}
For $J$ common-order score profiles and any baseline mask, all intersection constraints for one candidate can be checked in $O(m\log m)$ time after profile prefix sums have been formed. Searching all prefix additions gives a maximum-cardinality certified superset in
\[
 O\{Jm+m\log m+(m-b)m\log m\}\text{ time},
 \qquad O(Jm)\text{ memory}.
\]
These bounds are after evaluating the profiles and computing the baseline.
\end{proposition}
\begin{proof}
To sum the $n$ smallest entries of three sorted lists, let $k$ be the number still needed and $p\le3$ the number of nonempty lists. Probe $\lfloor k/p\rfloor\vee1$ entries into each list, truncated at its end. The prefix ending at the smallest probe can be removed: every other list has fewer than its probe position entries strictly below that probe, so the removed entries belong to a valid choice of the $k$ smallest. Add their sum from prefix sums. Each iteration either removes a constant fraction of the remaining rank or exhausts one of the three lists. Thus the calculation takes $O(\log(n+1))$ steps. Repeat over $n$, then over candidate sizes; the exchange argument proves that the candidate family suffices.
\end{proof}
This is exact for the specified size-residual collection, not a claim that every identity-dependent weighted collection has the same complexity. For a prefix baseline and nonincreasing $\alpha d_n/n$ within profile blocks, the stronger quadratic scan in Appendix~\ref{sec:quadratic} still applies. Generalized common shifts preserve the original p-value ordering, so their size-envelope anchors fall in that case. Coordinate-specific shifts, dBH, and confidence-bound shifts can require the three-list version.

\FloatBarrier

\section{Supporting results and implementation details}\label{app:support}
This appendix collects the score specification, comparator definitions, pointwise comparisons, sparse-regime results, and implementation checks. Detailed training logs, calibration arrays, and simulation records remain in the reproducibility package.

\subsection{Score specification and calibration}\label{app:calibration}
\Cref{tab:weights} specifies the frozen hinge mixtures used in the simulations. For dimensions other than $100$, the block boundaries are at fractions $0.1,0.3,0.6$ of $m$. The profiles were fitted externally using independent Gaussian training data at $m=100$, $\alpha=0.05$, nonnull fractions $0.1,0.5,0.9$, signal magnitudes $2.5,3,3.5$, and 2,000 repetitions per scenario. The 953-candidate search maximized average $\log\{0.02+\widehat{\mathrm{power}}\}$ for unanchored, unboosted closure; the weights were frozen before evaluation and are not claimed optimal for an anchor. The search history and numerical factors are in the source package.
\begin{table}[htbp]
\centering\small
\caption{Frozen Gaussian-hinge weights. Normal-quantile transforms give the corresponding fixed p-value calibrators in the PRDS experiments.}\label{tab:weights}
\begin{tabular}{@{}lrrrrr@{}}
\toprule
Intersection size ($m=100$) & $h_{1.5}$ & $h_2$ & $h_{2.5}$ & $h_3$ & $h_{3.25}$\\
\midrule
1--10 &0.50&0&0.50&0&0\\
11--30&0.15&0.70&0&0&0.15\\
31--60&0.15&0&0.60&0&0.25\\
61--100&0&0.05&0.45&0.10&0.40\\
\bottomrule
\end{tabular}
\end{table}

\paragraph{Monotonicity needed for the quadratic algorithm.}
For a fixed profile $f$, the mean of $n+1$ independent observations is the average of its $n$-observation leave-one-out means. Jensen's inequality therefore gives $\overline f_{n+1}\le_{\rm cx}\overline f_n$. Since $C_n=1/\{\alpha(1-n/m)\}$ increases with $n$,
\[
 n\longmapsto\E\min\{\eta\overline f_n,C_n\}
 \quad\text{is nondecreasing for each fixed }\eta.
\]
Thus the maximal feasible calibration factor is nonincreasing within a fixed-profile block. Taking the running minimum of conservatively feasible factors preserves validity and the monotonicity used in \cref{thm:quadratic}; factor one is always analytically valid. This argument, rather than the numerical fit, is the part needed by the computational theorem.

\paragraph{Numerical implementation.}
Score values are rounded upward to spacing $0.005$ before independent-sum FFT convolution; overflow above 8,000 contributes the full cap. Factors use a $10^{-6}$ expectation margin and the blockwise decreasing envelope, reverting to one when no boost is certified. Halving the spacing to $0.0025$ changed power by at most $0.0112$ percentage points in 18 checked settings. These are conservative discretization calculations in double precision, not interval-arithmetic certificates of FFT error. Exact validity requires \eqref{eq:calibration}; unboosted and adaptive-pilot variants avoid this numerical integration.

\subsection{Comparator definitions}\label{app:comparators}
We record the defining rules to distinguish the comparators from generic e-closure labels. For canonical closed BH, let $r_{k,s}$ denote the horizon matrix of \citet[Section~5]{Goeman2026}. Initialize $r_{0,s}=0$, $r_{k,m}=m$ for $k>0$, and $r_{k,s}=k$ for $k+s\le m+1$, except the preceding $s=m$ rule. For the remaining entries with $k<m$, set
\begin{align*}
 d&=m-s,\\
 B_{k,s}&=\begin{cases}
 \displaystyle\min\left\{m,\frac{d(k+s-m-1)r_{k-1,s}}{(k+s-m)d-r_{k-1,s}}\right\},
       &(k+s-m)d>r_{k-1,s},\\
 m,&\text{otherwise},
 \end{cases}\\
 r_{k,s}&=\min\left\{\lfloor B_{k,s}\rfloor,
                         \left\lfloor\frac{ks}{m-k}\right\rfloor\right\},
 \qquad r_{m,s}=m.
\end{align*}
The critical values and cumulative horizons are
\[
 A_{k,s}=\begin{cases}
 \alpha k/s,&r_{k,s}=k,\\
 \alpha(k+s-m)r_{k,s}/\{(r_{k,s}+s-m)s\},&r_{k,s}>k,
 \end{cases}
 \qquad
 v_{k,s}=\max_{j\le k}r_{j,s}\ind\{p_{(j)}\le A_{j,s}\}.
\]
The rejection count is $\max\{k:v_{k,s}\ge k\text{ for all }s\}$. Integer floors in the horizon table are evaluated by exact integer arithmetic. The implementation matches the published seven-hypothesis table and worked example.

For closed BY, for each intersection size $n$ let $H_n=\sum_{k=1}^n1/k$ and
\[
 f_n^{\rm CBY}(p)=\frac{\ind\{H_np\le\alpha\}}
                         {\max\{1,\lceil nH_np/\alpha\rceil\}}.
\]
Its e-collection is $\alpha^{-1}\sum_{i\in S} f_n^{\rm CBY}(p_i)$. We compute the largest certified p-value prefix, and then apply the Section~5.3 improvement of \citet{SunWang2026}: reject all $m$ when
\[
 p_{(m-1)}\le\frac\alpha{mH_m},\qquad
 \alpha<p_{(m)}\le\frac{5\alpha}{3}.
\]
Otherwise retain closed BY. Near-equality reciprocal-score sums are checked with exact common-denominator arithmetic, avoiding spurious failures at identities such as $1/2+1/3+1/6=1$. This comparator's supplied exact-arithmetic implementation is restricted to $m\le100$; the anchored algorithms are not.

\subsection{Pointwise incomparability examples}\label{app:incomparability}
These are deterministic checks of the stated rules, not additional simulation settings. They show that ``incomparable'' means lack of pointwise containment in either direction for concrete fixed variants, not that every choice of auxiliary evidence is incomparable.

For Anchored-BH, fix $m=7$, $\alpha=0.05$, and the single pilot $\beta=0.10$. The two-count rule of \cref{thm:pilot} and the canonical closed-BH horizon table give
\begin{center}\small
\begin{tabular}{@{}p{.56\textwidth}rrr@{}}
\toprule
Sorted p-values & BH & Anchored-BH & Closed BH\\
\midrule
$(0.0001,0.011,0.011,0.011,0.071,0.071,0.071)$ &4&5&4\\
$\alpha(2/7,2/7,7/12,7/12,1,15/8,2)$ &2&2&6\\
\bottomrule
\end{tabular}
\end{center}
The second p-value vector is the worked example of \citet{Goeman2026}. The pilot makes seven rejections in both cases. The first case permits one addition; in the second, \cref{prop:sparsebarrier} already forbids an addition.

For Anchored-BY, fix $m=10$, $\alpha=0.05$, and use the single unboosted score
$g(p)=h_{1.5}\{\Phi^{-1}(1-p/2)\}$ for every intersection size. With notation $x^{\times k}$ for $k$ repeated entries,
\begin{center}\small
\begin{tabular}{@{}p{.54\textwidth}rrr@{}}
\toprule
Sorted p-values & BY & Anchored-BY & Improved closed BY\\
\midrule
$(0.0001^{\times8},0.055,0.055)$ &8&10&8\\
$(0.0001,0.0037,0.9^{\times8})$ &1&1&2\\
\bottomrule
\end{tabular}
\end{center}
The Section~5.3 override does not trigger in either example. In the first, unmodified closed BY cannot reject either p-value above $\alpha$, whereas all Anchored-BY intersection constraints hold. In the second, direct closure certifies the first two for closed BY, while the complement of the single BY discovery blocks an addition by the hinge anchor. These witnesses were checked against exhaustive subset enumeration for the anchored rules and the defining comparator algorithms. The corresponding script and output are included in the supplement.

\subsection{Sparse-regime limits and a jointly calibrated refinement}\label{sec:sparse}\label{app:sparse}
Let $b=|B|$ for $B=\BH_\alpha(p)$, with $0<b<m$, and suppose $\bar B$ adds $q>0$ hypotheses. The constraint at $S=B^c$ is necessary even though $B^c$ is random: closure checks \emph{every} subset pointwise. It reads
\begin{equation}\label{eq:sparsecomplement}
 \frac{q}{b+q}\le\frac{\alpha b}{m}F_{B^c}.
\end{equation}
The strong BH discoveries disappear entirely from this constraint. A sparse gain must be paid for by evidence in the nonrejected complement, not by arbitrarily large scores of the existing discoveries.

\begin{proposition}[A small-baseline barrier for averaged scores]\label[proposition]{prop:sparsebarrier}
Suppose that, at each intersection size, the auxiliary e-value is an unboosted convex mixture of common-threshold averages
\[
 F_S=\sum_{\ell}w_{|S|,\ell}
       \frac1{|S|}\sum_{i\in S}
       \frac{\ind\{p_i\le T_{|S|,\ell}(p)\}}{T_{|S|,\ell}(p)},
 \qquad w_{|S|,\ell}\ge0,\quad\sum_\ell w_{|S|,\ell}=1,
\]
allowing integrals, a constant-score component and zero-threshold scores set to zero. Then every certified superset of $B$ with $q>0$ additions obeys
\begin{equation}\label{eq:sparsebarrier}
 q<\frac{b^2}{m-b}.
\end{equation}
In particular, if $b(b+1)\le m$, then $\bar B=B$. The bound is deterministic; validity of the thresholds is needed for FDR control, not for this inequality.
\end{proposition}
\begin{proof}
For every $k>b$, the definition of BH gives $p_{(k)}>\alpha k/m$. Fix any realized threshold $t>0$, let $c=|\{i:p_i\le t\}|$, and put $n=m-b$. If $c\le b$, its average score on $B^c$ is zero. Otherwise
\[
 \frac{c-b}{nt}
 <\frac{m(c-b)}{\alpha cn}
 \le\frac1\alpha.
\]
A constant-score component is also strictly below $1/\alpha$. Thus the convex mixture satisfies $F_{B^c}<1/\alpha$. Substitution into \eqref{eq:sparsecomplement} gives $q/(b+q)<b/m$, equivalent to \eqref{eq:sparsebarrier}.
\end{proof}

The proposition covers fixed normalized decreasing p-to-e calibrators, because they are mixtures of normalized indicators; it therefore includes the normal-quantile Gaussian hinges. It also covers pilot-BH indicators, the adaptive hinge, and averaged adaptive indicators borrowed from closed BH. At $m=100$, an unboosted rule in this class makes no additions whenever BH makes at most nine discoveries, regardless of how its mixture weights are tuned. This helps explain the sparse results in Section~\ref{sec:simulations}. It is a bound in terms of the observed $b$, not a claim that every sparse alternative yields $b\le9$.

A deterministic boost $\eta_{m-b}$ changes the necessary bound to
\[
 \frac q{b+q}<\frac{\eta_{m-b}b}{m}.
\]
Thus even one addition requires $\eta_{m-b}>m/\{b(b+1)\}$. For example, $m=100,b=5$ requires a factor strictly exceeding $10/3$. The factors near one in the fitted large-intersection profiles cannot remove this obstruction. The result does \emph{not} rule out other local e-values or a different valid enlargement of BH's certificates.

\paragraph{Borrowing closed-BH thresholds.}
For the canonical horizon rule, write $r^{\rm CBH}(p)$ for its rejection count and $a_{k,n}$ for its critical values, as in Appendix~\ref{app:comparators}. Define
\[
 k_n(p)=\max\bigl(\{k\le r^{\rm CBH}(p):p_{(k)}\le a_{k,n}\}\cup\{0\}\bigr),
 \qquad T_n(p)=a_{k_n(p),n},
\]
with $a_{0,n}=0$. The established monotonicity of these thresholds \citep{Goeman2026} makes
\begin{equation}\label{eq:borrowedhorizon}
 F_S^{\rm horizon}
   =\frac1{|S|}\sum_{i\in S}
       \frac{\ind\{p_i\le T_{|S|}(p)\}}{T_{|S|}(p)}
\end{equation}
valid under PRDS. It can be inserted into the anchor, or mixed with a smoother adaptive hinge. This borrows the intersection-specific rejection scale without using closed BH as the baseline. It is a concrete valid candidate, not an asserted improvement over the existing anchors. In particular, \cref{prop:sparsebarrier} still applies: merely transplanting its indicator scores cannot repair the small-$b$ regime.

\paragraph{Changes that address the obstruction.}
A first option, when independently supplied covariates or prespecified scientific priorities are available, is to replace the uniform within-intersection average by
\[
 F_S=\sum_{i\in S}v_{S,i}x_i,
 \qquad v_{S,i}\ge0,\quad\sum_{i\in S}v_{S,i}=1,
\]
where the $x_i$ are valid individual e-values and the weights are fixed independently of the testing data (or justified conditionally on side information). This retains the original BH anchor and its theorem but can concentrate the auxiliary evidence on a few plausible nonnulls. Unlike uniform averaging, it is not covered by \cref{prop:sparsebarrier}. Whether it helps depends on the information in the weights; arbitrary same-data ranking does not justify them.

A second option is to calibrate the new evidence \emph{jointly} with BH's certificate, rather than requiring the added band to have unit mean by itself. The following construction reallocates evidence that the original auxiliary score would spend on already-rejected coordinates.

\begin{proposition}[A jointly calibrated band anchor]\label[proposition]{prop:bandanchor}
Let $T_0=\alpha|B|/m$, and fix multipliers $\kappa_n\ge m/n$ for $1\le n<m$. For $T_0>0$, define
\begin{equation}\label{eq:bandanchor}
 E_S^{\rm band}=E_S^B+
 \left(1-\frac nm\right)
 \frac{\sum_{i\in S}\ind\{T_0<p_i\le\kappa_nT_0\}}
      {n(\kappa_n-1)T_0},\qquad n=|S|<m.
\end{equation}
Set these values to zero when $T_0=0$, and retain $E_{[m]}^{\rm band}=E_{[m]}^B$. Under the PRDS assumptions of \cref{thm:anchor}, this is a valid e-collection pointwise dominating $(E_S^B)$. Its certified supersets of $B$ therefore retain BH and control FDR, simultaneously with $B$.
\end{proposition}
\begin{proof}
Put
\[
 w_n=\frac nm-\frac{1-n/m}{\kappa_n-1}\in[0,1].
\]
Algebra rewrites \eqref{eq:bandanchor} as
\[
 E_S^{\rm band}=\frac1n\sum_{i\in S}\left[
 w_n\frac{\ind\{p_i\le T_0\}}{T_0}
 +(1-w_n)\frac{\ind\{p_i\le\kappa_nT_0\}}{\kappa_nT_0}\right].
\]
Both thresholds are coordinatewise nonincreasing functions of the p-values, so \eqref{eq:adaptiveindicator} proves validity. The weights are deterministic and nonnegative because $\kappa_n\ge m/n$. In \eqref{eq:bandanchor} the additional term is nonnegative, establishing pointwise domination and the closure guarantee.
\end{proof}

The band average in \eqref{eq:bandanchor} is \emph{not} claimed to be a separately valid $F_S$. Consequently this is a jointly calibrated extension of Anchored-BH, not an application of \cref{thm:anchor} to an unverified auxiliary e-value. The PRDS proof calibrates the total collection. It remains readily interpretable: BH's contribution is unchanged, and the new contribution is zero for already-rejected hypotheses.

This extension can escape \cref{prop:sparsebarrier} in small examples. At $m=6$, $\alpha=0.05$ and
\[
 p=(0.0001,0.0001,0.0256,0.9,0.9,0.9),
 \qquad \kappa_n=\max\{1.55,m/n\},
\]
BH rejects two hypotheses and the band anchor certifies the first three. Every unboosted averaged-score anchor in \cref{prop:sparsebarrier} must retain exactly two, since $b(b+1)=m$. The example was checked by exhaustive intersection enumeration. It is not a broad power comparison. Indeed, for $b<m/2$ the same complement argument gives the weaker but still restrictive necessary bound $q<b^2/(m-2b)$ for this band family: if $c$ p-values fall below $\kappa_{m-b}T_0$, then $c-b<(\kappa_{m-b}-1)b$. Thus the extension changes the bound without eliminating the difficulty at very small $b$.

Broader joint calibration could use closed-BH's local horizon evidence $Q_S$ through
\begin{equation}\label{eq:overlapanchor}
 E_S^{\rm overlap}=\max\{E_S^B,\gamma_SQ_S\},\qquad
 \sup_{P\in H_S}\E_P E_S^{\rm overlap}\le1.
\end{equation}
The maximum preserves $B$. The generic sum bound only gives $\gamma_S=1-|S|/m$; a real advantage requires a larger factor justified by overlap of the two certificates. The band construction provides one explicit form of joint accounting, not a solution for arbitrary adaptive $Q_S$.

Finally, simultaneous retention itself uses some of the budget. For any anchored output, \eqref{eq:simulfdr} implies
\[
 \E\bigl[\{\FDP_{I_0}(\bar B)-\FDP_{I_0}(B)\}_+\bigr]
 \le\alpha-\FDR(B).
\]
Under independent continuous uniform nulls, BH has $\FDR(B)=\alpha m_0/m$, so the right side is $\alpha(m-m_0)/m$. This is not a power upper bound, but it explains why sparse improvements that preserve simultaneous certification can be demanding. Canonical closed BH does not in general retain that same post hoc choice of the baseline.

\subsection{Implementation checks}\label{app:validation}
\paragraph{Anchored-2BH implementation.}
The new two-tail comparator runs the same canonical closed-BH horizon rule at $\alpha/2$ separately in the two tails, then takes their union. The Anchored-2BH baseline is also a union and can be a nonprefix set in the two-sided p-value ordering. Its implementation therefore uses the smallest-sum criterion following \eqref{eq:cost} directly for each candidate, with the baseline represented by a coordinate mask. Additions are ordered by increasing two-sided p-value; full baseline containment is checked, not merely the rejection count.

Validation is separate from the simulation study. Canonical closed BH was checked against its definition on 350 random small problems. The quadratic anchor was compared to full smallest-sum sorting on 720 problems and to exhaustive enumeration of all candidate supersets and intersections on 140 small problems. The pilot shortcut passed 4,930 count-level checks. Closed BY passed 210 exhaustive rational-arithmetic checks. All pointwise baseline-retention and calibrated-versus-unboosted invariants also passed across the simulations. The generic-baseline algorithm passed 815 exhaustive candidate comparisons on 210 problems with $m=2,\ldots,8$. Such checks test implementation consistency, not the probabilistic assumptions.

\subsubsection{Ordered procedures, validation, and provenance}\label{app:orderedvalidation}
The ordered study uses the definitions of SeqStep+ in \citet{BarberCandes2015} and Adaptive SeqStep in \citet{LeiFithian2016}. Selective SeqStep+ is the $\lambda=s$ specialization of the adaptive selective rule. All p-values are exact continuous uniforms under the null; the fixed-order assumption is not replaced by an order chosen using the tested p-values. Raw uncorrected SeqStep rules, with their different modified-FDR targets, are not ordinary-FDR baselines in these experiments.

The conditional fallback comparator solves the one-coordinate budget inequality in Section~\ref{sec:orderedprefix} and uses its stronger coordinate certificate to select the largest prefix or rectangle retaining the baseline. The recursive prefix comparator replaces the first calibration's lower denominator by the first calibrated prefix before repeating the conditional calculation. It is an ordered conditional-calibration comparator, not a claim that the published SeqStep+ procedure already includes that recursion. The residual-band method instead uses the actual conditional FDP contribution and \eqref{eq:ordband}. Both preserve their specified baseline; neither is an arbitrary filter of a merely FDR-valid set.

The imported validation covers 252 exact category-enumerated ordered null configurations, with largest expected total budget divided by $\alpha$ equal to $0.999999$, and 100 exhaustive additive-closure checks. The current package separately audits all 65,000 ordered replicate arrays against their archived counts and checks baseline retention and prefix structure. The proof is not replaced by these finite checks. Category probabilities are evaluated in double precision; a small residual margin prevents optimistic use of rounding error, but the calculations are not interval-arithmetic certified.

The data and results are the ordered subset of the completed conditional-residual study, not new simulations in this revision. The 72,000 calibrated-knockoff data sets from that separate study are not imported or included in any total. Shared finite-state source kernels are preserved for provenance, but the supplied ordered driver runs only the 12 ordered settings. All retained-method rejection masks and numerical results are unchanged. The study contains 1,278,000 distinct data sets in total, including the appendix experiments; the ordered grid is counted only once.
\FloatBarrier

\section{Anchored-NDBH under weak negative dependence}\label{sec:negative}\label{app:negative}
Negative dependence permits a different baseline and different auxiliary evidence. We improve the FDR-corrected BH procedure of \citet{CRW2024}, not uncorrected BH at level $\alpha$. The residual coefficient is derived from their bound with the null-proportion factor retained. Unlike the PRDS anchor, a sharpened full-intersection coefficient can also certify discoveries when the baseline is empty.

\subsection{The corrected baseline and its intersection budgets}
Assume valid null p-values satisfy the equal-threshold condition
\begin{equation}\label{eq:weaknegative}
 \Prob(p_i\le t\text{ for every }i\in J)\le t^{|J|}
 \quad\text{for all nonempty }J\subseteq I_0,\quad 0\le t\le1.
\end{equation}
This is the weak-negative-dependence condition sufficient for the Simes bounds used below. It does not require independence between nulls and nonnulls. Pairwise negative correlation alone is not a substitute. One-sided Gaussian p-values with nonpositive null--null covariances satisfy the condition; two-sided folding is not covered by that implication \citep{CRW2024}.

Write $H_m=\sum_{j=1}^m j^{-1}$ and
\[
 \Psi_m(u)=u\min\{H_m,\,3.18+\log(1/u)\},\qquad 0<u<1.
\]
Choose $\beta$ so that $\Psi_m(\beta)=\alpha$, and let $B^-=\BH_\beta(p)$. Theorem~19 of \citet{CRW2024} gives $\FDR(B^-)\le\alpha$. The harmonic branch is preferable at smaller dimensions; for example, with $\alpha=0.05$, $\beta=\alpha/H_{100}\approx0.0096388$, whereas at $m=5000$ the logarithmic branch gives $\beta\approx0.0060307>\alpha/H_{5000}\approx0.005498$.

For $1\le n<m$, set
\begin{equation}\label{eq:negativebudget}
 q_n=\frac{\Psi_m(\beta n/m)}{\alpha}.
\end{equation}
For the full intersection, use the sharper coefficient
\begin{equation}\label{eq:negativefull}
 C_m(\beta)=\min\left\{1,\sum_{k=1}^m\binom mk
              \left(\frac{\beta k}{m}\right)^k\right\},\qquad
 q_m=\min\{1,C_m(\beta)/\alpha\}.
\end{equation}
These coefficients are known and lie in $[0,1]$.

\begin{proposition}[Negative-dependence contribution bound]\label[proposition]{prop:negativebudget}
Under \eqref{eq:weaknegative}, every true intersection $S$ satisfies
\[
 \E\FDP_S(B^-)\le\alpha q_{|S|}.
\]
\end{proposition}
\begin{proof}
Let $n=|S|$, $u=\beta n/m$, and $Q_S=\min_{1\le k\le n}np_{(k:S)}/k$. If $v=|B^-\cap S|>0$ and $r=|B^-|$, BH self-consistency gives $Q_S\le n\beta r/(mv)$, hence
\[
 \FDP_S(B^-)\le\min\{u/Q_S,1\}.
\]
Writing $G_S(x)=\Prob(Q_S\le x)$, integration by parts yields
\[
 \E\min\{u/Q_S,1\}=u+u\int_u^1\frac{G_S(x)}{x^2}\,dx.
\]
This is the FDR-linking calculation of \citet{Su2018}. Applying the Simes bound and the integral estimate in the proof of Theorem~19 of \citet{CRW2024}, now at effective level $u$, gives $u\{3.18+\log(1/u)\}$. Separately, the arbitrary-dependence BH contribution bound gives $uH_m$. Their minimum proves \eqref{eq:negativebudget}. At $n=m$, the FDP equals $\ind\{B^-\ne\varnothing\}$, so the finite-dimensional Simes bound of \citet[Theorem~2]{CRW2024} gives \eqref{eq:negativefull} directly.
\end{proof}
The $q_n$ calculation is a consequence of the cited proofs; it is not stated there as an anchored procedure. Keeping the effective level $\beta n/m$ is important: BH's PRDS coefficient $n/m$ is not automatically valid in this setting.

\subsection{Local count and product e-values}
For fixed $t\in(0,1)$ let $N_S(t)=\sum_{i\in S}\ind\{p_i\le t\}$. For $n=|S|$ and $1\le k\le n$,
\begin{equation}\label{eq:negativecount}
 F_{S,k,t}=\frac{\binom{N_S(t)}k}{\binom nk t^k}
\end{equation}
is a local e-value: expand the numerator over $k$-subsets and apply \eqref{eq:weaknegative}. The case $k=1$ is the average indicator score; larger $k$ reward simultaneous small p-values. A smoother product is
\begin{equation}\label{eq:negativeproduct}
 F_{S,t,\lambda}=\prod_{i\in S}
       \left(1-\lambda+\lambda\frac{\ind\{p_i\le t\}}t\right),
 \qquad 0\le\lambda\le1.
\end{equation}
Expanding this product and using \eqref{eq:weaknegative} gives expectation at most $(1-\lambda+\lambda)^n=1$. This is a common-threshold application of the negative-dependence product principle in \citet{CRW2024}. Prespecified mixtures over thresholds, orders and product parameters remain valid, as do size-dependent choices.

These products should not be confused with independent-null cap boosting. For $p_1=U$, $p_2=1-U$ and $g(p)=2\ind\{p\le1/2\}$, one has $\E\min\{g(p_1)+g(p_2),2\}=2$, compared with $1.5$ for independent uniforms. Thus \cref{lem:prds} cannot be reused under negative dependence.

\subsection{The Anchored-NDBH guarantee}
For any valid local e-values $F_S$ under \eqref{eq:weaknegative}, define
\begin{equation}\label{eq:negativeanchor}
 E_S^-=\frac{|B^-\cap S|}{\alpha(|B^-|\vee1)}
                +(1-q_{|S|})F_S.
\end{equation}
\begin{corollary}[Anchored-NDBH]\label[corollary]{cor:negativeanchor}
Under \eqref{eq:weaknegative}, the values in \eqref{eq:negativeanchor} are local e-values. Every certified superset $\bar B^-\supseteq B^-$ obtained by e-closure satisfies
\[
 \bar B^-\supseteq\BH_\beta(p)\quad\text{pointwise},\qquad
 \FDR(\bar B^-)\le\alpha.
\]
\end{corollary}
\begin{proof}
By \cref{prop:negativebudget}, $\E E_S^-\le q_{|S|}+(1-q_{|S|})\E F_S\le1$. The first term certifies the baseline, and the actual-null constraint proves the FDR bound, exactly as in \cref{thm:anchor}.
\end{proof}
The cap relevant to an auxiliary score is now $1/\{\alpha(1-q_n)\}$ when $q_n<1$, but capping does not itself justify rescaling. No rescaling is used in the negative-dependence experiments.

The full-intersection refinement can be consequential. Take $m=2$, $\alpha=0.05$, $\beta=1/30$, and $p=(0.03,0.039)$. Then $B^-=\varnothing$, $q_1=1/2$ and $q_2=31/45$. With $t=0.04$ and the full indicator product $F_S=t^{-|S|}\prod_{i\in S}\ind\{p_i\le t\}$, the singleton constraints for rejecting both hypotheses have right sides $\alpha E_{\{i\}}^-=5/8>1/2$, and the full-intersection right side is $175/18>1$. Thus both hypotheses are certified although the corrected baseline is empty. This strict pointwise example is not a claim of a large detection gain in the simulation grid.

\subsection{Count-score certificates under negative dependence}\label{sec:negativecompute}
The product scores in this appendix are not additive, so \cref{thm:quadratic} is not invoked for them. A separate exact reduction is available. Order p-values increasingly and write the corrected baseline as $B^-=[b]$. For a candidate prefix $R=[r]\supseteq B^-$ and an intersection size $n$, a least-favorable intersection uses the smallest possible number $a=(n-m+b)_+$ of baseline discoveries. Replacing a baseline coordinate by a nonbaseline one both decreases the auxiliary evidence and weakens its baseline certificate. Among intersections with this $a$ and $j$ added discoveries, a least-favorable choice is
\[
 S_{n,j}=\{b-a+1,\ldots,b\}\cup\{r-j+1,\ldots,r\}
          \cup\{m-(n-a-j)+1,\ldots,m\},
\]
where empty ranges are omitted and
\[
 \max\{0,n-a-(m-r)\}\le j\le\min\{r-b,n-a\}.
\]
It suffices to check
\begin{equation}\label{eq:negativecompute}
 \min_j\left\{\frac{a}{b\vee1}
       +\alpha(1-q_n)F_{S_{n,j}}-\frac{a+j}{r}\right\}\ge0
 \quad(n=1,\ldots,m).
\end{equation}
All threshold counts are obtained from the ordered p-values. For fixed $n$, these counts are maxima of affine functions of $j$ on the feasible range. Normalized factorial-count moments and the product in \eqref{eq:negativeproduct} are increasing and discretely convex in the count. Consequently, the objective in \eqref{eq:negativecompute} is discretely convex in $j$, and its minimum can be found by binary search on forward differences.

With a fixed number $K$ of threshold/product components, a candidate costs $O(Km\log m)$ time after sorting, without subset enumeration. A prespecified finite collection of $M$ candidate sizes costs $O(MKm\log m)$; exhaustive search over prefix sizes has the corresponding $O(Km^2\log m)$ bound. The large-dimensional experiment checks a finite size grid and returns the largest passing candidate in that grid, not a claimed global maximizer of closure. The small-dimensional experiment searches every prefix. Every returned set is checked against all intersection constraints through \eqref{eq:negativecompute}, and the baseline is always an available output. Unlike the original BH anchor, the full-intersection constraint is checked even when $b=0$.

\subsection{Experiments: antithetic p-values and negative Gaussian blocks}\label{sec:sim-negative}
We compare the corrected baseline $B^-=\BH_\beta$, BY, an indicator-score anchor, and a product-score anchor. Both anchors use the sharpened budgets \eqref{eq:negativebudget}--\eqref{eq:negativefull} and the fixed threshold set $\mathcal T=\{0.01,0.05,0.2\}$. For $n=|S|$, their auxiliary e-values are respectively
\begin{equation}\label{eq:negativeexpscores}
 F_S^{\mathrm{ind}}=\frac13\sum_{t\in\mathcal T}\frac{N_S(t)}{nt},\qquad
 F_S^{\mathrm{prod}}=\frac13\sum_{t\in\mathcal T}
      \prod_{i\in S}\left(1-\frac1n+\frac{\ind\{p_i\le t\}}{nt}\right).
\end{equation}
These are direct instances of \eqref{eq:negativecount} and \eqref{eq:negativeproduct}; no Gaussian cap calibration, estimated covariance, or evaluation-data tuning is used. The dimension-dependent product parameter is the prespecified choice $\lambda_n=1/n$. A separate small exploratory check preceded the reported evaluation and is retained in the supplement.

\paragraph{Experiment 1: non-Gaussian antithetic pairs.}
With $m=100$, generate independent pairs $(U_j,1-U_j)$ and transform nonnull coordinates by $u\mapsto u^{1/\kappa}$. Use nonnull fractions $0.1,0.5,0.9$ and $\kappa\in\{0.05,0.10,0.20\}$. The null subvector satisfies \eqref{eq:weaknegative}. Each nonnull setting has 5,000 repetitions and the global null has 10,000, totaling 55,000 data sets. Here the corrected baseline is exactly BY, so the experiment isolates the value of the negative-dependence scores and budgets rather than a change of baseline. Improved closed BY is also evaluated.

At 90\% nonnulls, product-anchor powers are 86.34\%, 67.71\% and 32.43\% for the three increasing values of $\kappa$, compared with 77.80\%, 59.16\% and 30.49\% for the baseline. Improved closed BY gives 80.33\%, 63.99\% and 34.75\%, respectively. Thus the product anchor is better at the two stronger alternatives but worse at the weakest. At $\kappa=0.10$, its paired gain over its baseline is 8.56 percentage points (8.39--8.73). It makes no additions in the sparse settings and virtually none at 50\% nonnulls; the indicator anchor can be preferable there. The largest estimated FDR for either anchor is 1.00\%, attained at the global null. No empty-baseline augmentations occur in this grid, despite their theoretical possibility.

\begin{figure}[htbp]
\centering\includegraphics[width=.86\textwidth]{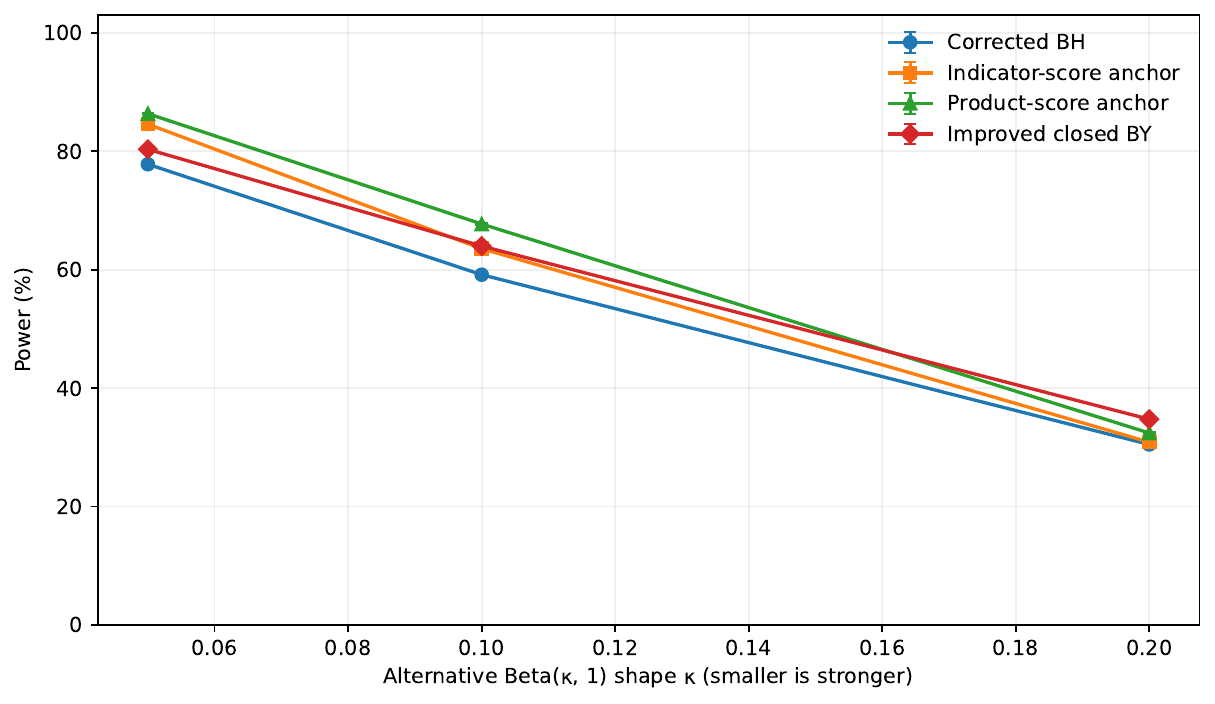}
\caption{Anchored-NDBH with $m=100$ antithetic p-values and 90\% nonnulls. The corrected BH baseline equals BY here. Smaller $\kappa$ means stronger alternatives; the gain over improved closed BY is not uniform.}\label{fig:negative-small}
\end{figure}
\FloatBarrier

\paragraph{Experiment 2: the non-harmonic branch at $m=5000$.}
Generate one-sided p-values from unit-variance Gaussian statistics, either independently or in independent blocks of size ten with every within-block correlation equal to $-0.1$. Nonnull means are positive with magnitude $2.5,3,3.5$, at fractions $0.1,0.5,0.9$. The nonpositive Gaussian covariances imply \eqref{eq:weaknegative} \citep{CRW2024}. Each nonnull setting has 1,000 repetitions; each of the two global nulls has 5,000, giving 28,000 data sets across 20 settings. Now $\beta\approx0.0060307$ exceeds BY's $\alpha/H_{5000}\approx0.005498$. Raw BH is only a diagnostic in the negative-block setting.

At 90\% nonnulls under negative blocks, the product anchor has powers 80.46\%, 91.25\% and 98.22\%, compared with 34.50\%, 61.26\% and 81.08\% for the corrected baseline (\cref{tab:negative-large,fig:negative-large}). At magnitude $3$, the paired gain is 29.99 points (29.95--30.04). The indicator anchor also gains power, but is substantially weaker at the lower signal magnitudes. Independent data give similar comparisons. The largest nonnull FDR estimate of the product anchor in this suite is 1.33\%, still below the nominal 5\% target. Sparse gains remain negligible; at 10\% nonnulls the product anchor adds discoveries in only 33 of 6,000 data sets. At 50\% nonnulls it adds discoveries in all 6,000 data sets, but the extent varies with signal strength. Neither anchor makes an empty-baseline augmentation in the evaluated large-dimensional settings.

\begin{table}[htbp]
\centering\small
\caption{Negative Gaussian blocks, $m=5000$, 90\% nonnulls: power (\%). All displayed procedures have their stated guarantee. The large-dimensional anchors search a finite grid of certified prefix supersets.}\label{tab:negative-large}
\begin{tabular}{@{}rrrrr@{}}
\toprule
Signal mean & BY & Corrected BH & Indicator anchor & Product anchor\\
\midrule
2.5 &32.91&34.50&39.40&80.46\\
3.0 &59.76&61.26&81.04&91.25\\
3.5 &80.11&81.08&96.83&98.22\\
\bottomrule
\end{tabular}
\end{table}

The comparisons are at a common nominal level, not matched realized FDR. For example, in the negative-block magnitude-$3$ dense setting the baseline's estimated FDR is about 0.06\%, while the product anchor's is about 0.61\%. The large power gains use some of that conservativeness; they do not prove that the product mixture is optimal. Canonical closed BH is not a valid negative-dependence comparator, and the supplied exact closed-BY implementation is limited to $m\le100$, so it is not extrapolated to this experiment.

\begin{figure}[htbp]
\centering\includegraphics[width=.86\textwidth]{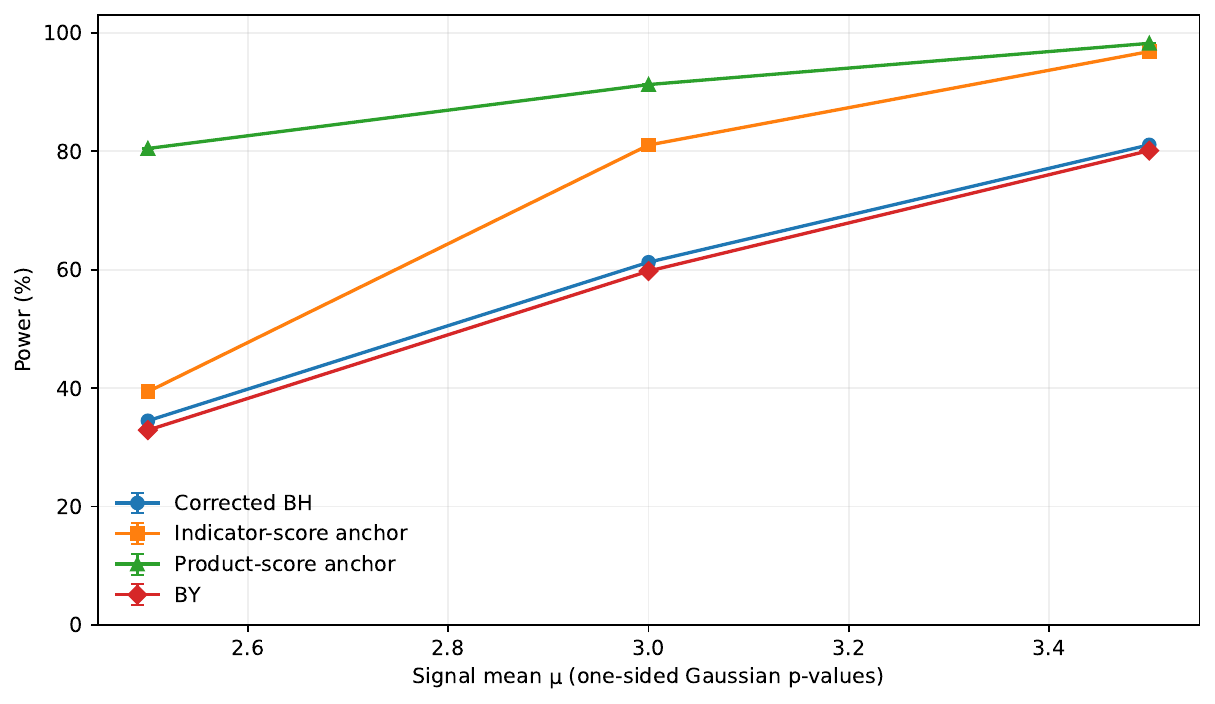}
\caption{Anchored-NDBH with $m=5000$, negative Gaussian blocks and 90\% nonnulls. Error bars are pointwise 95\% Monte Carlo intervals from 1,000 repetitions. The product score pools joint evidence without independent-null cap boosting.}\label{fig:negative-large}
\end{figure}
\FloatBarrier

\paragraph{Computation and checks.}
At $m=100$ we search all prefix supersets. At $m=5000$ we check at most 20 prespecified increment sizes, threshold cutoffs and an inexpensive necessary-condition upper bound; every candidate is tested using \eqref{eq:negativecompute}. This is not the additive quadratic algorithm and is not a claim to maximize over every rejection set. The certificate reduction agrees with exhaustive subset enumeration in 2,541 candidate checks across 280 problems of dimensions $2$--$9$. Baseline retention holds on every replicate. Budget evaluation is conservative to a numerical margin but not interval-arithmetic certified. All seeds, candidate grids, the separate exploratory check and replicate counts are supplied.

\FloatBarrier

\section{Anchoring shifted BH and its generalizations}\label{app:shifted}
This appendix consolidates the shifted-method extensions and their experiments. The published paper of \citet{SarkarZhang2025}, including its Appendix~A, is the basis for SBH-1 and SBH-2. The generalized and confidence-bound constructions are those of \citet{GhoshSarkar2025,GhoshSarkar2026}. The anchored e-collections below are deductions from their coordinatewise proofs, not theorems attributed to those papers. We distinguish the full weighted collection from an efficiently computable size-envelope version.

\subsection{Coordinate budgets and the published shifted procedures}
Suppose a baseline $A$ has deterministic coordinate bounds
\begin{equation}\label{eq:shiftcoordinate}
 \E\frac{\ind\{i\in A\}}{|A|\vee1}\le v_i\quad(i\in I_0),
 \qquad v_i\ge0,\qquad\sum_{i=1}^m v_i\le\alpha.
\end{equation}
Put $w_i=v_i/\alpha$ and $w(S)=\sum_{i\in S}w_i$. Then, construct
\begin{equation}\label{eq:weightedshiftanchor}
 E_S^A=\frac{\FDP_S(A)}\alpha+[1-w(S)]F_S.
\end{equation}
For any local e-values $F_S$, these are local e-values, since the two expectation budgets add to one. Every certified superset of $A$ retains all of its discoveries and controls FDR. The coefficient is generally not $1-|S|/m$.

Consider the published two-sided $z$ model $Z\sim N_m(\mu,\Sigma)$ with known positive definite correlation matrix. Define
\[
 \tau_i=\frac1{(\Sigma^{-1})_{ii}},\qquad
 G_a(u)=2\overline\Phi\!\left\{\sqrt a\,\Phi^{-1}(1-u/2)\right\},\quad a>0.
\]
SBH-1 applies step-up critical values $rc$, $r=1,\ldots,m$, to
\begin{equation}\label{eq:shiftedp}
 \widetilde p_i=G_{\tau_i}^{-1}(p_i)
 =2\overline\Phi(|Z_i|/\sqrt{\tau_i}),
 \qquad \sum_{i=1}^mG_{\tau_i}(c)=\alpha.
\end{equation}
Theorem~1(i) and its published proof give, coordinate by coordinate,
\begin{equation}\label{eq:sbhweights}
 v_i=G_{\tau_i}(c),\qquad w_i=G_{\tau_i}(c)/\alpha.
\end{equation}
To verify that no information is lost by reading the bound before summation, condition on $Z_{-i}$. Under $H_i$, $Z_i/\sqrt{\tau_i}$ is conditionally normal with variance one and a possibly nonzero mean. The conditional left-tail ratio
$\Prob(\widetilde p_i\le u\mid Z_{-i})/u$ is nonincreasing in $u$. In the leave-one-out FDR identity, replace the term at $rc$ divided by $r$ by the term at $c$. Averaging gives
$\E[\ind\{i\in A\}/(|A|\vee1)]\le\Prob(\widetilde p_i\le c)=G_{\tau_i}(c)$,
which is \eqref{eq:shiftcoordinate}. This is the published proof's reusable ingredient.

SBH-2 instead takes a common $s\in(0,\lambda_{\min}(\Sigma))$, applies $G_s^{-1}$ to every p-value, and uses $c=G_s^{-1}(\alpha/m)$. Theorem~1(ii) gives $v_i=\alpha/m$, so its anchor has the uniform residual $1-|S|/m$.
Our implementation follows the strict range in the published Definition~2, using $s=0.99\lambda_{\min}(\Sigma)$; under exact independence it uses the stated BH limit.\footnote{For positive equicorrelation, the displayed formula in the published Section~5.1 gives $\tau_i=(1-\rho)\{1+(m-1)\rho\}/\{1+(m-2)\rho\}>1-\rho=\lambda_{\min}$. Thus the equivalence asserted there between SBH-1 and SBH-2 cannot be used together with Definition~2's strict common-shift range. We use the formal definitions and the displayed conditional-variance formula separately, rather than imposing that equivalence.}

The same weighted argument applies to the published common-scale $t$ model. Replace $G_a$ by
$H_{a,\nu}(u)=2\overline T_\nu\{\sqrt a\,T_\nu^{-1}(1-u/2)\}$.
The scenario-2 proof in their Appendix~A retains the corresponding coordinate bounds. Auxiliary evidence must be valid for the $t$ model, for example a marginal p-to-e calibrator. No joint-Gaussian cap comparison is assumed for $t$ statistics.

\subsection{Generalized common shifts and auxiliary evidence}
For the generalized procedure of \citet{GhoshSarkar2025}, choose a prespecified common shift $s\in(0,1]$ and apply $G_s^{-1}$ to all p-values. For step-up critical values $ru/m$, the coordinate bounds are
\begin{equation}\label{eq:gsbhweights}
 v_i(s,u)=
 \begin{cases}
 G_s(u/m),&s\le\tau_i,\\[2pt]
 G_{\tau_i}\!\left\{G_{s/\tau_i}(u)/m\right\},&s>\tau_i.
 \end{cases}
\end{equation}
Calibrate $u$ by $\sum_i v_i(s,u)=\alpha$ and substitute $w_i=v_i(s,u)/\alpha$ into \eqref{eq:weightedshiftanchor}. Equation~\eqref{eq:gsbhweights} is the coordinatewise version of their Theorem~2 and Definition~1; the second branch uses their tail-probability minorant when the common shift exceeds a coordinate's conditional variance. Their later tail-dependence formulation retains these bounds \citep{GhoshSarkar2026}. A covariance-based choice of $s$ is allowed when the covariance is known; choosing it after inspecting the test statistics needs an additional validity argument.

Unboosted means of $h_t(Z_i)$ from \eqref{eq:gaussianhinge}, and the frozen mixtures in Appendix~\ref{app:calibration}, provide immediate local evidence. Their unit-null expectations use only marginal normality. Alternatively, applying Section~\ref{sec:coordinate}, a prespecified liberal shifted pilot $A_\gamma$ with positive coordinate bounds $v_i(\gamma)$ gives
\[
 x_i^{(\gamma)}=
 \frac{\ind\{i\in A_\gamma\}}{v_i(\gamma)(|A_\gamma|\vee1)},
 \qquad F_S=|S|^{-1}\sum_{i\in S}x_i^{(\gamma)}.
\]
These are calibrated using the pilot's contribution bounds, not PRDS. The shifted $\widetilde p_i$ themselves are generally not superuniform: under $H_i$ their distribution function is $G_{\tau_i}$, so they must not be passed to an ordinary p-to-e calibrator without correction.

A strict example illustrates the improvement without simulation. With $m=10$, $\alpha=0.05$, $\Sigma=0.5I+0.5\mathbf1\mathbf1^\top$, and $p=(10^{-8},\ldots,10^{-8},0.1)$, SBH-1 rejects nine. Its weights are all $0.1$, and the unboosted mean of $h_{1.5}$ certifies all ten. At the weakest singleton, $\alpha E_{\{10\}}\approx0.11121>1/10$; every proper-intersection constraint has a positive margin. The implementation verifies all $1,023$ nonempty intersections. No ordering relative to ordinary BH is inferred from this example.

\subsection{Independent covariance information and confidence failure}
For clarity, consider the standardized model
\[
 Z\sim N_m(\mu,\Sigma),\quad \Sigma_{ii}=1,\qquad
 D\sim W_m(N,\Sigma),\quad D\perp Z,\quad N\ge m.
\]
Only the off-diagonal covariance is unknown in this implementation; marginal variances are known. This does not solve unknown marginal scales by treating unstandardized observations as $z$ statistics.
Let $\nu=N-m+1$, choose $\epsilon_i>0$ with $\epsilon=\sum_i\epsilon_i<\alpha$, and put
\[
 L_i=\frac{1/(D^{-1})_{ii}}{\chi^2_{\nu,1-\epsilon_i}}.
\]
Then $\Prob(L_i>\tau_i)\le\epsilon_i$. The confidence-bound shifted BH (CBSBH) of \citet{GhoshSarkar2026} uses $G_{L_i}^{-1}(p_i)$ with critical values $rc(D)$, where
$\sum_iG_{L_i}\{c(D)\}=\alpha-\epsilon$.
Let $A_D$ be its final rejection set and $v_i(D)=G_{L_i}\{c(D)\}$.
Conditional on $D$, the coordinate bound
\[
 \E\left[\frac{\ind\{i\in A_D\}}{|A_D|\vee1}\given D\right]
 \le v_i(D)
\]
holds on $\{L_i\le\tau_i\}$ for a true-null $i$. This uses the same left-tail-ratio argument; other coordinates' fixed positive shifts do not need to satisfy their confidence bounds for the $i$th contribution inequality.

\begin{proposition}[Confidence-aware shifted anchor]\label[proposition]{prop:confidenceanchor}
Let $F_S$ satisfy $\E(F_S\mid D)\le1$ under $H_S$, and set
$w_i(D)=\{v_i(D)+\epsilon_i\}/\alpha$. Then
\begin{equation}\label{eq:confidenceanchor}
 E_S^{\rm conf}=\min\left\{
 \frac{\FDP_S(A_D)}\alpha+
 \left[1-\sum_{i\in S}w_i(D)\right]F_S,\;\frac1\alpha\right\}
\end{equation}
is a local e-collection. Its certified supersets retain every CBSBH discovery and control FDR at $\alpha$.
\end{proposition}
\begin{proof}
Put $\epsilon_S=\sum_{i\in S}\epsilon_i$ and
$\mathcal G_S=\bigcap_{i\in S}\{L_i\le\tau_i\}$. On $\mathcal G_S$, the conditional expectation is at most
$\sum_{i\in S}v_i(D)/\alpha+1-\sum_{i\in S}\{v_i(D)+\epsilon_i\}/\alpha
=1-\epsilon_S/\alpha$.
Everywhere the cap bounds the value by $1/\alpha$, while
$\Prob(\mathcal G_S^c)\le\epsilon_S$. Thus its expectation is at most one. The cap cannot remove baseline certification because $\FDP_S(A_D)\le1$. Apply e-closure.
\end{proof}
The confidence-failure allowance is charged explicitly; it cannot be spent again on unbounded evidence. Fixed marginal hinges are conditionally valid because $D$ is independent of $Z$. We use no independent-null cap boosting for these random weighted profiles.

\subsection{An efficient size-envelope version}\label{app:shiftenvelope}
The full weighted collection depends on the identities in $S$, and a generic quadratic-time claim for it would be unwarranted. A useful conservative version avoids that problem. Order the coordinate budgets as $w_{[1]}\ge\cdots\ge w_{[m]}$ and define
\begin{equation}\label{eq:weightenvelope}
 Q_n=\sum_{j=1}^n w_{[j]},\qquad d_n=1-Q_n.
\end{equation}
Replace $1-w(S)$ in \eqref{eq:weightedshiftanchor} by $d_{|S|}$. Since $Q_{|S|}\ge w(S)$, the new values remain valid and retain the baseline. They can be smaller than the full weighted anchor, but coincide with it when all weights are equal. In \eqref{eq:confidenceanchor}, use the same envelope conditional on $D$ and retain the total cap. On $\mathcal G_S$, the expectation is still at most $1-\epsilon_S/\alpha$, so the confidence proof is unchanged.

For unboosted common-order score means, the total cap may be omitted during certification because the left side of every constraint is at most one. Apply \cref{prop:threelist} with the coefficient $\alpha d_n/n$. This gives an exact search over all certified supersets for the \emph{size-envelope procedure} in $O(m^2\log m)$ time with finitely many score profiles, plus the baseline and weight computation. The word ``exact'' does not claim equivalence to the larger identity-specific weighted collection. Common-shift baselines are prefixes of the marginal p-value order; since $d_n/n$ is nonincreasing, their envelope anchors also admit the quadratic scan in Appendix~\ref{sec:quadratic}.

This supplies a practical shortcut for all the shifted families considered here, at an explicit and potentially material price when their weights are heterogeneous. It is not necessary to assume identical conditional variances to obtain a computable improvement. Our experiments use the size envelope throughout and report its no-gain settings too.

\subsection{Simulations for the shifted families}\label{app:shiftsimulations}
Two independent evaluation suites use the published definitions and the weighted size-envelope algorithm above. They add 73,000 data sets across 37 settings. As in Section~\ref{sec:sim-dbh}, the two auxiliary choices are the previously frozen mixture and $h_{1.5}$; both are unboosted, and neither is chosen using the evaluation outcomes. All tests are two-sided with $\alpha=0.05$, known unit marginal variances and $m=40$. These dimensions and covariance types are close to the multivariate-mean experiment in \citet[Section~6.1]{SarkarZhang2025}, but this is a new experiment, not a replication of its figures. In particular, we add a dense-signal regime to test the residual budget.

\paragraph{Experiment 1: known covariance.}
Compare SBH-1, SBH-2 and generalized shifted BH with $s=m^{-1}\sum_i\tau_i$ (GSBH-mean), and both anchors of each. The common-shift choice is a covariance-based candidate discussed by \citet{GhoshSarkar2025}. Covariances are equicorrelation $0.5$; five independent size-eight blocks with correlations $0.2,0.4,0.6,0.8,0.5$; and AR(1) with parameter $-0.5$. Nonnull fractions are $0.2$ and $0.9$, with balanced mean magnitudes $2.5,3,3.5$. Each of 18 nonnull settings uses 1,500 repetitions; each of three global nulls uses 4,000: 39,000 data sets in total. The complete results also report diagnostic ordinary BH and valid BY.

\begin{table}[htbp]
\centering\small
\caption{Shifted baselines and their size-envelope anchors: power (percent) at 90\% nonnulls and mean magnitude $3$. Each row uses 1,500 repetitions. Covariance and all score choices are fixed independently of the testing outcomes.}\label{tab:shiftpower}
\begin{tabular}{@{}llrrr@{}}
\toprule
Covariance & Baseline & Baseline power & Mixture anchor & $h_{1.5}$ anchor\\
\midrule
Equicorrelation $0.5$ & SBH-1 & 61.38 & 67.09 & 77.11 \\
Equicorrelation $0.5$ & SBH-2 & 60.53 & 66.98 & 76.40 \\
Equicorrelation $0.5$ & GSBH-mean & 61.38 & 67.09 & 77.11 \\
Heterogeneous blocks & SBH-1 & 23.13 & 23.13 & 23.13 \\
Heterogeneous blocks & SBH-2 & 48.11 & 63.68 & 62.43 \\
Heterogeneous blocks & GSBH-mean & 56.51 & 61.83 & 64.52 \\
Negative AR(1) & SBH-1 & 64.48 & 67.06 & 74.84 \\
Negative AR(1) & SBH-2 & 53.11 & 64.39 & 68.16 \\
Negative AR(1) & GSBH-mean & 64.66 & 66.99 & 78.42 \\
\bottomrule
\end{tabular}
\end{table}

For equicorrelation $0.5$ and mean magnitude $3$, SBH-1 improves from 61.38\% to 77.11\% with $h_{1.5}$, a paired gain of 15.73 points (15.33--16.12). The mixture gives 67.09\%. SBH-2 also improves, from 60.53\% to 76.40\%. With equal conditional variances, GSBH-mean equals SBH-1; this equality does not identify either with the strict-range SBH-2 rule. At magnitude $3.5$, the SBH-1/$h_{1.5}$ comparison is 79.36\% versus 93.80\%.

Heterogeneous budgets reveal the cost of the envelope: the SBH-1 anchors add nothing throughout the heterogeneous-block nonnull grid. In the same dense covariance setting at magnitude $3$, common-shift SBH-2 improves from 48.11\% to 63.68\% with the mixture; GSBH-mean improves from 56.51\% to 64.52\% with $h_{1.5}$. Thus practical improvements are possible without equal weights, but a generic efficient envelope can leave substantial evidence unused. Sparse-setting power gains are absent or negligible in this grid. The maximum global-null estimated FDR among these baselines and anchors is 4.40\%; improvements in dense settings typically spend part of substantial baseline conservativeness. No arbitrary-sign validity is asserted for the diagnostic BH comparator.

\paragraph{Experiment 2: independent covariance information.}
Generate $D\sim W_m(N,\Sigma)$ independently for each data set, with $N\in\{200,1000\}$, known unit marginal variances, and either equicorrelation $0.5$ or the heterogeneous block matrix. Use CBSBH with $\epsilon=0.005$, equal $\epsilon_i$, and lower bounds as in \cref{prop:confidenceanchor}. The dense grid has 90\% nonnulls and mean magnitudes $2.5,3,3.5$, with 1,500 repetitions per setting. Four corresponding global nulls use 4,000 repetitions: 16 settings and 34,000 data sets. The oracle comparator is SBH-1 supplied with the true covariance, not an optimum over every known-covariance procedure. Both anchors recompute the confidence-aware size envelope from the realized $D$ and retain the total cap.

For equicorrelation with $N=1000$ and magnitude $3$, CBSBH has power 56.08\%, its mixture anchor 64.45\%, and its $h_{1.5}$ anchor 68.21\%; oracle SBH-1 has 61.16\%. The $h_{1.5}$ gain is 12.13 points (11.78--12.48). This can exceed the \emph{unanchored} oracle baseline because it additionally uses the nonnull-proportion budget; it does not show that less covariance information is intrinsically beneficial. With $N=200$, the corresponding CBSBH, mixture, and $h_{1.5}$ powers are 49.95\%, 59.04\%, and 55.42\%. The ordering of score choices is not uniform.

\begin{figure}[htbp]
\centering\includegraphics[width=.86\linewidth]{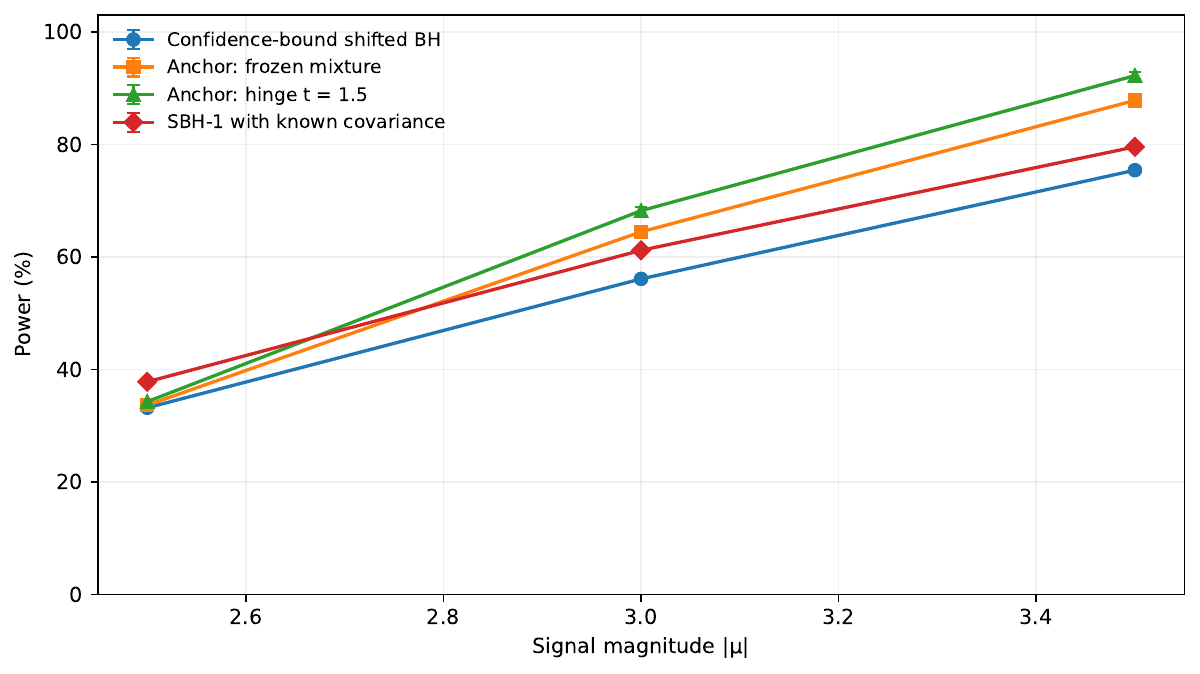}
\caption{Confidence-aware shifted anchors with independent $W_{40}(1000,\Sigma)$ information, equicorrelation $0.5$, and 90\% nonnulls. The oracle curve is unanchored SBH-1 with known covariance. Each setting has 1,500 repetitions; intervals are pointwise 95\% Monte Carlo intervals.}\label{fig:shiftconfidence}
\end{figure}

In heterogeneous blocks with $N=1000$ and magnitude $3.5$, CBSBH improves from 34.03\% to 58.22\% with the mixture, while the single-hinge anchor remains near 34.04\%. All such no-gain settings are retained in the supplied results. The largest global-null estimate for CBSBH and its anchors is 3.10\%; the largest for oracle SBH-1 in this suite is 3.425\%. These are finite Monte Carlo checks, not proofs of the confidence guarantee. The simulations use the original unknown-correlation, known-marginal-scale model; the common-scale $t$ extension above is theoretical and is not claimed to have been simulated here.

\paragraph{Algorithmic scope.}
Every shifted anchor reported here solves its \emph{size-envelope} closure exactly over certified supersets using \cref{prop:threelist}. None substitutes a Monte Carlo estimate of a local expectation. Weight/root calculations and Wishart lower bounds are double-precision, with downward residual margins; exact analytical validity is distinct from interval-certified numerics. All 73,000 realized baselines were retained. We do not claim that these timings or rejection counts apply to the larger identity-specific weighted collection, whose efficient exact optimization remains a separate question.

\section{Anchored IndBH under graphical dependence}
\label{app:anchored-indbh}

 \cite{NguyenFithian2025} introduce
IndBH for FDR control under a known dependency graph
$D$, without specifying the dependence along its edges.
Their key result is that any $D$-adapted procedure $A$
satisfies
\begin{equation}
\label{eq:indbh-coordinatewise}
\mathbb E\left[
\frac{\mathbf 1\{i\in A\}}{|A|\vee1}
\right]\le\frac{\alpha}{m},
\qquad i\in I_0,
\end{equation}
provided the null p-values are superuniform and $D$
is a dependency graph. Recall that $D$-adaptivity
requires self-consistency, coordinatewise monotonicity,
and neighbor-blindness. In particular,
\eqref{eq:indbh-coordinatewise} holds for IndBH and
its iterated improvements, including SUD.

This coordinatewise inequality allows us to improve
IndBH through anchoring, without requiring the
augmented procedure itself to be $D$-adapted.

\paragraph{Anchored IndBH.}
Let $B=\operatorname{IndBH}_{\alpha}(p)$ and
$b=|B|$. For arbitrary valid local e-values $F_S$
for the intersection hypotheses
$H_S=\bigcap_{i\in S}H_i$, define
\begin{equation}
\label{eq:anchored-indbh}
E_S=
\frac{|B\cap S|}{\alpha(b\vee1)}
+\left(1-\frac{|S|}{m}\right)F_S,
\qquad \varnothing\ne S\subseteq[m].
\end{equation}

\begin{proposition}[Anchored IndBH]
\label{prop:anchored-indbh}
Under the dependency-graph assumptions of
Nguyen and Fithian, the quantities
\eqref{eq:anchored-indbh} are valid local e-values.
Any measurable certified superset $\bar B\supseteq B$
satisfying
\[
\frac{|\bar B\cap S|}{|\bar B|\vee1}
\le\alpha E_S
\qquad\text{for all nonempty }S\subseteq[m]
\]
controls FDR at level $\alpha$.
In particular, choosing a largest certified superset
gives
\[
\bar B\supseteq\operatorname{IndBH}_{\alpha}(p),
\qquad
\operatorname{FDR}(\bar B)\le\alpha.
\]
\end{proposition}

\begin{proof}
For every true-null intersection $S$, summing
\eqref{eq:indbh-coordinatewise} gives
\[
\mathbb E[E_S]
\le\frac{|S|}{m}
+\left(1-\frac{|S|}{m}\right)\mathbb E[F_S]
\le1.
\]
Moreover, pointwise,
\[
\frac{|B\cap S|}{b\vee1}\le\alpha E_S,
\]
so $B$ itself is certified. Evaluating the closure
inequality at the actual null set $I_0$ and taking
expectations proves the FDR guarantee.
\end{proof}

\paragraph{A liberal IndBH pilot.}
A particularly simple choice of $F_S$ requires only
one additional run of IndBH. Fix
$\alpha<\beta<1$, let
\[
C=\operatorname{IndBH}_{\beta}(p),
\qquad c=|C|,
\]
and note that $B\subseteq C$. Applying
\eqref{eq:indbh-coordinatewise} at level $\beta$
shows that
\[
X_i^{(\beta)}
=\frac{m\mathbf 1\{i\in C\}}{\beta(c\vee1)}
\]
is an individual e-value under $H_i$. Consequently,
\[
F_S^{(\beta)}
=\frac{m|C\cap S|}
{\beta(c\vee1)|S|}
\]
is a valid local e-value. Our pilot-based anchor is
therefore
\begin{equation}
\label{eq:anchored-indbh-pilot}
E_S^{\mathrm{pilot}}
=
\frac{|B\cap S|}{\alpha(b\vee1)}
+
\left(1-\frac{|S|}{m}\right)
\frac{m|C\cap S|}
{\beta(c\vee1)|S|}.
\end{equation}
This construction is valid for arbitrary dependence
compatible with the known graph, and requires no
additional distributional or covariance assumptions.

\paragraph{An efficient two-count algorithm.}
The full e-closure of
\eqref{eq:anchored-indbh-pilot} can be computed
without enumerating intersections.

Suppose $0<b<c$, and consider a candidate
\[
R=B\cup Q,\qquad
Q\subseteq C\setminus B,\qquad |Q|=q\ge1.
\]
Any certified strict superset of $B$ must be contained
in $C$, since $E_{\{i\}}^{\mathrm{pilot}}=0$
for $i\notin C$. The following result characterizes
the admissible number of additions.

\begin{proposition}[Two-count characterization]
\label{prop:indbh-two-count}
The candidate $R=B\cup Q$ is certified by
\eqref{eq:anchored-indbh-pilot} if and only if
\begin{align}
\alpha(b+q)(c-q)
&\ge \beta c(q+m-c),
\label{eq:indbh-pilot-cond1}\\
\frac{q}{b+q}
&\le
\frac{\alpha b(c-b)}
{\beta c(m-b)}.
\label{eq:indbh-pilot-cond2}
\end{align}
These conditions depend on $Q$ only through $q$.
Both become harder as $q$ increases, so the largest
feasible $q$ can be found by binary search.
\end{proposition}

\begin{proof}
Put $r=b+q$, and partition the indices into
$B$, $Q$, $C\setminus R$, and $C^c$.
For an intersection $S$, write
\[
a=|S\cap B|,\quad x=|S\cap Q|,\quad
y=|S\cap(C\setminus R)|,\quad
z=|S\cap C^c|,
\]
and $n=a+x+y+z$. Its closure constraint is
equivalent to
\[
\frac{aq}{br}-\frac{x}{r}
+
\frac{\alpha}{\beta c}
\frac{(m-n)(a+x+y)}{n}\ge0.
\]
For fixed $a,x,y$, the left side is nonincreasing
in $z$, so its minimum occurs at $z=m-c$.
Writing $u=a+x+y$, the remaining nonlinear term
is proportional to
\[
f(u)=\frac{u(c-u)}{u+m-c}.
\]
This function is concave in $u$; when $c=m$,
use its continuous extension at $u=0$,
corresponding to the excluded empty intersection.

Thus the minimum over the feasible box
\[
0\le a\le b,\quad 0\le x\le q,\quad
0\le y\le c-b-q
\]
occurs at a corner. All corner constraints are
automatic except those with
$(a,x,y)=(0,q,0)$ and $(0,q,c-b-q)$.
These give \eqref{eq:indbh-pilot-cond1} and
\eqref{eq:indbh-pilot-cond2}, respectively.

Finally, the left side minus the right side of
\eqref{eq:indbh-pilot-cond1} is strictly decreasing
in $q$ because $\beta>\alpha$; the left side of
\eqref{eq:indbh-pilot-cond2} is increasing.
Hence feasibility is monotone in $q$.
\end{proof}

After computing IndBH at levels $\alpha$ and $\beta$,
the optimal number of additions is obtained in
$O(\log m)$ time. Any $q$ hypotheses in $C\setminus B$
may be added; choosing those with the smallest
p-values requires at most $O(m\log m)$ additional
time. When $b=0$ or $c=b$, the procedure returns $B$.
Thus the computational cost is dominated by the two
IndBH runs, rather than the e-closure itself.

\paragraph{Strict improvement beyond SUD.}
Consider $m=10$, $\alpha=0.05$, $\beta=0.10$, and a
dependency graph with one edge between vertices
9 and 10, all other vertices being isolated. Let
\[
p_1=\cdots=p_8=10^{-4},
\qquad p_9=p_{10}=0.06.
\]
At level $\alpha$, IndBH rejects the first eight
hypotheses, while at level $\beta$, its union over
independent sets rejects all ten:
\[
B=[8],\qquad C=[10].
\]
Taking $q=2$, conditions
\eqref{eq:indbh-pilot-cond1}--\eqref{eq:indbh-pilot-cond2}
become
\[
4\ge2,\qquad \frac15\le\frac25.
\]
Hence Anchored IndBH rejects all ten hypotheses
while controlling FDR at level $0.05$.

Interestingly, the optimal $D$-adapted procedure
SUD \citep{NguyenFithian2025} also rejects only
the first eight hypotheses on this input:
$D$-adaptivity requires self-consistency, whereas
$p_9=p_{10}=0.06>\alpha$.
Therefore anchoring can strictly outperform SUD
on some inputs, without contradicting SUD's
optimality within the class of $D$-adapted
procedures. 

\paragraph{Discussion.}
The improvement exploits a budget that remains
available after IndBH's coordinatewise FDR
guarantee has been established. Its mechanism
is distinct from the gap-chasing iterations of
Nguyen and Fithian, which improve IndBH while
preserving $D$-adaptivity. Anchoring may step
outside that class, while retaining finite-sample
FDR control under exactly the same dependency
graph assumptions. The pilot construction is
particularly attractive because its closure
requires only two scalar inequalities, even when
the graph itself is large or complicated.

\section{Anchored support line with an adaptive pilot}
\label{app:anchored-sl}

This section illustrates anchoring for boundary error rather than
ordinary FDR. For a rejection set $R$, let $J(R)$ be the index of its
largest p-value, with a fixed tie rule and $J(\varnothing)=0$. The
\emph{boundary false discovery rate} is
\[
\operatorname{bFDR}(R)=\Pr\{J(R)\in I_0\},\qquad 0\notin I_0.
\]
This controls the probability that the boundary discovery is null,
not the expected FDP or the local false discovery rate of every
selected hypothesis. Throughout, p-values are independent and
$p_i\sim\operatorname{Unif}(0,1)$ under $H_i$; nonnull distributions
are otherwise unrestricted.

\paragraph{The support-line baseline.}
Write $p_{(0)}=0$ and, for $q\in(0,1)$, define
\[
K_q=\max\operatorname*{arg\,max}_{0\le k\le m}
       \left\{\frac{qk}{m}-p_{(k)}\right\},
\qquad
B_q=\{(1),\ldots,(K_q)\},\qquad J_q=J(B_q).
\]
Thus $B_q=\operatorname{SL}_q(p)$, and $B_q=\varnothing$ if $K_q=0$.
The proof of Theorem~5.1 of \citet{XiangSoloffFithian2025} gives
\begin{equation}
\label{eq:asl-conditional-boundary}
\Pr_{H_i}(J_q=i\mid p_{-i})=\frac qm.
\end{equation}
Consequently, $\operatorname{bFDR}(B_q)=q|I_0|/m$.
Two deterministic properties will also be useful: $B_q$ increases
with $q$, and every nonempty boundary satisfies
$p_{J_q}\le qK_q/m\le q$.
We retain $B=B_\alpha$ while selecting a larger pilot level using
only upper-tail information.

\paragraph{An explicit adaptive pilot.}
Assume $m\ge2$ and fix, before observing the p-values,
\[
0<\alpha<\beta_{\max}\le t<1,
\qquad 1\le h<1/t,
\qquad \ell=\frac{1-th}{1-t}\in(0,1].
\]
Use the normalized two-level helper
\begin{equation}
\label{eq:asl-helper}
g(u)=h\mathbf1\{u\le t\}+\ell\mathbf1\{u>t\},
\qquad \int_0^1g(u)\,du=1,
\end{equation}
and let $V=\sum_{i=1}^m\mathbf1\{p_i>t\}$.
For $v=0,\ldots,m-2$, define the certificate-derived ceiling
\begin{equation}
\label{eq:asl-pilot-ceiling}
\beta_{\mathrm{cap}}(v)
=\min_{1\le n\le m-1}
\frac{\alpha(m-n)}{n}
\ell^{\min\{v,n-1\}}h^{(n-1-v)_+}.
\end{equation}
Set
\begin{equation}
\label{eq:asl-adaptive-pilot}
\psi(v)=
\begin{cases}
\max\{\alpha,\min\{\beta_{\max},\beta_{\mathrm{cap}}(v)\}\},
   &0\le v\le m-2,\\
\alpha,&v\in\{m-1,m\},
\end{cases}
\qquad \widehat\beta=\psi(V).
\end{equation}
The resulting rule is particularly simple:
\begin{equation}
\label{eq:asl-adaptive-output}
\bar B=
\begin{cases}
\varnothing,& B_\alpha=\varnothing,\\
B_{\widehat\beta},& B_\alpha\ne\varnothing.
\end{cases}
\end{equation}
The helper parameters and $\beta_{\max}$ are prespecified, but the
actual pilot level is data-adaptive. The minimum in
\eqref{eq:asl-pilot-ceiling} comes from the full closure constraints,
not from optimizing power on the observed data. The restriction
$\widehat\beta\le t$ is essential to the following validity argument.

\paragraph{Adaptive pilot evidence and anchoring.}
For each nonempty intersection $S$, with $n=|S|$, define
\begin{align}
F_S^{\mathrm{SL}}
&=\frac{m}{n}
  \sum_{j\in S}\frac{\mathbf1\{J_{\widehat\beta}=j\}}
                         {\widehat\beta}
  \prod_{i\in S\setminus\{j\}}g(p_i),
\label{eq:asl-enriched-pilot}\\
E_S^{\mathrm{SL}}
&=\frac{\mathbf1\{J_\alpha\in S\}}{\alpha}
  +\left(1-\frac nm\right)F_S^{\mathrm{SL}}.
\label{eq:asl-anchor}
\end{align}
The empty product is one, and the second term is omitted for
$S=[m]$. The pilot boundary contributes calibrated evidence,
enriched by the helper scores of the other coordinates in $S$.
For boundary loss, e-closure certifies a prefix $R$ precisely when
\begin{equation}
\label{eq:asl-closure}
\mathbf1\{J(R)\in S\}\le\alpha E_S^{\mathrm{SL}}
\qquad\text{for every nonempty }S\subseteq[m].
\end{equation}

\begin{proposition}[Adaptively piloted Anchored-SL]
\label{prop:asl}
Under the stated assumptions, both $F_S^{\mathrm{SL}}$ and
$E_S^{\mathrm{SL}}$ are valid local e-values for $H_S$.
The rule \eqref{eq:asl-adaptive-output} is the largest p-value prefix
containing $B_\alpha$ that is certified by \eqref{eq:asl-closure}.
In particular,
\[
\bar B\supseteq B_\alpha\quad\text{for every data vector},
\qquad \operatorname{bFDR}(\bar B)\le\alpha.
\]
The rule is computable in $O(m\log m)$ time and $O(m)$ memory.
\end{proposition}

\begin{proof}
First justify the adaptive choice. For each $j$, put
\[
\beta_j=\psi\!\left(\sum_{i\ne j}\mathbf1\{p_i>t\}\right).
\]
This is the pilot level obtained after replacing $p_j$ by zero,
and it is measurable with respect to $p_{-j}$. If $p_j\le t$,
then $\beta_j=\widehat\beta$. If $p_j>t$, neither pilot can have
boundary $j$, since both levels are at most $t$. Therefore,
pathwise,
\begin{equation}
\label{eq:asl-leave-one-out}
\frac{m\mathbf1\{J_{\widehat\beta}=j\}}{\widehat\beta}
=\frac{m\mathbf1\{J_{\beta_j}=j\}}{\beta_j}.
\end{equation}
Conditioning on $p_{-j}$ and applying
\eqref{eq:asl-conditional-boundary} shows that this score has
conditional null expectation one. Hence, under $H_S$,
\[
\mathbb E F_S^{\mathrm{SL}}
=\frac1n\sum_{j\in S}
  \mathbb E\prod_{i\in S\setminus\{j\}}g(p_i)=1,
\]
using independence of the null p-values and $\int g=1$.
It follows that
$\mathbb E E_S^{\mathrm{SL}}=n/m+(1-n/m)=1$.
The baseline is certified pointwise by the first term of
\eqref{eq:asl-anchor}.

It remains to check that the chosen pilot passes. There is nothing
to prove if $B_\alpha=\varnothing$ or $B_{\widehat\beta}=B_\alpha$.
Otherwise, both distinct boundary p-values are at most $t$, so
their helpers equal $h$. Intersections containing $J_\alpha$ pass
automatically. Among intersections of size $n$ containing
$J_{\widehat\beta}$ but excluding $J_\alpha$, the weakest helper
product uses the $n-1$ smallest helpers among the remaining
$m-2$ coordinates: $V$ copies of $\ell$ followed by copies of $h$.
Thus the minimum local evidence is
\begin{equation}
\label{eq:asl-fast}
\min_{\substack{|S|=n,\ J_{\widehat\beta}\in S\\J_\alpha\notin S}}
\alpha E_S^{\mathrm{SL}}
=
\frac{\alpha(m-n)}{\widehat\beta n}
\ell^{\min\{V,n-1\}}h^{(n-1-V)_+}.
\end{equation}
When the boundary changes, $\widehat\beta>\alpha$, and
\eqref{eq:asl-adaptive-pilot} ensures
$\widehat\beta\le\beta_{\mathrm{cap}}(V)$.
Every value in \eqref{eq:asl-fast} is therefore at least one.
All other candidate boundaries have zero singleton evidence,
so no larger prefix is certified. If the baseline is empty,
$E_{[m]}^{\mathrm{SL}}=0$ excludes every nonempty output.
Evaluating \eqref{eq:asl-closure} at $I_0$ now gives the bFDR
bound; the case $I_0=\varnothing$ is immediate.
Finally, one sort suffices for both SL runs. Counting $V$,
evaluating the $m-1$ terms of \eqref{eq:asl-pilot-ceiling}, and
scanning the two support-line objectives are all linear-time
operations. Products may be evaluated on the logarithmic scale.
\end{proof}

The leave-one-out identity \eqref{eq:asl-leave-one-out}, rather
than independence of the selected level and the data, justifies
adaptation. Arbitrary data-dependent choices of $t$, $h$, or the
pilot level are not covered by this argument. Also, the nonempty
rejection event is unchanged: the method augments existing SL
discoveries, rather than detecting configurations missed entirely
by the baseline.

For a strict example, take $m=4$, $\alpha=0.1$,
$\beta_{\max}=0.2$, $t=1/4$, $h=3$, and $\ell=1/3$.
At $p=(0.001,0.002,0.003,0.04)$, we have $V=0$,
$\beta_{\mathrm{cap}}(0)=0.3$, and $\widehat\beta=0.2$.
The baseline rejects three hypotheses, whereas
\eqref{eq:asl-adaptive-output} rejects all four.

\paragraph{An illustrative power comparison.}
We use $m=100$, $\alpha=0.10$, $\beta_{\max}=0.40$, and
\begin{equation}
\label{eq:asl-simulation-score}
t=1/2,\qquad h=1.1,\qquad \ell=0.9,
\qquad g(u)=0.9+0.2\,\mathbf1\{u\le1/2\}.
\end{equation}
Only $\widehat\beta$ adapts to the observed upper-tail count;
the displayed parameters are not fitted to the evaluation data.
Generate independent $Z_i\sim N(\mu_i,1)$ and use
$p_i=1-\Phi(Z_i)$. Table~\ref{tab:asl-power} has 90 nonnulls
with common mean $\mu$ and ten nulls with mean zero, with
10,000 repetitions per setting. Power is the mean fraction of
nonnulls rejected. Gaussianity is only the simulation model,
not an assumption of the method.

\begin{table}[t]
\centering
\small
\caption{SL and Anchored-SL with the adaptive pilot
\eqref{eq:asl-adaptive-pilot}, at nominal boundary-FDR level
$\alpha=0.10$. There are $m=100$ tests, 90 nonnulls, and 10,000
repetitions per row. The pilot ceiling is $\beta_{\max}=0.40$,
and the helper is \eqref{eq:asl-simulation-score}.
Power and empirical bFDR are percentages; gains are percentage
points computed before rounding.}
\label{tab:asl-power}
\begin{tabular}{rrrrrr}
\toprule
& \multicolumn{2}{c}{Power (\%)} & &
  \multicolumn{2}{c}{bFDR (\%)}\\
\cmidrule(lr){2-3}\cmidrule(lr){5-6}
$\mu$ & SL & Anchored-SL & Gain (pp) & SL & Anchored-SL\\
\midrule
1.5 & 22.85 & 42.30 & 19.45 & 0.90 & 2.64\\
2.0 & 44.71 & 70.34 & 25.64 & 0.99 & 3.88\\
2.5 & 63.83 & 82.17 & 18.33 & 0.94 & 3.79\\
3.0 & 77.70 & 89.49 & 11.79 & 0.95 & 3.98\\
3.5 & 87.20 & 94.13 &  6.93 & 1.06 & 3.80\\
\bottomrule
\end{tabular}
\end{table}

At $\mu=2$, the paired power gain is 25.64 percentage points,
with pointwise 95\% Monte Carlo interval $[25.44,25.83]$.
This spends some of the baseline's conservativeness: estimated
bFDR rises from $0.99\%$ to $3.88\%$, below the common nominal
$10\%$ target. These are not matched-realized-error comparisons.
In the companion grid with 10\% nonnulls there are no additions;
with 50\% nonnulls, power gains are below 0.02 percentage points.
An independent global-null check with 50,000 repetitions gives
bFDR $10.064\%$ for both procedures, compatible with $10\%$
within Monte Carlo uncertainty. The full study comprises
200,000 data sets across 16 settings. These results demonstrate
dense-signal gains, not a general sparse-power improvement or
superiority to the later procedures discussed next.

\paragraph{Comparison with adaptive support-line procedures.}
\citet{MostowXiang2026} adjust the support-line slope through
null-proportion estimation. Our pilot level is also data-adaptive,
but it is chosen from the explicit boundary-closure ceiling
\eqref{eq:asl-pilot-ceiling}, rather than by substituting a null
proportion estimate into the testing level. In particular, their
plus-one Storey estimator can exceed one and reduce the original
SL slope, while their calibrated two-stage rule starts at
$\alpha/(1+\alpha)$; neither construction necessarily retains
$B_\alpha$. Here $\widehat\beta\ge\alpha$ and the anchored
certificate guarantees retention. Conversely, their procedures
can choose different slopes or exceed our prespecified
$\beta_{\max}$. 
Table~\ref{tab:asl-power} compares the anchor with original SL,
not with those adaptive competitors.

\paragraph{Comparison with Domino.}
\citet{ZhangZhangZouRen2026} develop a general boundary-closure
framework whose dependence assumptions are determined by the
local tests. For the single-boundary case, a proposed boundary is
certified when every intersection containing it is rejected.
Taking $\phi_S=\mathbf1\{E_S^{\mathrm{SL}}\ge1/\alpha\}$ gives
exactly \eqref{eq:asl-closure}. Anchored-SL is therefore a
particular local-test construction within boundary closure,
not a competing closure principle. Its specific ingredients are
the upper-tail-adaptive pilot, the product-enriched evidence
justified by \eqref{eq:asl-leave-one-out}, and the explicit
sorting-time rule retaining the SL baseline. 

\section{Anchored conformal multiple testing}
\label{app:anchored-conformal}

Conformal testing supplies distribution-free p-values for detecting
outliers relative to a reference sample. Sharing the calibration sample
makes these p-values dependent, but their dependence can still yield
the coordinatewise bounds needed for anchoring. We first give a
sorting-time improvement of the split-conformal BH procedure of
\citet{BatesEtAl2023}, using only one additional BH scan. We then extend
the construction to ECOT, whose conditional calibration allows more
flexible data reuse. These are deductions from the respective contribution
bounds, not procedures proposed in those papers.

\paragraph{Split-conformal p-values and the baseline.}
Let $X_1^0,\ldots,X_N^0$ be calibration observations from $P_0$, and let
$X_1,\ldots,X_m$ be test observations with hypotheses $H_j:X_j\sim P_0$.
A score $s$ is fitted on independent training data; orient it so that larger
values indicate stronger evidence against the null. Conditional on the
training data, assume that the calibration and test observations are
mutually independent, that the calibration observations and true-null
test observations have common law $P_0$, and that $s(X)$ has a continuous
null distribution. Nonnull distributions are unrestricted. Define
\begin{equation}
\label{eq:aconf-pvalues}
p_j=\frac{1+\sum_{i=1}^N\mathbf1\{s(X_i^0)\ge s(X_j)\}}{N+1},
\qquad j\in[m].
\end{equation}
Theorem~2 of \citet{BatesEtAl2023} gives superuniformity and PRDS on the
null coordinates, conditionally on the training data. Continuity concerns
the scores: the p-values themselves lie on a finite grid and may tie.

For $\gamma\in(0,1)$, write
\[
b_\gamma=\max\bigl(\{k:p_{(k)}\le\gamma k/m\}\cup\{0\}\bigr),
\qquad
B_\gamma=\{j:p_j\le\gamma b_\gamma/m\},
\]
with $B_\gamma=\varnothing$ when $b_\gamma=0$.
The coordinatewise BH bound under PRDS \citep{BR2008} is
\begin{equation}
\label{eq:aconf-coordinate}
\mathbb E_{H_j}\left[
\frac{\mathbf1\{j\in B_\gamma\}}{|B_\gamma|\vee1}
\right]\le\frac\gamma m.
\end{equation}
The expectation averages over the calibration and test observations.
We do not assert FDR control conditional on the realized calibration
sample; conditioning on the independent training sample is permitted.

\paragraph{One liberal run supplies the auxiliary e-values.}
Fix $0<\alpha<\beta<1$ before inspecting the testing data, for example
$\beta=2\alpha$ when $\alpha<1/2$. Set
\[
B=B_\alpha,\quad C=B_\beta,\quad b=|B|,\quad c=|C|.
\]
Then $B\subseteq C$. By \eqref{eq:aconf-coordinate},
\[
x_j^{(\beta)}=
\frac{m\mathbf1\{j\in C\}}{\beta(c\vee1)}
\]
is an individual e-value. Hence, for every nonempty $S\subseteq[m]$,
\begin{align}
F_S^{\mathrm{conf}}
&=\frac{1}{|S|}\sum_{j\in S}x_j^{(\beta)}
=\frac{m|C\cap S|}{\beta(c\vee1)|S|},
\label{eq:aconf-pilot}\\
E_S^{\mathrm{conf}}
&=\frac{|B\cap S|}{\alpha(b\vee1)}
  +\left(1-\frac{|S|}{m}\right)F_S^{\mathrm{conf}}.
\label{eq:aconf-anchor}
\end{align}
Define $\operatorname{FDP}_S(R)=|R\cap S|/(|R|\vee1)$, and choose a
largest-cardinality certified superset $\bar B\supseteq B$ satisfying
\begin{equation}
\label{eq:aconf-closure}
\operatorname{FDP}_S(\bar B)\le\alpha E_S^{\mathrm{conf}}
\qquad(\varnothing\ne S\subseteq[m]).
\end{equation}
We call this the pilot-based \emph{Anchored conformal BH} procedure.

\begin{proposition}[Anchored conformal BH]
\label{prop:aconf-guarantee}
Under the stated split-conformal assumptions, the collection
\eqref{eq:aconf-anchor} consists of valid local e-values. Every measurable
certified superset satisfies
\[
\bar B\supseteq B_\alpha\quad\text{pointwise},
\qquad
\operatorname{FDR}(\bar B)\le\alpha.
\]
No independence between the baseline and the pilot is required.
\end{proposition}

\begin{proof}
Under $H_S$, summing \eqref{eq:aconf-coordinate} over $j\in S$ gives
\[
\mathbb E E_S^{\mathrm{conf}}
\le\frac{|S|}{m}
 +\left(1-\frac{|S|}{m}\right)\mathbb E F_S^{\mathrm{conf}}
\le1.
\]
The baseline is certified because its FDP certificate is the first
term of \eqref{eq:aconf-anchor}. For the actual nonempty null set $I_0$,
\eqref{eq:aconf-closure} gives
$\operatorname{FDR}(\bar B)\le\alpha\mathbb E E_{I_0}^{\mathrm{conf}}
\le\alpha$; if $I_0=\varnothing$, the FDP is zero.
\end{proof}

\paragraph{Exact computation from two counts.}
Although \eqref{eq:aconf-closure} ranges over all intersections, its full
superset search uses only $b$ and $c$. Return $B$ immediately if $b=0$
or $c=b$. Otherwise every addition must lie in $C\setminus B$, since a
singleton outside $B\cup C=C$ has zero evidence. For
$R=B\cup Q$, where $Q\subseteq C\setminus B$ and $|Q|=q\ge1$, put $r=b+q$.

\begin{proposition}[Two-count certificate]
\label{prop:aconf-counts}
For $0<b<c$, the candidate $R$ is certified if and only if
\begin{align}
\alpha(b+q)(c-q)&\ge\beta c(q+m-c),
\label{eq:aconf-count1}\\
\frac{q}{b+q}&\le
\frac{\alpha b(c-b)}{\beta c(m-b)}.
\label{eq:aconf-count2}
\end{align}
Certification depends on $Q$ only through $q$. Both conditions become
harder as $q$ increases.
\end{proposition}

\begin{proof}
For an intersection $S$, let $a,x,y,z$ count its elements in
$B,Q,C\setminus R,C^c$, respectively, and put $n=a+x+y+z$.
The closure inequality is
\[
\frac{aq}{br}-\frac{x}{r}
+\frac{\alpha}{\beta c}\frac{(m-n)(a+x+y)}{n}\ge0.
\]
For fixed $a,x,y$, the left side is minimized at $z=m-c$.
Writing $u=a+x+y$, its nonlinear part is proportional to the concave
function $u(c-u)/(u+m-c)$. Minimizing over the remaining count box
therefore reduces to its corners. The only nonautomatic corners are
$(a,x,y)=(0,q,0)$ and $(0,q,c-b-q)$, giving
\eqref{eq:aconf-count1} and \eqref{eq:aconf-count2}.
When $c=m$, omit the empty intersection and use the continuous extension
at $u=0$ for this corner argument. Finally,
$\alpha(b+q)(c-q)-\beta c(q+m-c)$ decreases in $q$, whereas
$q/(b+q)$ increases.
\end{proof}

Find the largest feasible $q\in\{1,\ldots,c-b\}$ by binary search,
with $q=0$ always permitted, and add any $q$ elements of $C\setminus B$.
A natural choice is decreasing $s(X_j)$, breaking score ties by index;
this also resolves ties between equal conformal p-values.
After the p-values are sorted, the two BH counts require linear scans
and the augmentation requires $O(\log m)$ arithmetic. Thus testing and
anchoring together cost $O(m\log m)$ time and $O(m)$ memory after the
conformal p-values are formed. Forming all ranks costs
$O(N\log N+m\log(N+1))$ after score evaluation.
There is no refitting, permutation integration, or extra sample split
for the pilot anchor.

For a strict conformal-grid example, take $N=99$, $m=10$,
$\alpha=0.05$, $\beta=0.10$, and
\[
p_1=\cdots=p_8=0.01,\qquad p_9=p_{10}=0.06.
\]
These are possible ranks in \eqref{eq:aconf-pvalues}. The baseline
selects eight and the pilot selects ten. With $q=2$, the two conditions
are $4\ge2$ and $1/5\le2/5$, so the anchor selects all ten.

\paragraph{Simulation: a shared calibration sample.}
We use $m=100$, $\alpha=0.10$, $\beta=0.20$, and $N\in\{199,999\}$.
Calibration and null test scores are independent $N(0,1)$;
nonnull scores are independent $N(\mu,1)$. The score is the observation
itself, with no training or model selection. Every test p-value uses
the same calibration sample through \eqref{eq:aconf-pvalues}, so the
p-values are dependent even in this simple experiment.
The grid has $10$, $50$, or $90$ nonnulls and $\mu\in\{2,3\}$, with
5,000 repetitions per setting, plus two 10,000-repetition global nulls:
80,000 data sets across 14 settings. Table~\ref{tab:aconf-power}
reports the dense settings. Comparisons share the simulated data and
use a common nominal FDR level, not matched realized FDR.

\begin{table}[t]
\centering
\small
\caption{Conformal BH and its pilot anchor: $m=100$, 90 nonnulls,
$\alpha=0.10$, and $\beta=0.20$. Each row has 5,000 repetitions.
Power and empirical FDR are percentages; gains are percentage points.
This is a fixed-score split-conformal experiment, not an evaluation
of learned or model-selected ECOT scores.}
\label{tab:aconf-power}
\begin{tabular}{rrrrrrr}
\toprule
& & \multicolumn{2}{c}{Power (\%)} && \multicolumn{2}{c}{FDR (\%)}\\
\cmidrule(lr){3-4}\cmidrule(lr){6-7}
$N$ & $\mu$ & Conformal BH & Anchored & Gain & Conformal BH & Anchored\\
\midrule
199 & 2 & 65.78 & 72.12 & 6.34 & 0.97 & 1.41\\
199 & 3 & 94.44 & 98.03 & 3.59 & 1.00 & 1.97\\
999 & 2 & 67.69 & 74.91 & 7.22 & 0.98 & 1.47\\
999 & 3 & 94.88 & 98.13 & 3.25 & 1.00 & 2.00\\
\bottomrule
\end{tabular}
\end{table}

For $N=199$ and $\mu=2$, the paired power gain is 6.34 percentage
points, with pointwise 95\% Monte Carlo interval $[6.19,6.48]$.
There are no additions in any of the 40,000 evaluations with 10\%
or 50\% nonnulls. The global-null estimated FDRs are 0.35\% for
$N=199$ and 9.87\% for $N=999$, identical for baseline and anchor.
The finite-grid calibration can be conservative, and the full-intersection
certificate forces an empty output whenever the baseline is empty.
The method therefore does not remove calibration-resolution or sparse-signal
limitations. These are our illustrative simulations, not results reported
by \citet{BatesEtAl2023}. The count shortcut was checked against exhaustive
closure in 210 small cases using exact rational arithmetic; all 1,023
nonempty intersections were also checked in the strict example above.

\subsection{Anchored ECOT}
\label{app:anchored-ecot}

ECOT (Enhanced COnformal Testing) uses individualized scores and
permutation calibration to reuse labeled and unlabeled data
\citep{HuoWuZouRen2026}. Its p-values need not satisfy the split-conformal
PRDS argument. We instead use the coordinatewise conditional-calibration
bound, including the final pruning step. The result here concerns the
unrescaled ECOT procedure, not versions that already estimate the null
proportion and rescale the testing level.

\paragraph{A self-contained conditional-calibration formulation.}
Index test hypotheses by $[m]$, and let $\mathcal C$ be disjoint
calibration indices. Assume that calibration and null test observations
are exchangeable conditional on the remaining data. For each $j$, let
$\Omega_j$ permute $\mathcal C\cup\{j\}$ while fixing all other indices.
For a fitted score $s^{(j)}$, let $s_\sigma^{(j)}$ be the score refitted
after permutation $\sigma$, and define
\[
p_j=\frac1{|\Omega_j|}\sum_{\sigma\in\Omega_j}
\mathbf1\{s^{(j)}(X_j)\le s_\sigma^{(j)}(X_{\sigma(j)})\}.
\]
Set $\tilde p_j^{(j)}=0$ and, for $k\ne j$, take
\[
\tilde s^{(j)}(X_k)=
\operatorname{Median}_{\sigma\in\Omega_j}s_\sigma^{(j)}(X_k),
\qquad
\tilde p_k^{(j)}=\frac1{|\Omega_j|}\sum_{\sigma\in\Omega_j}
\mathbf1\{\tilde s^{(j)}(X_k)\le s_\sigma^{(j)}(X_{\sigma(j)})\}.
\]
Let $\Psi_j$ retain the unordered observations in
$\mathcal C\cup\{j\}$ and all remaining data. ECOT's Theorem~1 establishes,
under $H_j$,
\begin{equation}
\label{eq:aecot-invariance}
\Pr(p_j\le u\mid\Psi_j,H_j)\le u,
\qquad
\tilde p^{(j)}\text{ is }\Psi_j\text{-measurable}.
\end{equation}
An equivalent symmetry-based implementation may replace the full
permutation calculation when the score construction permits it.

At a testing level $\gamma$, put
\[
r_j(\gamma)=|\operatorname{BH}_\gamma(\tilde p^{(j)})|\ge1,
\qquad
I_\gamma=\{j:p_j\le\gamma r_j(\gamma)/m\}.
\]
If $I_\gamma=\varnothing$, return $A_\gamma=\varnothing$.
If $|I_\gamma|\ge r_j(\gamma)$ for every $j\in I_\gamma$, return
$A_\gamma=I_\gamma$. Otherwise draw independent
$U_j\sim\operatorname{Unif}(0,1)$, independent of the data, and let
$A_\gamma$ be BH at level one applied, within $I_\gamma$, to
$U_jr_j(\gamma)/|I_\gamma|$. This is ECOT's completed
conditional-calibration rule \citep{FithianLei2022,HuoWuZouRen2026}.

\paragraph{The contribution bound and the anchored guarantee.}
The pruning argument gives, conditional on all observed data,
\[
\mathbb E_U\left[
\frac{\mathbf1\{j\in A_\gamma\}}{|A_\gamma|\vee1}
\,\middle|\,\text{data}\right]
\le \frac{\mathbf1\{p_j\le\gamma r_j(\gamma)/m\}}{r_j(\gamma)}.
\]
The same inequality is pointwise when pruning is unnecessary.
Conditioning next on $\Psi_j$ and applying
\eqref{eq:aecot-invariance} yields
\begin{equation}
\label{eq:aecot-coordinate}
\mathbb E_{H_j}\left[
\frac{\mathbf1\{j\in A_\gamma\}}{|A_\gamma|\vee1}
\,\middle|\,\Psi_j\right]\le\frac\gamma m.
\end{equation}
This is the coordinatewise content of the proof of ECOT's Theorem~1,
especially its equation~(E.16), before summing over null hypotheses.

For $A=A_\alpha$, any valid local e-values $F_S$ therefore give
\begin{equation}
\label{eq:aecot-anchor}
E_S^{\mathrm{ECOT}}=
\frac{|A\cap S|}{\alpha(|A|\vee1)}
+\left(1-\frac{|S|}{m}\right)F_S.
\end{equation}
Summing \eqref{eq:aecot-coordinate} over $j\in S$ proves
$\mathbb E_{H_S}E_S^{\mathrm{ECOT}}\le1$.
Choosing any measurable superset $\bar A\supseteq A$ with
$\operatorname{FDP}_S(\bar A)\le\alpha E_S^{\mathrm{ECOT}}$ for all $S$
then gives
\[
\bar A\supseteq A_\alpha
\quad\text{for every realized data set and pruning outcome},
\qquad \operatorname{FDR}(\bar A)\le\alpha.
\]
The retained baseline is the final ECOT output, not its initial set
$I_\alpha$, which need not itself control FDR.

A fixed liberal run $C=A_\beta$ supplies the immediate choice
\begin{equation}
\label{eq:aecot-pilot}
x_j^{(\beta)}=
\frac{m\mathbf1\{j\in C\}}{\beta(|C|\vee1)},
\qquad
F_S=\frac1{|S|}\sum_{j\in S}x_j^{(\beta)}.
\end{equation}
Their validity follows from \eqref{eq:aecot-coordinate} at level $\beta$.
The two runs may share all observations; no independence between them
is needed. When the score-construction rule is fixed across testing
levels, its permutation calculations can be reused for the second run.
If fitting or score selection depends on the level, it must instead be
rerun in its prescribed permutation-valid manner.

\paragraph{Cheap closure even without nested ECOT outputs.}
In general $A\subseteq C$ need not hold. Nevertheless, the pilot
collection \eqref{eq:aecot-pilot} still has a count-based exact closure.
Write $b=|A|$, $c=|C|$, and $d=|A\cap C|$.
If $b=0$, $c=0$, or $C\subseteq A$, return $A$.
Otherwise every certified addition belongs to $C\setminus A$.
Consider $R=A\cup Q$ with $Q\subseteq C\setminus A$,
$|Q|=q$, and $r=b+q$.
For each $n=1,\ldots,m$, set
\[
\theta_n=\frac{\alpha(m-n)}{\beta cn},
\qquad
a_i(n,R)=\frac{\mathbf1\{i\in A\}}b
          -\frac{\mathbf1\{i\in R\}}r
          +\theta_n\mathbf1\{i\in C\}.
\]
All size-$n$ closure inequalities hold exactly when the sum of the
$n$ smallest values of $a_i(n,R)$ is nonnegative.
There are only five distinct values, with their multiplicities:
\[
\begin{array}{c|c}
\text{value}&\text{multiplicity}\\ \hline
1/b-1/r+\theta_n & d\\
1/b-1/r & b-d\\
\theta_n-1/r & q\\
\theta_n & c-d-q\\
0 & m-b-c+d.
\end{array}
\]
Sorting these five numbers and filling $n$ places takes constant
arithmetic per $n$. Checking every $q\in\{1,\ldots,c-d\}$ therefore
finds a largest certified superset in $O(m^2)$ additional arithmetic
and $O(m)$ memory, after the two ECOT outputs are available.
This is an exact reduction for the pilot collection; no nesting
or common p-value ordering is assumed. Select any $q$ pilot-only
hypotheses, for example those with smallest conformal p-values.
If $A\subseteq C$, the sharper two-count characterization
\eqref{eq:aconf-count1}--\eqref{eq:aconf-count2} applies instead.

\paragraph{A pilot without a second pruning step.}
The conditional bound also supplies
\[
\tilde x_j^{(\beta)}=
\frac{\mathbf1\{p_j\le t_j(\beta)\}}{t_j(\beta)},
\qquad t_j(\beta)=\frac{\beta r_j(\beta)}m,
\qquad
\mathbb E_{H_j}[\tilde x_j^{(\beta)}\mid\Psi_j]\le1.
\]
Their averages are another valid choice of $F_S$ in
\eqref{eq:aecot-anchor}. This avoids a second randomized pruning step,
although its unequal scores no longer have the five-count form above.
Uncalibrated products of conformal scores, naive Monte Carlo replacements
of permutation normalizers, and data-selected pilot levels are not
justified by these arguments.

When all test points use one common score invariant under permutations
of the pooled calibration and null observations, ECOT reduces to ordinary
BH on conformal p-values (\citealp{HuoWuZouRen2026}, Proposition~2.2),
so the first part
of this appendix and its sorting-time algorithm apply directly.
Table~\ref{tab:aconf-power} also evaluates this common-score ECOT
special case. It does not evaluate individualized learned scores,
level-dependent model selection, or random pruning. In those cases the
bounds above concern the cost of augmentation after ECOT's own
calibration; they do not make arbitrary full-permutation fitting cheap.

\section{Timing and reproducibility}\label{app:reproducibility}
\Cref{tab:timing} reports median elapsed times over seven repeats on one fixed input per dimension. Inputs have independent Gaussian noise, 90\% positive nonnull means equal to 3, and one-sided p-values. Times include sorting and score construction, but exclude first-time compilation. The examples often terminate early because of screening bounds; they are not empirical worst-case complexity estimates.
\begin{table}[ht]
\centering\small
\caption{Elapsed time in milliseconds. The general finite-profile bound is quadratic; the pilot bound is $O(m\log m)$.}\label{tab:timing}
\begin{tabular}{@{}rrrrr@{}}
\toprule
$m$ & BH & Pilot-BH anchor & Fixed-hinge anchor & Adaptive staircase\\
\midrule
100 &0.002&0.003&0.062&0.013\\
1,000&0.019&0.020&0.322&0.027\\
5,000&0.111&0.118&1.111&0.153\\
10,000&0.261&0.263&2.682&0.500\\
\bottomrule
\end{tabular}
\end{table}

\FloatBarrier
\subsection*{Ordered algorithms: application times and reproduction}\label{app:orderedtiming}
\Cref{tab:ordered-timing} shows timings from the ordered study, including score construction and conditional calibration but excluding first-time compilation. Each entry is the median of five calls on one fixed input per size. These are not worst-case inputs. The full recursive prefix is not asserted to have the quadratic cost of the raw-baseline anchor.
\begin{table}[htbp]
\centering\small
\caption{Median ordered-method elapsed times (milliseconds). The final structural scans are cheaper than the conditional calculations. A dash denotes an unmeasured dimension.}\label{tab:ordered-timing}
\begin{tabular}{@{}lrrrr@{}}
\toprule
Method & $m=100$ & $m=300$ & $m=1000$ & $m=5000$\\
\midrule
Adaptive SeqStep, same cutoff &0.160&0.943&10.21&217.00\\
Adaptive SeqStep, rectangle &0.175&1.107&10.52&221.59\\
SeqStep+, calibrated and recycled prefix &18.79&---&---&---\\
\bottomrule
\end{tabular}
\end{table}

\FloatBarrier

\end{document}